\documentclass[11pt]{article}
\usepackage{amsfonts}
\usepackage{amssymb}
\usepackage{amstext}
\usepackage{amsmath}
\usepackage{xspace}
\usepackage{ntheorem}
\usepackage{graphicx}
\usepackage{url}
\usepackage{graphics}
\usepackage{colordvi}
\usepackage{colordvi}
\usepackage{subfigure}

\advance \oddsidemargin by -0.8in
\newcommand{\myparskip}{3pt}
\parskip \myparskip

\newcommand{\approxfactor}{2^{O(\sqrt{\log n}\cdot \log\log n)}}

\newcommand{\ceil}[1]{\ensuremath{\left\lceil#1\right\rceil}}
\newcommand{\floor}[1]{\ensuremath{\left\lfloor#1\right\rfloor}}

\newcommand{\event}{{\cal{E}}}
\newcommand{\HSC}{{\sf HSC}\xspace}
\newcommand{\NDP}{{\sf NDP}\xspace}
\newcommand{\EDP}{{\sf EDP}\xspace}

\newcommand{\rNDPgrid}{{\sf Restricted NDP-Grid}\xspace}
\newcommand{\NDPgrid}{{\sf NDP-Grid}\xspace}
\newcommand{\EDPwall}{{\sf EDP-Wall}\xspace}
\newcommand{\NDPwall}{{\sf NDP-Wall}\xspace}

\newcommand{\NDPplanar}{{\sf NDP-Planar}\xspace}

\newcommand{\RNDP}{{\sf Restricted NDP-Grid}\xspace}
\newcommand{\restrictedNDP}{{\sf Restricted NDP-Grid}\xspace}
\newcommand{\row}{\operatorname{row}}
\newcommand{\col}{\operatorname{col}}

\newcommand{\NP}{\mbox{\sf NP}}

\newcommand{\BPTIME}{\mbox{\sf BPTIME}}
\newcommand{\ZPTIME}{\mbox{\sf ZPTIME}}

\newcommand{\DTIME}{\mbox{\sf DTIME}}

\newcommand{\opt}{\mathsf{OPT}}
\newcommand{\optLP}{\mathsf{OPT}_{\mathsf{LP}}}
\newcommand{\optNDP}{\mathsf{OPT}_{\mathsf{NDP}}}
\newcommand{\optEDP}{\mathsf{OPT}_{\mathsf{EDP}}}
\newcommand{\thset}{\tilde{\mathcal{H}}}

\newcommand{\set}[1]{\left\{ #1 \right\}}

\newcommand{\tset}{{\mathcal T}}
\newcommand{\iset}{{\mathcal{I}}}
\newcommand{\ifamily}{{\mathfrak{I}}}
\newcommand{\pset}{{\mathcal{P}}}
\newcommand{\qset}{{\mathcal{Q}}}

\newcommand{\aset}{{\mathcal{A}}}
\newcommand{\cset}{{\mathcal{C}}}
\newcommand{\fset}{{\mathcal{F}}}
\newcommand{\mset}{{\mathcal M}}
\newcommand{\hmset}{\hat{\mathcal M}}
\newcommand{\tmset}{\tilde{\mathcal M}}
\newcommand{\tG}{\tilde{G}}
\newcommand{\nset}{{\mathcal N}}
\newcommand{\jset}{{\mathcal{J}}}

\newcommand{\wset}{{\mathcal{W}}}

\newcommand{\yset}{{\mathcal{Y}}}
\newcommand{\rset}{{\mathcal{R}}}

\newcommand{\hset}{{\mathcal{H}}}
\newcommand{\sset}{{\mathcal{S}}}
\newcommand{\dset}{{\mathcal{D}}}
\newcommand{\Y}{\Upsilon}

\newcommand{\zset}{{\mathcal{Z}}}

\newcommand{\be}{\begin{enumerate}}
\newcommand{\ee}{\end{enumerate}}
\newcommand{\bd}{\begin{description}}
\newcommand{\ed}{\end{description}}
\newcommand{\bi}{\begin{itemize}}
\newcommand{\ei}{\end{itemize}}

\newtheorem{theorem}{Theorem}[section]
\newtheorem{lemma}[theorem]{Lemma}
\newtheorem{observation}[theorem]{Observation}
\newtheorem{corollary}[theorem]{Corollary}
\newtheorem{claim}[theorem]{Claim}

\newtheorem*{definition}{Definition.}
\def\stopproof{\square}
\def\square{\vbox{\hrule height.2pt\hbox{\vrule width.2pt height5pt \kern5pt
\vrule width.2pt} \hrule height.2pt}}

\newenvironment{proof}{\par \smallskip{\bf Proof:}}{\hfill\stopproof}

\newenvironment{proofof}[1]{\noindent{\bf Proof of #1.}}%
        {\hfill\stopproof}

\renewcommand{\phi}{\varphi}
\newcommand{\eps}{\epsilon}

\newcommand{\poly}{\operatorname{poly}}

\newcommand{\expect}[2][]{\text{\bf E}_{#1}\left [#2\right]}
\newcommand{\prob}[2][]{\text{\bf Pr}_{#1}\left [#2\right]}

\newenvironment{properties}[2][0]
{
\begin{enumerate} \setcounter{enumi}{#1}}{\end{enumerate}}

\begin{document}

\title{Improved Approximation Algorithm for Node-Disjoint Paths in Grid Graphs with Sources on Grid Boundary}
\author{Julia Chuzhoy\thanks{Toyota Technological Institute at Chicago. Email: {\tt cjulia@ttic.edu}. Supported in part by NSF grants CCF-1318242 and CCF-1616584.}\and David H. K. Kim\thanks{Computer Science Department, University of Chicago. Email: {\tt hongk@cs.uchicago.edu}. Supported in part by NSF grants CCF-1318242 and CCF-1616584.} \and Rachit Nimavat\thanks{Toyota Technological Institute at Chicago. Email: {\tt nimavat@ttic.edu}. Supported in part by NSF grant CCF-1318242.}}

\begin{titlepage}

\maketitle
\thispagestyle{empty}

\maketitle
\begin{abstract}
We study the classical Node-Disjoint Paths (\NDP) problem: given an undirected $n$-vertex graph $G$, together with a set $\set{(s_1,t_1),\ldots,(s_k,t_k)}$ of pairs of its vertices,  called source-destination, or demand pairs, find a maximum-cardinality set $\pset$ of mutually node-disjoint paths that connect the demand pairs. The best current approximation for the problem is achieved by a simple greedy $O(\sqrt{n})$-approximation algorithm. Until recently, the best negative result was an $\Omega(\log^{1/2-\eps}n)$-hardness of approximation, for any fixed $\eps$, under standard complexity assumptions. 
A special case of the problem, where the underlying graph is a grid, has been studied extensively. The best current approximation algorithm for this special case achieves an $\tilde{O}(n^{1/4})$-approximation factor. On the negative side, a recent result by the authors shows that \NDP is hard to approximate to within factor  $2^{\Omega(\sqrt{\log n})}$, even if the underlying graph is a subgraph of a grid, and all source vertices lie on the grid boundary. In a very recent follow-up work, the authors further show that \NDP in grid graphs is hard to approximate to within factor $\Omega(2^{\log^{1-\eps}n})$  for any constant $\eps$ under standard complexity assumptions, and to within factor $n^{\Omega(1/(\log\log n)^2)}$ under randomized ETH.

 In this paper we study the \NDP problem in grid graphs, where all source vertices $\set{s_1,\ldots,s_k}$ appear on the grid boundary. Our main result is an efficient randomized  $\approxfactor$-approximation algorithm for this problem. Our result in a sense complements the  $2^{\Omega(\sqrt{\log n})}$-hardness of approximation for sub-graphs of grids with sources lying on the grid boundary, and should be contrasted with the above-mentioned almost polynomial hardness of approximation of \NDP in grid graphs (where the sources and the destinations may lie anywhere in the grid).
Much of the work on approximation algorithms for \NDP relies on the multicommodity flow relaxation of the problem, which is known to have an $\Omega(\sqrt n)$ integrality gap, even in grid graphs, with all source and destination vertices lying on the grid boundary. Our work departs from this paradigm, and uses a (completely different) linear program only to select the pairs to be routed, while the routing itself is computed by other methods.  We generalize this result to instances where the source vertices lie within a prescribed distance from the grid boundary.

 \end{abstract}
 \end{titlepage}

\label{-------------------------------------sec: intro-----------------------------------}
\section{Introduction}\label{sec:intro}

We study the classical Node-Disjoint Paths (\NDP) problem, where the input consists of an undirected $n$-vertex graph $G$ and a collection $\mset=\set{(s_1,t_1),\ldots,(s_k,t_k)}$ of pairs of its vertices, called \emph{source-destination} or \emph{demand} pairs. We say that a path $P$ \emph{routes} a demand pair $(s_i,t_i)$ iff the endpoints of $P$ are $s_i$ and $t_i$. The goal is to compute a maximum-cardinality set $\pset$ of node-disjoint paths, where each path $P\in \pset$ routes a distinct demand pair in $\mset$. We denote by \NDPplanar the special case of the problem when the underlying graph $G$ is planar, and by \NDPgrid the special case where $G$ is a square grid\footnote{We use the standard  convention of denoting $n=|V(G)|$, and so the grid has dimensions $(\sqrt{n}\times \sqrt{n})$; we assume that $\sqrt{n}$ is an integer.}. We refer to the vertices in  set $S=\set{s_1,\ldots,s_k}$ as \emph{source vertices};  to the vertices in set $T=\set{t_1,\ldots,t_k}$ as   \emph{destination vertices}, and to the vertices in set $S\cup T$ as \emph{terminals}.

\NDP is a fundamental graph routing problem that has been studied extensively in both graph theory and theoretical computer science communities. Robertson and Seymour~\cite{RobertsonS,flat-wall-RS} explored the problem in their Graph Minor series, providing an efficient algorithm for \NDP when  the number $k$ of the demand pairs is bounded by a constant. But when $k$ is a part of input, the problem becomes $\NP$-hard~\cite{Karp-NDP-hardness,EDP-hardness}, even in planar graphs~\cite{npc_planar}, and even in grid graphs~\cite{npc_grid}. The best current approximation factor of $O(\sqrt{n})$ for \NDP is achieved by a simple greedy algorithm \cite{KolliopoulosS}. Until recently, this was also the best approximation algorithm for \NDPplanar and \NDPgrid.
A natural way to design approximation algorithms for \NDP is via the multicommodity flow relaxation: instead of connecting each routed demand pair with a path, send maximum possible amount of (possibly fractional) flow between them. The optimal solution to this relaxation can be computed via a standard linear program. The $O(\sqrt{n})$-approximation algorithm of~\cite{KolliopoulosS}  can be cast as an LP-rounding algorithm of this relaxation. Unfortunately, it is well-known that the  integrality gap of this relaxation is $\Omega(\sqrt{n})$, even when the underlying graph is a grid, with all terminals lying on its boundary. In a recent work, Chuzhoy and Kim~\cite{NDP-grids} designed an  $\tilde{O}(n^{1/4})$-approximation for \NDPgrid, thus bypassing this integrality gap barrier. Their main observation is that, if all terminals lie close to the grid boundary (say within distance $O(n^{1/4})$), then a simple dynamic programming-based algorithm yields an  $O(n^{1/4})$-approximation. On the other hand, if, for every demand pair, either the source or the destination lies at a distance at least $\Omega(n^{1/4})$ from the grid boundary, then the integrality gap of the multicommodity flow relaxation improves, and one can obtain an $\tilde{O}(n^{1/4})$-approximation  via LP-rounding. A natural question is whether the integrality gap improves even further, if all terminals lie further away from the grid boundary. Unfortunately, the authors show in~\cite{NDP-grids} that the integrality gap remains at least $\Omega(n^{1/8})$, even if all terminals lie within distance $\Omega(\sqrt n)$ from the grid boundary.
The $\tilde{O}(n^{1/4})$-approximation algorithm for \NDPgrid was later extended and generalized to an $\tilde O(n^{9/19})$-approximation algorithm for \NDPplanar~\cite{NDP-planar}.

On the negative side, until recently, only an $\Omega(\log^{1/2-\eps}n)$-hardness of approximation was known for the general version of \NDP, for any constant $\eps$, unless $\NP \subseteq \ZPTIME(n^{\poly \log n})$~\cite{AZ-undir-EDP,ACGKTZ},  and only APX-hardness was known for \NDPplanar and \NDPgrid~\cite{NDP-grids}.
In a recent work~\cite{NDP-hardness}, the authors have shown that \NDP is hard to approximate to within a $2^{\Omega(\sqrt{\log n})}$ factor  unless $\NP\subseteq \DTIME(n^{O(\log n)})$, even if the underlying graph is a planar graph with maximum vertex degree at most $3$, and all source vertices lie on the boundary of a single face. The result holds even when the input graph $G$ is a vertex-induced subgraph of a grid, with all sources lying on the grid boundary. In a very recent work~\cite{NDP-hardness-grid}, the authors show that \NDPgrid is $2^{\Omega(\log^{1-\eps}n)}$-hard to approximate for any constant $\eps$ assuming $\NP\nsubseteq \BPTIME(n^{\poly\log n})$, and moreover, assuming randomized ETH, the hardness of approximation factor becomes $n^{\Omega(1/(\log\log n)^2)}$.  We note that the instances constructed in these latter hardness proofs require all terminals to lie far from the grid boundary.

In this paper we explore \NDPgrid. This important special case of \NDP was initially motivated by applications in VLSI design, and has received a lot of attention since the 1960's. 
We focus on a restricted version of \NDPgrid, that we call \rNDPgrid: here, in addition to the graph $G$ being a square grid, we also require that all source vertices $\set{s_1,\ldots,s_k}$ lie on the grid boundary. We do not make any assumptions about the locations of the destination vertices, that may appear anywhere in the grid. The best current approximation algorithm for \rNDPgrid is the same as that for the general \NDPgrid, and achieves a $\tilde{O}(n^{1/4})$-approximation~\cite{NDP-grids}. Our main result is summarized in the following theorem.

\begin{theorem}\label{thm: main}
There is an efficient randomized $\approxfactor$-approximation algorithm for \restrictedNDP.
\end{theorem}


This result in a sense complements the $2^{\Omega(\sqrt{\log n})}$-hardness of approximation of \NDP on sub-graphs of grids with all sources lying on the grid boundary of~\cite{NDP-hardness}\footnote{Note that the two results are not strictly complementary: our algorithm only applies to grid graphs, while the hardness result is only valid for sub-graphs of grids.}, and should be contrasted with the recent almost polynomial hardness of approximation of~\cite{NDP-hardness-grid} for \NDPgrid mentioned above. Our algorithm departs from previous work on NDP in that it does not use the multicommodity flow relaxation. Instead, we define sufficient conditions that allow us to route a subset $\mset'$ of demand pairs via disjoint paths, and show that there exists a subset of demand pairs satisfying these conditions, whose cardinality is at least $\opt/\approxfactor$, where $\opt$ is the value of the optimal solution. It is then enough to compute a maximum-cardinality subset of the demand pairs satisfying these conditions. We write an LP-relaxation for this problem and design a $\approxfactor$-approximation LP-rounding algorithm for it. We emphasize that the linear program is only used to select the demand pairs to be routed, and not to compute the routing itself.

We then generalize the result to instances where the source vertices lie within a prescribed distance from the grid boundary.

\begin{theorem} \label{thm: sources at dist d}
For every integer $\delta \geq 1$, there is an efficient randomized $\left(\delta \cdot \approxfactor\right )$-approximation algorithm for the special case of \NDPgrid where all source vertices lie within distance at most $\delta$ from the grid boundary.
\end{theorem}


We note that for instances of \NDPgrid where both the sources and the destinations are within distance at most $\delta$ from the grid boundary, it is easy to obtain an efficient $O(\delta)$-approximation algorithm (see, e.g.~\cite{NDP-grids}).

A problem closely related to \NDP is the Edge-Disjoint Paths (\EDP) problem. It is defined similarly, except that now the paths chosen to route the demand pairs may share vertices, and are only required to be edge-disjoint. 
 The approximability status of \EDP is very similar to that of \NDP: there is an $O(\sqrt n)$-approximation algorithm~\cite{EDP-alg}, and an $\Omega(\log^{1/2-\eps}n)$-hardness of approximation for any constant $\eps$, unless $\NP \subseteq \ZPTIME(n^{\poly \log n})$~\cite{AZ-undir-EDP,ACGKTZ}.
As in the \NDP problem, we can use the standard multicommodity flow LP-relaxation of the problem, in order to obtain the $O(\sqrt n)$-approximation algorithm, and the integrality gap of the LP-relaxation is $\Omega(\sqrt n)$ even in planar graphs. 
Recently, Fleszar et al.~\cite{fleszar_et_al} designed an  $O(\sqrt{r}\cdot \log (kr))$-approximation algorithm for \EDP, where $r$ is the feedback vertex set number of the input graph $G=(V,E)$ --- the smallest number of vertices that need to be deleted from $G$ in order to turn it into a forest. 

Several special cases of \EDP have better approximation algorithms: an $O(\log^2n)$-approximation is known for even-degree planar graphs \cite{CKS,CKS-planar1,Kleinberg-planar}, and an $O(\log n)$-approximation is known for nearly-Eulerian uniformly high-diameter planar graphs, and nearly-Eulerian densely embedded graphs, including grid graphs~\cite{grids1,grids3,grids4}. Furthermore, an $O(\log n)$-approximation algorithm is known for \EDP on 4-edge-connected planar, and Eulerian planar graphs~\cite{KK-planar}. It appears that the restriction of the graph $G$ to be Eulerian, or near-Eulerian, makes the \EDP problem on planar graphs significantly simpler, and in particular improves the integrality gap of the standard multicommodity flow LP-relaxation.

The analogue of the grid graph for the \EDP problem is the wall graph (see Figure~\ref{fig: wall}): the integrality gap of the multicommodity flow relaxation for \EDP on wall graphs is $\Omega(\sqrt n)$. The $\tilde O(n^{1/4})$-approximation algorithm of~\cite{NDP-grids} for \NDPgrid extends to \EDP on wall graphs, and the $2^{\Omega(\sqrt{\log n})}$-hardness of approximation of~\cite{NDP-hardness} for \NDPplanar also extends to \EDP  on sub-graphs of walls, with all sources lying on the top boundary of the wall. The recent hardness result of~\cite{NDP-hardness-grid} for \NDPgrid also extends to an $2^{\Omega(\log^{1-\eps}n)}$-hardness of \EDP on wall graphs, assuming $\NP\nsubseteq \BPTIME(n^{\poly\log n})$, and to $n^{\Omega(1/(\log\log n)^2)}$-hardness assuming randomized ETH. We extend our results to \EDP and \NDP on wall graphs: 

\begin{theorem}\label{thm: main edp}
There is an efficient randomized $\approxfactor$-approximation algorithm for \EDP and for \NDP on wall graphs, when all source vertices lie on the wall boundary.
\end{theorem}


\begin{figure}[h]
\centering
\scalebox{0.4}{\includegraphics{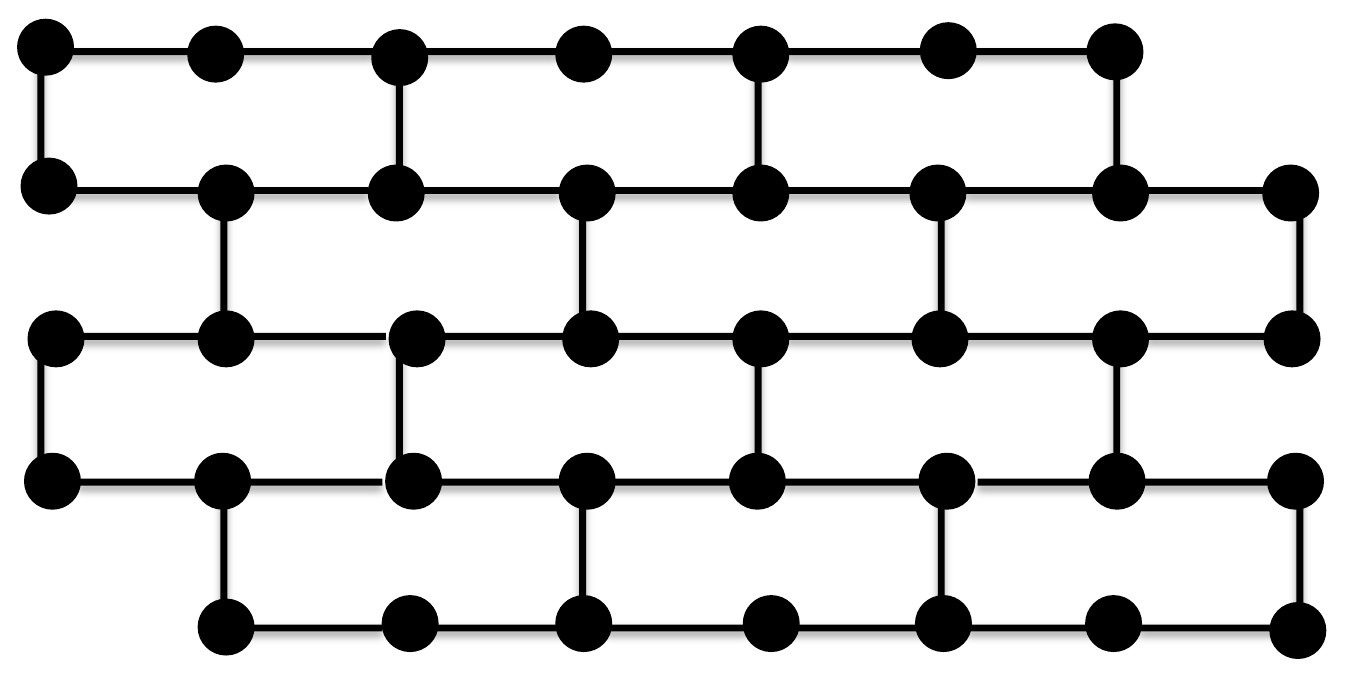}}
\caption{A wall graph.\label{fig: wall}}
\end{figure}


\paragraph*{Other related work.} Cutler and Shiloach~\cite{Cutler-Shiloach} studied an even more restricted version of \NDPgrid, where all source vertices lie on the top row $R^*$ of the grid, and all destination vertices lie on a single row  $R'$ of the grid, far enough from its top and bottom boundaries. They considered three different settings of this special case. In the packed-packed setting, all sources appear consecutively on $R^*$, and all destinations appear consecutively on $R'$ (but both sets may appear in an arbitrary order). They show a necessary and a sufficient condition for all demand pairs to be routable via node-disjoint paths in this setting. The second setting is the packed-spaced setting. Here, the sources again appear consecutively on $R^*$, but all destinations are at a distance at least $d$ from each other. For this setting, the authors show that if $d\geq k$, then all demand pairs can be routed. We note that \cite{NDP-grids} extended their algorithm to a more general setting, where the destination vertices may appear anywhere in the grid, as long as the distance between any pair of the destination vertices, and any destination vertex and the boundary of the grid, is at least $\Omega(k)$. 
Robertson and Seymour~\cite{NDP-surface} provided sufficient conditions for the existence of node-disjoint routing of a given set of demand pairs in the more general setting of graphs drawn on surfaces, and they designed an algorithm whose running time is $\poly(n)\cdot f(k)$ for finding the routing, where $f(k)$ is at least exponential in $k$. Their result implies the existence of the routing in grids, when the destination vertices are sufficiently far from each other and from the grid boundaries, but it does not provide an efficient algorithm to compute such a routing.
The third setting studied by Cutler and Shiloach is the spaced-spaced setting, where the distances between every pair of source vertices, and every pair of destination vertices are at least $d$. The authors note that they could not come up with a better algorithm for this setting, than the one provided for the packed-spaced case.
Aggarwal, Kleinberg, and Williamson~\cite{AKW} considered a special case of \NDPgrid, where the set of  the demand pairs is a permutation: that is, every vertex of the grid participates in exactly one demand pair. They show that $\Omega(\sqrt{n}/\log n)$ demand pairs are routable in this case via node-disjoint paths. They further show that if all terminals are at a distance at least $d$ from each other, then at least $\Omega(\sqrt{nd}/\log n)$ pairs are routable. 

A variation of the NPD and EDP problems, where small congestion is allowed, has been a subject of extensive study, starting with the classical paper of Raghavan and Thompson~\cite{RaghavanT} that introduced the randomized rounding technique. We say that a set $\pset$ of paths causes congestion $c$, if at most $c$ paths share the same vertex or the same edge, for the \NDP and the \EDP settings respectively. A recent line of work~\cite{CKS,Raecke,Andrews,RaoZhou,Chuzhoy11,  ChuzhoyL12,ChekuriE13,NDPwC2} has lead to an $O(\poly\log k)$-approximation for both \NDP and \EDP problems with congestion $2$. For planar graphs, a constant-factor approximation with congestion 2 is known~\cite{EDP-planar-c2}. 

\paragraph*{Organization.} We start with a high-level intuitive overview of our algorithm in Section~\ref{sec: alg-high-level}. We then provide Preliminaries in Section~\ref{sec:prelims} and the algorithm for \rNDPgrid in Section~\ref{sec:alg-overview}, with parts of the proof being deferred to Sections~\ref{sec: shadow property}--\ref{sec:destinations anywhere}. We extend our algorithm to \EDP and \NDP on wall graphs in Section~\ref{sec: ndp to edp}. We generalize our algorithm to the setting where the sources are within some prescribed distance from the grid boundary in Section~\ref{sec: sources at dist d}.

\label{-------------------------------------sec: alg-overview -----------------------------------}
\section{High-Level Overview of the Algorithm}\label{sec: alg-high-level}

The goal of this section is to provide an informal high-level overview of the main result of the paper -- the proof of Theorem~\ref{thm: main}. 
With this goal in mind, the values of various parameters  are given imprecisely in this section, in a way that best conveys the intuition. 
The following sections contain a formal description of the algorithm and the precise settings of all parameters.

We first consider an even more restricted special case of \NDPgrid, where all source vertices appear on the top boundary of the grid, and all destination vertices appear far enough from the grid boundary, and design an efficient randomized  $\approxfactor$-approximation algorithm $\aset$ for this problem. We later show how to reduce \restrictedNDP to this special case of the problem; we focus on the description of the algorithm $\aset$ for now.

We assume that our input graph $G$ is the $(\ell\times \ell)$-grid, and we denote by $n=\ell^2$ the number of its vertices. We further assume that the set of the demand pairs is $\mset=\set{(s_1,t_1),\ldots,(s_k,t_k)}$, with the vertices in set $S=\set{s_1,\ldots,s_k}$ called source vertices; the vertices in set $T=\set{t_1,\ldots,t_k}$ called destination vertices; and the vertices in $S\cup T$ called terminals. Let $\opt$ denote the value of the optimal solution to the \NDP instance $(G,\mset)$. We assume that the vertices of $S$ lie on the top boundary of the grid, that we denote by $R^*$, and the vertices of $T$ lie sufficiently far from the grid boundary -- say, at a distance at least $\opt$ from it. For a subset $\mset'\subseteq \mset$ of the demand pairs, we denote by $S(\mset')$ and $T(\mset')$ the sets of the source and the destination vertices of the demand pairs in $\mset'$, respectively. 
As our starting point, we consider a simple observation of Chuzhoy and Kim~\cite{NDP-grids}, that generalizes the results of Cutler and Shiloach~\cite{Cutler-Shiloach}.
Suppose we are given an instance of \NDPgrid with $k$ demand pairs, where the sources lie on the top boundary of the grid, and the destination vertices may appear anywhere in the grid, but the distance between every pair of the destination vertices, and every destination vertex and the boundary of the grid, is at least $(8k+8)$ -- we call such instances \emph{spaced-out instances}. In this case, all demand pairs in $\mset$ can be efficiently routed via node-disjoint paths, as follows. Consider, for every destination vertex $t_i\in T$, a square sub-grid $B_i$ of $G$, of size $(2k\times 2k)$, such that $t_i$ lies roughly at the center of $B_i$. We construct a set $\pset$ of $k$ node-disjoint paths, that originate at the vertices of $S$, and traverse the sub-grids $B_i$ one-by-one in a snake-like fashion (see a schematic view on Figure~\ref{fig: global-routing}). We call this part of the routing \emph{global routing}. The \emph{local routing} needs to specify how the paths in $\pset$ traverse each box $B_i$. This is done in a straightforward manner, while ensuring  that the unique path originating at vertex $s_i$ visits the vertex $t_i$ (see Figure~\ref{fig: local-routing}). By suitably truncating the final set $\pset$ of paths, we obtain a routing of all demand pairs in $\mset$ via node-disjoint paths.

\begin{figure}[h]
\centering
\subfigure[Global routing. In this figure, the sub-grids $B_i$ are aligned vertically and horizontally. A similar (but somewhat more complicated) routing can be performed even if they are not aligned. For convenience we did not include all source vertices and all paths.]{\scalebox{0.4}{\includegraphics{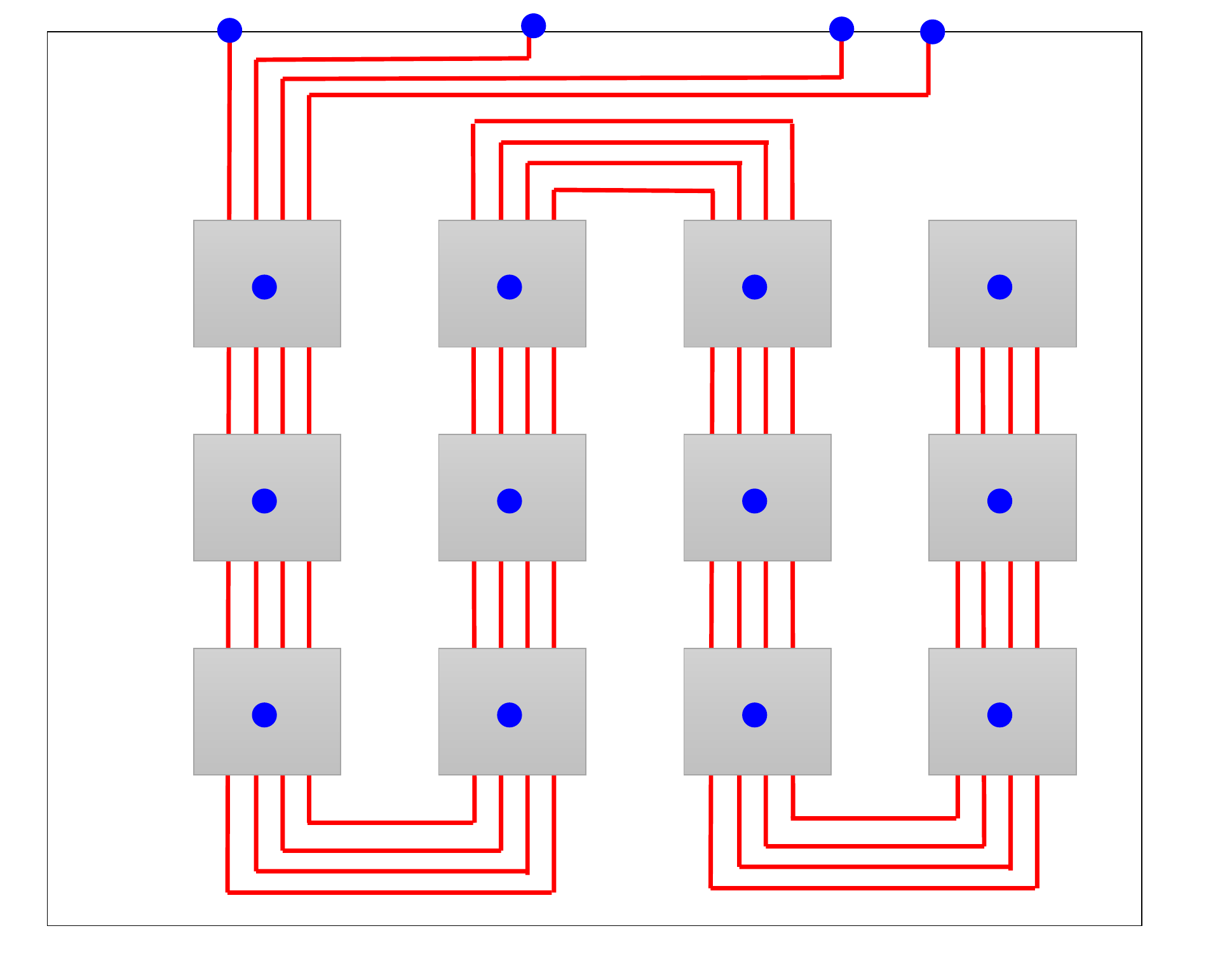}}\label{fig: global-routing}}
\hspace{1cm}
\subfigure[Local routing inside $B_i$]{
\scalebox{0.4}{\includegraphics{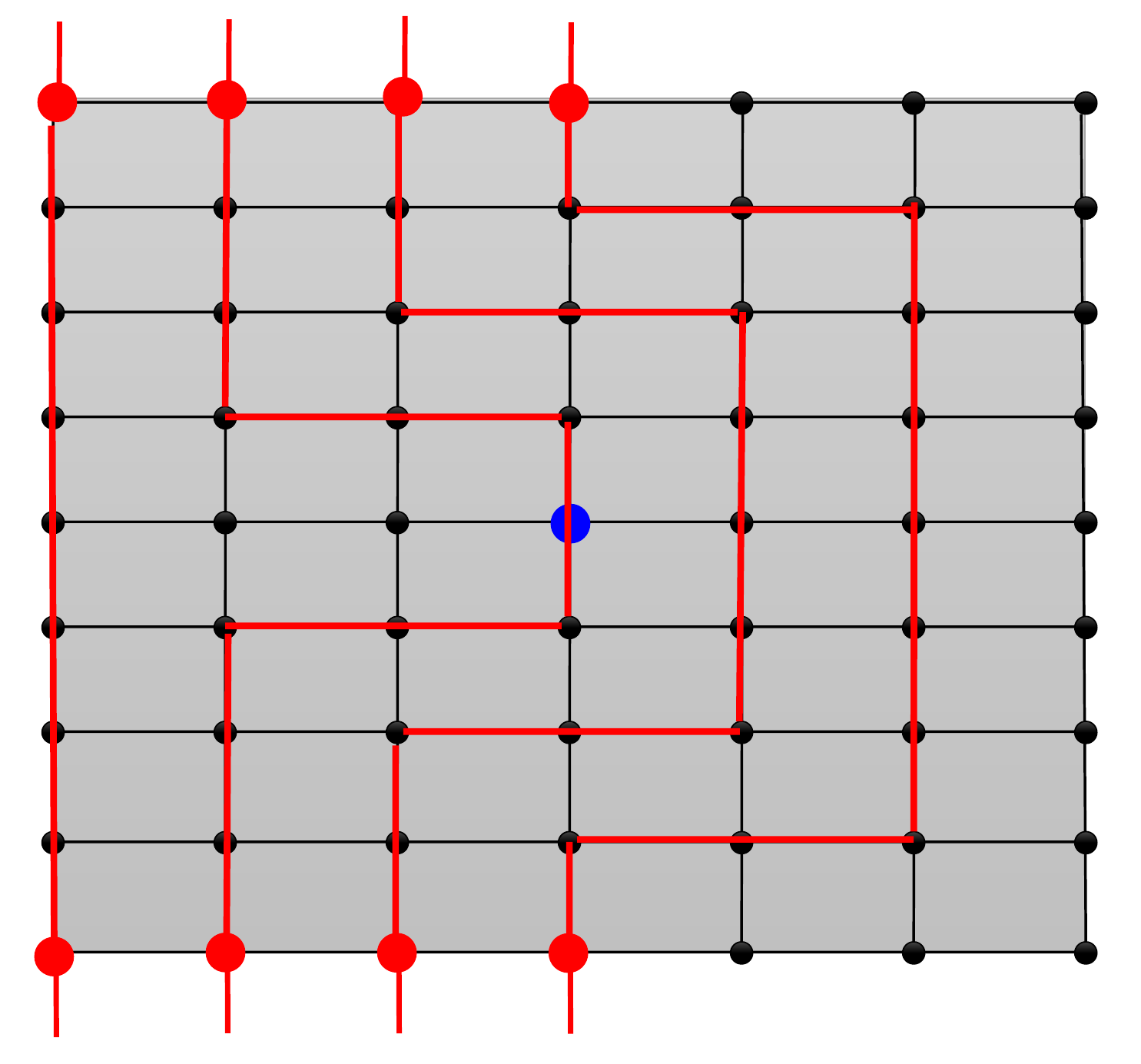}}\label{fig: local-routing}}
\caption{Schematic view of routing of spaced-out instances.\label{fig: routing-spaced}}
\end{figure}

Unfortunately, in our input instance $(G,\mset)$, the destination vertices may not be located sufficiently far from each other. We can try to select a large subset $\mset'\subseteq \mset$ of the demand pairs, so that every pair of destination vertices in $T(\mset')$ appear at a distance at least $\Omega(|\mset'|)$ from each other; but in some cases the largest such set $\mset'$ may only contain $O(\opt/\sqrt{k})$ demand pairs (for example, suppose all destination vertices lie consecutively on a single row of the grid). One of our main ideas is to generalize this simple algorithm to a number of recursive levels.

For simplicity, let us first describe the algorithm with just two recursive levels. Suppose we partition the top row of the grid into $z$ disjoint intervals, $I_1,\ldots,I_z$. Let $\mset'\subseteq \mset$ be a set of demand pairs that we would like to route. Denote $|\mset'|=k'$, and assume that we are given a collection $\qset$ of square sub-grids of $G$, of size $(4k'\times 4k')$ each (that we call \emph{squares}), such that every pair $Q,Q'\in \qset$ of distinct squares is at a distance at least $4k'$ from each other. Assume further that each such sub-grid $Q\in \qset$ is assigned a color $\chi(Q)\in \set{c_1,\ldots,c_z}$, such that, if $Q$ is assigned the color $c_j$, then all demand pairs $(s,t)\in \mset'$ whose destination $t$ lies in $Q$ have their source $s\in I_j$  (so intuitively, each color $c_j$ represents an interval $I_j$). Let $\mset'_j\subseteq \mset'$ be the set of all demand pairs $(s,t)\in \mset'$ with $s\in I_j$. We would like to ensure that $|\mset'_j|$ is roughly $k'/z$, and that all destination vertices of $T(\mset'_j)$ are at a distance at least $|\mset'_j|$ from each other. We claim that if we could find the collection $\set{I_1,\ldots,I_z}$ of the intervals of the first row, a collection $\qset$ of sub-grids of $G$,  a coloring $\chi:\qset\rightarrow \set{c_1,\ldots,c_z}$, and a subset $\mset'\subseteq\mset$ of the demand pairs with these properties, then we would be able to route all demand pairs in $\mset'$. 

In order to do so, for each square $Q\in \qset$, we construct an augmented square $Q^+$, by adding a margin of $k'$ rows and columns around $Q$. Our goal is to construct a collection $\pset$ of node-disjoint paths routing the demand pairs in $\mset'$. We start by constructing a global routing, where all paths in $\pset$ originate from the vertices of $S(\mset')$ and then visit the squares in $\set{Q^+\mid Q\in \qset}$ in a snake-like fashion, just like we did for the spaced-out instances described above (see Figure~\ref{fig: global-routing}). Consider now some square $Q\in \qset$ and the corresponding augmented square $Q^+$. Assume that $\chi(Q)=c_j$, and let $\pset_j\subseteq\pset$ be the set of paths originating at the source vertices that lie in $I_j$. While traversing the square $Q^+$, we ensure that only the paths in $\pset_j$ enter the square $Q$; the remaining paths use the margins on the left and on the right of $Q$ in order to traverse $Q^+$. This can be done because the sources of the paths in $\pset_j$ appear consecutively on $R^*$, relatively to the sources of all paths in $\pset$. In order to complete the local routing inside the square $Q$, observe that the destination vertices appear far enough from each other, and so we can employ the simple algorithm for spaced-out instances inside $Q$.

In order to optimize the approximation factor that we achieve, we extend this approach to $\rho=O(\sqrt{\log n})$ recursive levels.  Let $\eta=2^{\ceil{\sqrt{\log n}}}$.
We define auxiliary parameters $d_1>d_2>\cdots>d_{\rho}>d_{\rho+1}$. Roughly speaking, we can think of $d_{\rho+1}$ as being a constant (say $16$), of $d_1$ as being comparable to $\opt$, and for all $1\leq h\leq \rho$, $d_{h+1}=d_{h}/\eta$. The setup for the algorithm consists of three ingredients: (i) a hierarchical decomposition $\thset$ of the grid into square sub-grids (that we refer to as squares); (ii) a hierarchical partition $\ifamily$ of the first row $R^*$ of the grid into intervals; and (iii) a hierarchical coloring $f$ of the squares in $\thset$ with colors that correspond to the intervals of $\ifamily$, together with a selection of a subset $\mset'\subseteq \mset$ of the demand pairs to route. We define sufficient conditions on the hierarchical system $\thset$ of squares, the hierarchical partition $\ifamily$ of $R^*$ into intervals, the coloring $f$ and the subset $\mset'$ of the demand pairs, under which a routing of all pairs in  $\mset'$ exists and can be found efficiently. For a fixed hierarchical system $\thset$ of squares, a triple $(\ifamily, f, \mset')$  satisfying these conditions is called a \emph{good ensemble}. We show that a good ensemble with a large enough set $\mset'$ of demand pairs exists, and then design an approximation algorithm for computing a good ensemble maximizing $|\mset'|$. We now describe each of these ingredients in turn.


\subsection{A Hierarchical System of Squares}
A hierarchical system $\thset$ of squares consists of a sequence $\qset_1,\qset_2,\ldots,\qset_{\rho}$ of sets of sub-grids of $G$. For each $1\leq h\leq \rho$, 
$\qset_h$ is a collection of disjoint sub-grids of $G$ (that we refer to as \emph{level-$h$ squares}); every such square $Q\in \qset_h$ has size $(d_h\times d_h)$, and every pair of distinct squares $Q,Q'\in \qset_h$ are within distance at least $d_h$ from each other (see Figure~\ref{fig: squares}). We require that for each $1< h\leq \rho$, for every square $Q\in \qset_h$, there is a unique square $Q'\in \qset_{h-1}$ (called the \emph{parent-square} of $Q$) that contains $Q$. We say that a demand pair $(s,t)$ \emph{belongs} to the hierarchical system $\thset=(\qset_1,\qset_2,\ldots,\qset_{\rho})$ of squares iff $t\in \bigcup_{Q\in \qset_{\rho}}Q$. We show a simple efficient algorithm to construct $2^{O(\sqrt{\log n})}$ such hierarchical systems of squares, so that every demand pair belongs to at least one of them. Each such system $\thset$ of squares induces an instance of \NDP ---  the instance is defined over the same graph $G$, and the set $\tmset\subseteq \mset$ of demand pairs that belong to the system $\thset$.   It is then enough to obtain a factor $\approxfactor$-approximation algorithm for each resulting instance $(G,\tmset)$ separately. From now on we fix one such hierarchical system $\thset=(\qset_1,\qset_2,\ldots,\qset_{\rho})$ of squares, together with the set $\tmset\subseteq \mset$ of demand pairs, containing all pairs $(s,t)$ that belong to $\thset$, and focus on designing an $\approxfactor$-approximation algorithm for instance $(G,\tmset)$.

\begin{figure}[h]
\centering
\scalebox{0.3}{\includegraphics{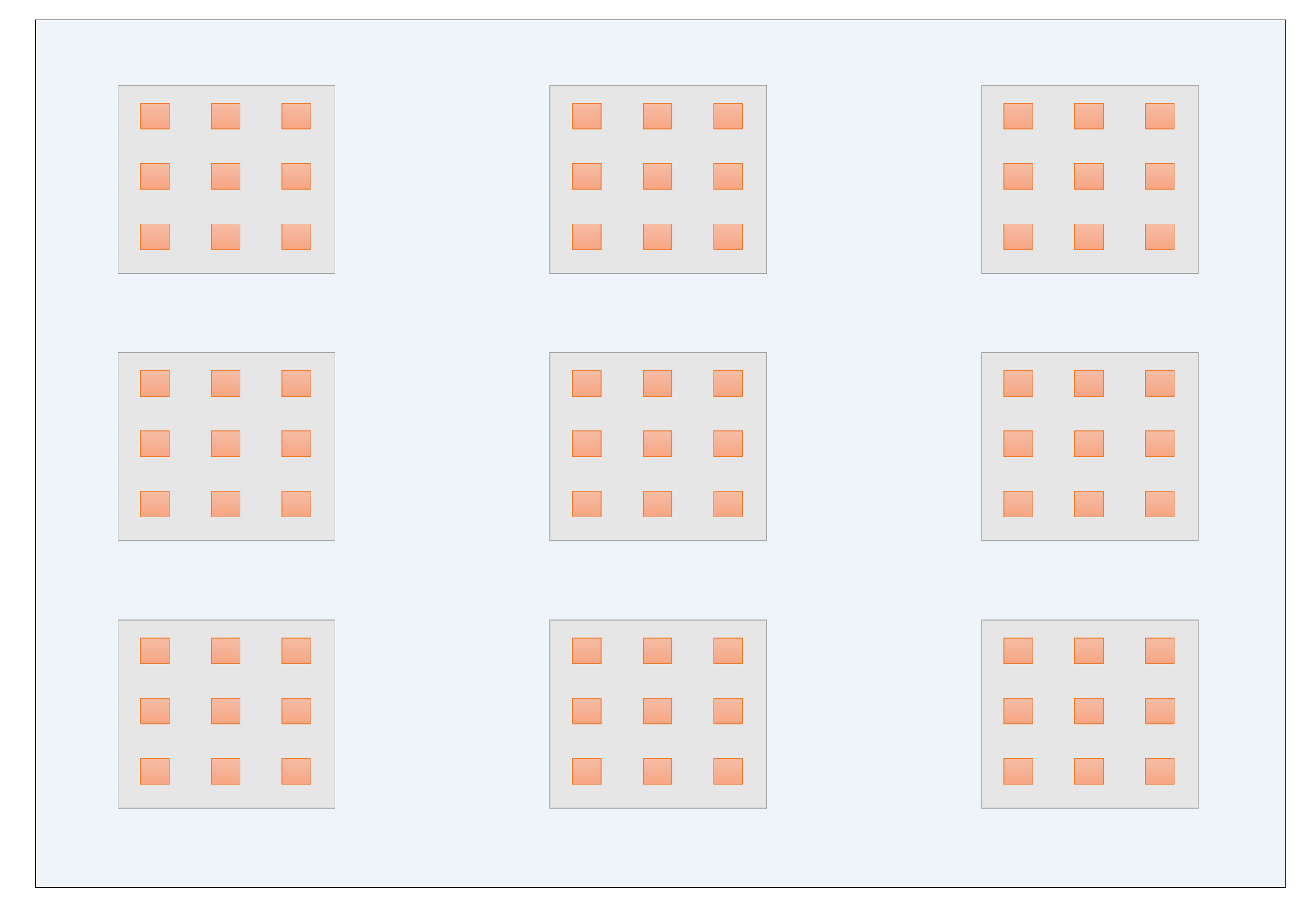}}
\caption{A schematic view of a hierarchical system of squares with 2 levels.\label{fig: squares}}
\end{figure}


\subsection{A Hierarchical Partition of the Top Grid Boundary} Recall that $R^*$ denotes the first row of the grid. A hierarchical partition $\ifamily$ of $R^*$ is a sequence $\iset_1,\iset_2,\ldots,\iset_{\rho}$ of sets of sub-paths of $R^*$, such that for each $1\leq h\leq \rho$, the paths in $\iset_h$ (that we refer to as \emph{level-$h$ intervals}) partition the vertices of $R^*$. We also require that for all $1<h\leq \rho$, every level-$h$ interval $I\in \iset_h$ is contained in a unique level-$(h-1)$ interval $I'\in \iset_{h-1}$, that we refer to as the \emph{parent-interval} of $I$. For every level $1\leq h\leq \rho$, we define a collection $\chi_h$ of colors, containing one color $c_h(I)$ for each level-$h$ interval $I\in \iset_h$. If $I'\in \iset_h$ is a parent-interval of $I\in \iset_{h+1}$, then we say that color $c_h(I')$ is a \emph{parent-color} of $c_{h+1}(I)$.


\subsection{Coloring the Squares and Selecting Demand Pairs to Route}
The third ingredient of our algorithm is an assignment $f$ of colors to the squares, and a selection of a subset of the demand pairs to be routed. For every level $1\leq h\leq \rho$, for every level-$h$ square $Q\in \qset_h$, we would like to assign a single level-$h$ color $c_h(I)\in \chi_h$ to $Q$, denoting $f(Q)=c_h(I)$. Intuitively, if color $c_h(I)$ is assigned to $Q$, then the only demand pairs $(s,t)$ with $t\in Q$ that we may route are those whose source vertex $s$ lies on the level-$h$ interval $I$. We require that the coloring is consistent across levels: that is, for all $1<h\leq \rho$, if a level-$h$ square is assigned a level-$h$ color $c_h$, and its parent-square is assigned a level-$(h-1)$ color $c_{h-1}$, then $c_{h-1}$ must be a parent-color of $c_h$. We call such a coloring $f$ a \emph{valid coloring} of $\thset$ with respect to $\ifamily$.

Finally, we would like to select a subset $\mset'\subseteq \tmset$ of the demand pairs to route. Consider some demand pair $(s,t)$ and some level $1\leq h\leq \rho$. Let $I_h$ be the level-$h$ interval to which $s$ belongs. Then we say that $s$ has the level-$h$ color $c_h(I_h)$. Therefore, for each level $1\leq h\leq \rho$, 
vertex $s$ is assigned the unique level-$h$ color $c_h(I_h)$, and for $1\leq h<\rho$, $c_{h}(I_{h})$ is the parent-color of $c_{h+1}(I_{h+1})$.
 Let $Q_{\rho}\in \qset_{\rho}$ be the level-$\rho$ square to which $t$ belongs. We may only add $(s,t)$ to $\mset'$ if the level-$\rho$ color of $Q_{\rho}$ is $c_{\rho}(I_{\rho})$ (that is, it is the same as the level-$\rho$ color of $s$). Notice that in particular, this means that for every level $1\leq h\leq\rho$, if $Q_h$ is the level-$h$ square containing $t$, and it is assigned the color $c_h(I_h)$, then $s$ is assigned  the same level-$h$ color, and so $s\in I_h$. Finally, we require that for all $1\leq h\leq \rho$, for every level-$h$ color $c_h$, the total number of all demand pairs $(s,t)\in \mset'$, such that the level-$h$ color of $s$ is $c_h$, is no more than $d_{h+1}/16$ (if $h=\rho$, then the number is no more than $1$). If $\mset'$ has all these properties, then we say that it \emph{respects the coloring $f$}. We say that $(\ifamily, f,\mset')$ is a \emph{good ensemble} iff $\ifamily$ is a hierarchical partition of $R^*$ into intervals; $f$ is a valid coloring of the squares in $\thset$ with respect to $\ifamily$; and $\mset'\subseteq \tmset$ is a subset of the demand pairs that respects the coloring $f$. The \emph{size} of the ensemble is $|\mset'|$.

\subsection{The Routing}  We show that, if we are given a good ensemble $(\ifamily, f,\mset')$, then we can route all demand pairs in $\mset'$. The routing itself follows the high-level idea outlined above. We gradually construct a collection $\pset$ of node-disjoint paths routing the demand pairs in $\mset'$. At the highest level, all these paths depart from their sources and then visit the level-$1$ squares one-by-one, in a snake-like fashion, as in Figure~\ref{fig: global-routing}. Consider now some level-$1$ square $Q$, and assume that its level-$1$ color is $c_1(I)$, where $I\in \iset_1$ is some level-1 interval of $R^*$. Then only the paths $P\in \pset$ that originate at the vertices of $I$ will enter the square $Q$; the remaining paths will exploit the spacing between the level-$1$ squares in order to bypass it; the spacing between the level-$1$ squares is sufficient to allow this. Once we have defined this global routing, we need to specify how the routing is carried out inside each square. We employ the same procedure recursively. Consider some level-$1$ square $Q$, and let $\pset'\subseteq \pset$ be the set of all paths that visit $Q$. Assume further that the level-$1$ color of $Q$ is $c_1(I)$. Since we are only allowed to have at most $d_2/16$ demand pairs in $\mset'$ whose level-1 color is $c_1(I)$, $|\pset'|\leq d_2/16$. Let $\qset'\subseteq \qset_2$ be the set of all level-$2$ squares contained in $Q$. The paths in $\pset'$ will visit the squares of $\qset'$ one-by-one in a snake-like fashion (but this part of the routing is performed inside $Q$). As before, for every level-2 square $Q'\subseteq Q$, if the level-$2$ color of $Q'$ is $c_2(I')$, then only those paths of $\pset'$ that originate at the vertices of $I'$ will enter $Q'$; the remaining paths will use the spacing between the level-$2$ squares to bypass $Q'$. Since $|\pset'|\leq d_2/16$, and all level-$2$ squares are at distance at least $d_2$ from each other, there is a sufficient spacing to allow this routing. We continue this process recursively, until, at the last level of the recursion, we route at most one path per color, to its destination vertex.

In order to complete the proof of the theorem, we need to show that there exists a good ensemble $(\ifamily, f,\mset')$ of size $|\mset'|\geq |\opt|/\approxfactor$, and that we can find such an ensemble efficiently.

\subsection{The Existence of the Ensemble}
The key notion that we use in order to show that a large good ensemble  $(\ifamily, f,\mset')$ exists is that of a \emph{shadow property}. Suppose $Q$ is some $(d\times d)$ sub-grid of $G$, and let $\hmset\subseteq \mset$ be some subset of the demand pairs. Among all demand pairs $(s,t)\in \hmset$ with $t\in Q$, let $(s_1,t_1)$ be the one with $s_1$ appearing earliest on the first row $R^*$ of $G$, and let $(s_2,t_2)$ be the one with $s_2$ appearing latest on $R^*$. The \emph{shadow of $Q$ with respect to $\hmset$} is the sub-path of $R^*$ between $s_1$ and $s_2$. Let $N_{\hmset}(Q)$ be the number of all demand pairs $(s,t)\in \hmset$ with $s$ lying in the shadow of $Q$ (that is, $s$ lies between $s_1$ and $s_2$ on $R^*$). We say that $\hmset$ has the \emph{shadow property with respect to $Q$} iff $N_{\hmset}(Q)\leq d$. We say that $\hmset$ has the \emph{shadow property with respect to the hierarchical system $\thset=(\qset_1,\ldots,\qset_{\rho})$ of squares}, iff $\hmset$ has the shadow property with respect to every square in $\bigcup_{h=1}^{\rho}\qset_h$. Let $\pset^*$ be the optimal solution to the instance $(G,\tmset)$ of \NDP, where $\tmset$ only includes the demand pairs that belong to $\thset$. 
Let $\mset^*\subseteq \tmset$ be the set of the demand pairs routed by $\pset^*$. For every demand pair $(s,t)\in \mset^*$, let $P(s,t)\in \pset^*$ be the path routing this demand pair. Intuitively, it feels like $\mset^*$ should have the shadow property. Indeed, let $Q\in \bigcup_{h=1}^{\rho}\qset_h$ be some square of size $(d_h\times d_h)$, and let $(s_1,t_1),(s_2,t_2)\in \mset^*$ be defined for $Q$ as before, so that the shadow of $Q$ with respect to $\mset^*$ is the sub-path of $R^*$ between $s_1$ and $s_2$. Let $P$ be any path of length at most $2d_h$ connecting $t_1$ to $t_2$ in $Q$, and let $\gamma$ be the closed curve consisting of the union of $P(s_1,t_1)$, $P$, $P(s_2,t_2)$, and the shadow of $Q$.  Consider the disc $D$ whose boundary is $\gamma$. The intuition is that, if $(s,t)\in \mset^*$ is a demand pair whose source lies in the shadow of $Q$, and destination lies outside of $D$, then $P(s,t)$ must cross the path $P$, as it needs to escape the disc $D$. Since path $P$ is relatively short, only a small number of such demand pairs may exist. The main difficulty with this argument is that we may have a large number of demand pairs $(s,t)$, whose source lies in the shadow of $Q$, and the destination lies in the disc $D$. Intuitively, this can only happen if $P(s_1,t_1)$ and $P(s_2,t_2)$ ``capture'' a large area of the grid. We show that, in a sense, this cannot happen too often, and that there is a subset $\mset^{**}\subseteq \mset^*$ of at least $|\mset^*|/\approxfactor$ demand pairs, such that $\mset^{**}$ has the shadow property with respect to $\thset$.

Finally, we show that there exists a good ensemble $(\ifamily,f,\mset')$ with $|\mset'|\geq |\mset^{**}|/\approxfactor$. 
We construct the ensemble over the course of $\rho$ iterations, starting with $\mset'=\mset^{**}$. In the $h$th iteration we construct the set $\iset_h$ of the level-$h$ intervals of $R^*$, assign level-$h$ colors to all level-$h$ squares of $\thset$, and discard some demand pairs from $\mset'$. Recall that  $\eta=2^{\ceil{\sqrt{\log n}}}$. In the first iteration, we let $\iset_1$ be a partition of the row $R^*$ into intervals, each of which contains roughly $\frac{d_1}{16\eta}=\frac{d_2}{16} \leq \frac{|\mset^*|}{\eta}$ vertices of $S(\mset')$. Assume that these intervals are $I_1,\ldots,I_r$, and that they appear in this left-to-right order on $R^*$. We call all intervals $I_j$ where $j$ is odd \emph{interesting intervals}, and the remaining intervals $I_j$ \emph{uninteresting intervals}.  We discard from $\mset'$ all demand pairs $(s,t)$, where $s$ lies on an uninteresting interval.  Consider now some level-$1$ square $Q$, and let $\mset(Q)\subseteq \mset'$ be the set of all demand pairs whose destinations lie in $Q$. Since the original set $\mset^{**}$ of demand pairs had the shadow property with respect to $Q$, it is easy to verify that all source vertices of the demand pairs in $\mset(Q)$ must belong to a single interesting interval of $\iset_1$. Let $I$ be that interval. Then we color the square $Q$ with the level-$1$ color $c_1(I)$ corresponding to the interval $I$. This completes the first iteration. Notice that for each level-1 color $c_1(I)$, at most $d_2/16$ demand pairs $(s,t)\in \mset'$ have $s\in I$. In the following iteration, we similarly partition every interesting level-$1$ interval into level-$2$ intervals that contain roughly $d_3/16\leq |\mset^{*}|/\eta^2$ source vertices of $\mset'$ each, and then define a coloring of all level-$2$ squares similarly, while suitably updating the set $\mset'$ of the demand pairs. We continue this process for $\rho$ iterations, eventually obtaining a good ensemble $(\ifamily, f, \mset')$. Since we only discard a constant fraction of the demand pairs of $\mset'$ in every iteration, at the end, $|\mset'|\geq |\mset^{**}|/2^{O(\rho)}=|\mset^{**}|/2^{O(\sqrt{\log n})}\geq |\mset^*|/\approxfactor$.

\subsection{Finding the Good Ensemble}
In our final step, our goal is to find a good ensemble $(\ifamily, f, \mset')$ maximizing $|\mset'|$. We show an efficient randomized $\approxfactor$-approximation algorithm for this problem. First, we show that, at the cost of losing a small factor in the approximation ratio, we can restrict our attention to a small collection $\ifamily_1,\ifamily_2,\ldots,\ifamily_z$ of hierarchical partitions of $R^*$ into intervals, and that it is enough to obtain a $\approxfactor$-approximate solution for the problem of finding the largest ensemble $(\ifamily_j,f,\mset')$  for each such partition $\ifamily_j$ separately.

We then fix one such hierarchical partition $\ifamily_j$, and design an LP-relaxation for the problem of computing a coloring $f$ of $\thset$ and a collection $\mset'$ of demand pairs, such that $(\ifamily_j,f,\mset')$ is a good ensemble, while maximizing $|\mset'|$. Finally, we design an efficient randomized LP-rounding $\approxfactor$-approximation algorithm for the problem.


\subsection{Completing the Proof of Theorem~\ref{thm: main}}
So far we have assumed that all source vertices lie on the top boundary of the grid, and all destination vertices are at a distance at least $\Omega(\opt)$ from the grid boundary. Let $\aset$ be the randomized efficient $\approxfactor$-approximation algorithm for this special case. We now extend it to the general \rNDPgrid problem. 
For every destination vertex $t$, we identify the closest vertex $\tilde{t}$ that lies on the grid boundary. Using standard grouping techniques, and at the cost of losing an additional $O(\log n)$ factor in the approximation ratio, we can assume that all source vertices lie on the top boundary of the grid, all vertices in $\set{\tilde t\mid t\in T(\mset)}$ lie on a single boundary edge of the grid (assume for simplicity that it is the bottom boundary), and that there is some integer $d$, such that for every destination vertex $t\in T(\mset)$, $d\leq d(t,\tilde t)<2d$. We show that we can define a collection $\zset=\set{Z_1,\ldots,Z_r}$ of disjoint square sub-grids of $G$, and a collection $\iset=\set{I_1,\ldots,I_r}$ of disjoint sub-intervals of $R^*$, such that the bottom boundary of each sub-grid $Z_i$ is contained in the bottom boundary of $G$, the top boundary of $Z_i$ is within distance at least $\opt$ from $R^*$, $Z_1,\ldots,Z_r$ appear in this left-to-right order in $G$, and $I_1,\ldots,I_r$ appear in this left-to-right order on $R^*$. For each $1\leq j\leq r$, we let $\mset_j$ denote the set of all demand pairs with the sources lying on $I_j$ and the destinations lying in $Z_j$. For each $1\leq j\leq r$, we then obtain a new instance $(G,\mset_j)$ of the \NDP problem. We show that there exist a collection $\zset$ of squares and a collection $\iset$ of intervals, such that the value of the optimal solution to each instance $(G,\mset_j)$, that we denote by $\opt_j$, is at most $d$, while $\sum_{j=1}^r\opt_j\geq \opt/\approxfactor$. Moreover, it is not hard to show that, if we can compute, for each $1\leq j\leq r$, a routing of some subset $\mset'_j\subseteq \mset_j$ of demand pairs in $G$, then we can also route all demand pairs in $\bigcup_{j=1}^r\mset'_j$ simultaneously in $G$.

There are two problems with this approach. First, we do not know the set $\zset$ of sub-grids of $G$ and the set $\iset$ of intervals of $R^*$. Second, it is not clear how to solve each resulting problem $(G,\mset_j)$. To address the latter problem, we define a simple mapping of all source vertices in $S(\mset_j)$ to the top boundary of grid $Z_j$, obtaining an instance of \rNDPgrid, where all source vertices lie on the top boundary of the grid $Z_j$, and all destination vertices lie at a distance at least $\opt_j\leq d$ from its boundary. We can then use algorithm $\aset$ in order to solve this problem efficiently. It is easy to see that, if we can route some subset $\mset'_j$ of the demand pairs via node-disjoint paths in $Z_j$, then we can extend this routing to the corresponding set of original demand pairs, whose sources lie on $R^*$.

Finally, we employ dynamic programming in order to find the set $\zset$ of sub-grids of $G$ and the set $\iset$ of intervals of $I$. For each such potential sub-grid $Z$ and interval $I$, we use algorithm $\aset$ in order to find a routing of a large set of demand pairs of the corresponding instance defined inside $Z$, and then exploit the resulting solution values for each such pair $(I,Z)$ in a simple dynamic program, that allows us to compute the set $\zset$ of sub-grids of $G$, the set $\iset$ of intervals of $I$, and the final routing.

\label{-------------------------------------sec: prelims -----------------------------------}
\section{Preliminaries}\label{sec:prelims}

All logarithms in this paper are to the base of $2$.

For a pair  $h,\ell> 0$ of integers, we let $G^{h,\ell}$ denote the grid of height $h$ and length $\ell$.
The set of its vertices is $V(G^{h,\ell})=\set{v(i,j)\mid 1\leq i\leq h, 1\leq j\leq \ell}$, and the set of its edges is the union of two subsets: the set $E^H=\set{(v_{i,j},v_{i,j+1})\mid 1\leq i\leq h, 1\leq j<\ell}$ of horizontal edges  and the set $E^{V}=\set{(v_{i,j},v_{i+1,j})\mid 1\leq i< h, 1\leq j\leq\ell}$ of vertical edges. The subgraph of $G^{h,\ell}$ induced by the edges of $E^H$ consists of $h$ paths, that we call the \emph{rows} of the grid; for $1\leq i\leq h$, the $i$th row $R_i$ is the row containing the vertex $v(i,1)$. Similarly, the subgraph induced by the edges of $E^V$ consists of $\ell$ paths that we call the \emph{columns} of the grid, and for $1\leq j\leq \ell$, the $j$th column $W_j$ is the column containing $v(1,j)$. We think of the rows  as ordered from top to bottom and the columns as ordered from left to right. Given two vertices $u=v(i,j)$ and $u'=v(i',j')$ of the grid, the shortest-path distance between them in $G$ is denoted by $d(u,u')$. Given two vertex subsets $X,Y\subseteq V(G^{\ell,h})$, the distance between them is $d(X,Y)=\min_{u\in X,u'\in Y}\set{d(u,u')}$. Given a vertex $v=v(i,j)$ of the grid, we denote by $\row(v)$ and $\col(v)$ the row and the column, respectively, that contain $v$. The \emph{boundary of the grid} is $\Gamma(G^{h,\ell})=R_1\cup R_h\cup W_1\cup W_{\ell}$. We sometimes refer to $R_1$ and $R_h$ as the top and the bottom boundaries of the grid respectively, and to $W_1$ and $W_{\ell}$ as the left and the right boundaries of the grid. We say that $G^{h,\ell}$ is a \emph{square grid} iff $h=\ell$.

Given a set $\rset$ of consecutive rows of a grid $G=G^{\ell,h}$ and a set $\wset$ of consecutive columns of $G$, we let $\Y_G(\rset,\wset)$ be the subgraph of $G$ induced by the set $\set{v(j,j')\mid R_j\in \rset, W_{j'}\in \wset}$ of vertices. We say that $\Y=\Y_G(\rset,\wset)$ is the \emph{sub-grid of $G$ spanned by the set $\rset$ of rows and the set $\wset$ of columns}. A sub-graph $G'\subseteq G$ is called a \emph{sub-grid of $G$} iff there is a set $\rset$ of consecutive rows and a set $\wset$ of consecutive columns of $G$, such that $G'=\Y_G(\rset,\wset)$. If additionally $|\rset|=|\wset|=d$ for some integer $d$, then we say that $G'$ is \emph{a square of size $(d\times d)$}.

In the \NDPgrid problem, the input is a grid $G=G^{\ell,\ell}$ and a set $\mset=\set{(s_1,t_1),\ldots,(s_k,t_k)}$ of pairs of its vertices, called demand pairs. We refer to the vertices in set $S(\mset)=\set{s_1,\ldots,s_k}$ as source vertices, to the vertices in set $T(\mset)=\set{t_1,\ldots,t_k}$ as destination vertices, and to the vertices of $S(\mset)\cup T(\mset)$ as terminals. The solution to the problem is a set $\pset$ of node-disjoint paths in $G$, where each path $P\in \pset$ connects some demand pair in $\mset$. The goal is to maximize the number of the demand pairs routed, that is, $|\pset|$. In this paper we consider a special case of \NDPgrid, that we call \restrictedNDP, where all source vertices appear on the boundary of the grid. 
We denote by $n=\ell^2$ the number of vertices in the input graph $G$.

Given a subset $\mset'\subseteq \mset$ of the demand pairs, we denote by $S(\mset')$ and $T(\mset')$ the sets of all source and all destination vertices participating in the pairs in $\mset'$, respectively. 


\label{-------------------------------------sec: alg-overview -----------------------------------}
\section{The Algorithm}\label{sec:alg-overview}
Throughout, we assume that we know the value $\opt$ of the optimal solution; in order to do so, we simply guess the value $\opt$ (that is, we go over all such possible values), and run our approximation algorithm for each such guessed value. It is enough to show that the algorithm returns a factor-$\approxfactor$-approximate solution whenever the value $\opt$ has been guessed correctly. We use a parameter $\eta=2^{\ceil{\sqrt{\log n}}}$.
We first further restrict the problem and consider a special case where the destination vertices appear far from the grid boundary. 
We show an efficient randomized algorithm for approximately solving this special case of the problem, by proving the following theorem, which is one of our main technical contributions. 

\begin{theorem}\label{thm: main w destinations far from boundary}
There is an efficient randomized algorithm $\aset$, that, given an instance $(G,\mset)$ of \restrictedNDP and an integer $\opt>0$, such that the value of the optimal solution to $(G,\mset)$ is at least $\opt$, and every destination vertex $t\in T(\mset)$ lies at a distance at least $\opt/\eta$ from $\Gamma(G)$,  returns a solution 
that routes at least $\opt/2^{O(\sqrt{\log n}\cdot\log\log n)}$ demand pairs, with high probability.
\end{theorem}

In Section~\ref{sec:destinations anywhere} we complete the proof of Theorem~\ref{thm: main} using algorithm $\aset$  from Theorem~\ref{thm: main w destinations far from boundary} as a subroutine, by reducing the given instance of \restrictedNDP to a number of instances of the form required by Theorem~\ref{thm: main w destinations far from boundary}, and then applying algorithm $\aset$ to each of them. In this section, we focus on the proof of Theorem~\ref{thm: main w destinations far from boundary}.

Recall that our input grid $G$ has dimensions $(\ell\times \ell)$.  Throughout, for an integer $r>0$, we denote $[r]=\set{1,\ldots,r}$.  Notice that by losing a factor $4$ in the approximation guarantee, we can assume w.l.o.g. that all sources lie on the top row of $G$, and that $\opt\leq \ell$.

\paragraph*{Parameters.}

The following parameters are used throughout the algorithm. Let $c^*\geq 2$ be a large constant, whose value will be set later. 
Recall that $\eta=2^{\ceil{\sqrt{\log n}}}$. Let $\rho$ be the largest integer, for which $\eta^{\rho+2}\leq \opt/2^{c^*\sqrt{\log n}\log\log n}$. 
Intuitively, we will route no more than $\eta^{\rho+2}$ demand pairs, and it is helpful that this number is significantly smaller than $\opt$. The parameter $\rho$ will be used as the number of recursive levels in the hierarchical system of squares, and in the algorithm overall.
Clearly, $\rho\leq \sqrt{\log n}$, and $\opt/(\eta\cdot 2^{c^*\sqrt{\log n}\log\log n})<\eta^{\rho+2}\leq \opt/2^{c^*\sqrt{\log n}\log\log n}$. In particular, $\eta^{\rho+2}<\ell$, where $\ell$ is the size of the side of the grid $G$.

It would be convenient for us if we could assume that $\ell$ is an integral multiple of $\eta^{\rho+2}$. Since this is not true in general, below we define a sub-grid $\tG$ of $G$ of dimensions $(\ell'\times \ell')$, where $\ell'$ is a large enough integral multiple of $\eta^{\rho+2}$. We show that $\tG$ can be defined so that, in a sense, we can restrict our attention to the sub-problem induced by $\tG$, without paying much in the approximation factor.

To this end, we let $\ell'$ be the largest integral multiple of $\eta^{\rho+2}$ with $\ell'< \ell$, so $\ell-\ell'\leq \eta^{\rho+2}\leq\opt/\eta$. Let $G'$ be the sub-grid of $G$ spanned by its first $\ell'$ columns and all its rows, and let $G''$ be the sub-grid of $G$ spanned by its last $\ell'$ columns  and all its rows. Finally, let $\mset'$ and $\mset''$ be the subsets of the demand pairs contained in $G'$ and $G''$, respectively. Notice that, since we have assumed that all destination vertices lie at a distance at least $\opt/\eta$ from $\Gamma(G)$, all vertices of $T(\mset)$ are contained in both $G'$ and $G''$.
Consider two instances of the \NDP problem, both of which are defined on the original graph $G$; the first one uses the set $\mset'$ of the demand pairs, and the second one uses the set $\mset''$ of demand pairs. Clearly, one of these instances has a solution of value at least $\opt/2$, as for each demand pair $(s,t)\in \mset$, both $s$ and $t$ must belong to either $\mset'$ or to $\mset''$ (or both). We assume without loss of generality that the problem induced by the set $\mset'$ of demand pairs has solution of value at least $\opt/2$. In particular, we assume that all source vertices lie on the first row of $G'$, that we denote by $R^*$ from now on. Notice however that we are only guaranteed that a routing of  at least $\opt/2$ demand pairs of $\mset'$ exists in the original graph $G$. For convenience, abusing the notation, from now on we  will denote $\mset'$ by $\mset$ and $\opt/2$ by $\opt$.

For each $1\leq h\leq \rho$, we let $d_h=\eta^{\rho-h+3}$, so that $d_1=\eta^{\rho+2}$ and $\opt/(\eta\cdot 2^{c^*\sqrt{\log n}\log\log n})<d_1\leq  \opt/2^{c^*\sqrt{\log n}\log\log n}$; $d_{\rho}=\eta^3$; and $d_h=d_{h-1}/\eta$ for all $h>1$.
Throughout, we assume that $n$ is large enough, so that, for example, $\log n>2^{27}$.

One of the main concepts that we use is that of a hierarchical decomposition of the grid into squares, that we define next.
Recall that $G'$ is the sub-grid of $G$ spanned by its first $\ell'$ columns and all its rows. We let $\tG\subseteq G'$ be the $(\ell'\times \ell')$-sub-grid of $G'$, that is spanned by all its columns and its $\ell'$ bottommost rows. Notice that since $\ell'<\ell$, the top row $R^*$ of $G'$, where the source vertices reside, is disjoint from $\tG$.

\label{-----------------------------------subsec: hierarchical system of squares-----------------------------}
\subsection{Hierarchical Systems of Squares}\label{subsec: hierarchical system of squares}

A subset $I\subseteq [\ell']$  of consecutive integers is called an \emph{interval}. We say that two intervals $I,I'$ are \emph{disjoint} iff $I\cap I'=\emptyset$, and we say that they are $d$-separated iff for every pair of integers $i\in I, j\in I'$, $|i-j|>d$. A collection $\iset$ of intervals of $[\ell']$ is called $d$-canonical, iff for each interval $I\in \iset$, $|I|= d$, and every pair of intervals in $\iset$ is $d$-separated.


Consider now the $(\ell'\times \ell')$-grid $\tG$.
Given two intervals $I,I'$ of $[\ell']$, we denote by $Q(I,I')$ the sub-graph of $\tG$ induced by all vertices in $\set{v(i,j)\mid i\in I, j\in I'}$. Given two sets $\iset,\iset'$ of intervals of $[\ell']$, we let $\qset(\iset,\iset')$ be the corresponding set of sub-graphs of $G'$, that is, $\qset(\iset,\iset')=\set{Q(I,I')\mid I\in \iset,I'\in \iset'}$.

\begin{definition}
Let $\qset$ be any collection of sub-graphs of $\tG$, and let $d$ be an integer. We say that $\qset$ is a \emph{$d$-canonical family of squares} iff there are two $d$-canonical sets $\iset,\iset'$ of intervals, such that $\qset=\qset(\iset,\iset')$. In particular, each sub-graph in $\qset$ is a square of size $(d\times d)$, and  for every pair $Q,Q'\in \qset$ of distinct squares, for every pair $v\in V(Q),v'\in V(Q')$ of vertices, $d(v,v')\geq d$.
\end{definition}

We are now ready to define a hierarchical system of squares. 
We use the integral parameters $\eta$ and $\rho$ defined above, as well as the parameters $\set{d_h}_{1\leq h\leq \rho}$.

\begin{definition} A hierarchical system of squares is a sequence $\thset=(\qset_1,\qset_2,\ldots,\qset_{\rho})$ of sets of squares, such that:

\begin{itemize}
\item for all $1\leq h\leq \rho$, $\qset_h$ is a $d_h$-canonical family of squares; and

\item for all $1<h\leq \rho$, for every square $Q\in \qset_h$, there is a square $Q'\in \qset_{h-1}$, such that $Q\subseteq Q'$.
\end{itemize}
\end{definition}

Given a hierarchical  system $\thset=(\qset_1,\qset_2,\ldots,\qset_{\rho})$, of squares  we denote by $V(\thset)$ the set of all vertices lying in the squares of $\qset_{\rho}$, that is, $V(\thset)=\bigcup_{Q\in \qset_{\rho}}V(Q)$. We say that the vertices of $V(\thset)$ \emph{belong} to $\thset$. We need the following simple claim, whose proof appears in the Appendix.

\begin{claim}\label{claim: hierarchical system of squares}
There is an efficient algorithm, that  constructs $4^{\rho}$ hierarchical systems $\thset_1,\ldots,\thset_{4^{\rho}}$ of squares of $\tG$, such that every vertex of $\tG$ belongs to exactly one system.
\end{claim}

Given a hierarchical system $\thset=(\qset_1,\ldots,\qset_{\rho})$  of squares and an integer $1\leq h\leq \rho$, we call the squares in $\qset_h$ \emph{level-$h$ squares}. If $h>1$, then for each level-$h$ square $Q$, we call the unique level-$(h-1)$ square $Q'$ with $Q\subseteq Q'$ the \emph{parent-square of $Q$}, and we say that $Q$ is a \emph{child-square} of $Q'$.

From now on, we fix the collection $\fset=\set{\thset_1,\ldots,\thset_{4^{\rho}}}$ of the hierarchical systems of squares of $\tG$ that is given by Claim~\ref{claim: hierarchical system of squares}. 
Consider an optimal solution $\pset^*$ to the \NDP instance $(G,\mset)$, and let $\mset^*\subseteq \mset$ be the set of the demand pairs routed by this solution. For each $1\leq i\leq 4^{\rho}$, we let $\mset_i^*\subseteq\mset^*$ be the set of all demand pairs whose destinations belong to the system $\thset_i$. Clearly, there is an index $i^*$, such that $|\mset^*_{i^*}|\geq |\mset^*|/4^{\rho}$. Even though we do not know the index $i^*$, we can run our approximation algorithm for each possible value of $i^*$. It is enough to show that our algorithm finds a $\approxfactor$-approximate solution to instance $(G,\mset)$ of \NDP if the index $i^*$ is guessed correctly. Therefore, we assume from now on that the index $i^*$ is known to us. For convenience, we denote $\thset_{i^*}$ by $\thset$ from now on. We also denote $\mset^*_{i^*}$ by $\mset^{**}$. Let $\tmset\subseteq \mset$ denote the set  of all demand pairs whose destination vertices lie in the system $\thset$. Then we are guaranteed that there is a solution to \NDP instance $(G,\tmset)$ of value at least $\opt/4^{\rho}\geq \opt/2^{2\sqrt{\log n}}$.

\label{-----------------------------------subsec: shadow property-----------------------------}
\subsection{The Shadow Property} \label{subsec: shadow property}

\begin{definition}
Let $Q\subseteq \tG$ be a square sub-grid of $\tG$ of size $(d\times d)$, and let $\hmset\subseteq \mset$ be any subset of the demand pairs, such that all vertices in $S(\hmset)$ are distinct. The \emph{shadow of $Q$ with respect to $\hmset$}, $J_{\hmset}(Q)$, is the shortest sub-path of row $R^*$ of $G'$ that contains the source vertices of all demand pairs $(s,t)\in \hmset$ with $t\in Q$. If no demand pairs in $\hmset$ have destinations in $Q$, then $J_{\hmset}(Q)=\emptyset$. The length of the shadow, $L_{\hmset}(Q)$, is the total number of vertices $s\in S(\hmset)$, that belong to $J_{\hmset}(Q)$. Given a parameter $\beta>0$, we say that $Q$ has the \emph{$\beta$-shadow property with respect to $\hmset$}, iff $L_{\hmset}(Q)\leq \beta \cdot d$. 
\end{definition}

We use the above definition in both the regimes where $\beta\leq 1$ and $\beta>1$.
We need the following simple observation, whose proof is deferred to the Appendix.

\begin{observation}\label{obs: boosting shadow}
Let $\qset$ be any collection of disjoint squares in $\tG$, let $0<\beta_1<\beta_2$ be parameters, and let $\hmset$ be any set of demand pairs, such that all vertices in $S(\hmset)$ are distinct, and every square in $\qset$ has the $\beta_2$-shadow property with respect to $\hmset$. Then there is an efficient algorithm to find a subset $\hmset'\subseteq \hmset$ of at least $\floor{\frac{\beta_1|\hmset|}{4\beta_2}}$ demand pairs, such that every square in $\qset$ has the $\beta_1$-shadow property with respect to $\hmset'$.
\end{observation}

The proof of the following theorem appears in Section~\ref{sec: shadow property}.

\begin{theorem}\label{thm: shadow property for one set of squares}
Let $\qset$ be a $d$-canonical set of squares, and let $\hmset\subseteq \mset$ be a subset of demand pairs, such that:
\begin{itemize}
\item there is a set $\pset$ of node-disjoint paths in $G$ routing all demand pairs in $\hmset$; and
\item for each demand pair $(s,t)\in \hmset$, there is a square $Q\in \qset$ with $t\in Q$.
\end{itemize}

Then there is a subset $\hmset'\subseteq \hmset$ of at least $|\hmset|/\log^6 n$ demand pairs, such that every square $Q\in \qset$ has the $1$-shadow property with respect to $\hmset'$.
\end{theorem}

Recall that $\fset$ is the family of hierarchical systems of squares of $\tG$ given by Claim~\ref{claim: hierarchical system of squares}, $\thset\in \fset$ is the hierarchical system that we have guessed, and $\tmset\subseteq \mset$ is the set of all demand pairs whose destinations belong to $\thset$. Recall also that we have assumed that the value of the optimal solution for instance $(G,\tmset)$ of \NDPgrid is at least $\opt/2^{2\sqrt{\log n}}$, and we have denoted by $\mset^{**}$ the set of the demand pairs routed by such a solution.

\begin{corollary}\label{cor: shadow property of all squares}
Denote $\thset=(\qset_1,\ldots,\qset_{\rho})$.
Then there is a set $\tmset^*\subseteq \tmset$ of at least $\frac{\opt}{2^{8\sqrt{\log n}\cdot\log\log n}}$ demand pairs, such that: (i)
 all demand pairs in $\tmset^*$ can be simultaneously routed via node-disjoint paths in $G$; and
(ii) for each $1\leq h\leq \rho$, all squares in set $\qset_h$ have the $1/\eta^2$-shadow property with respect to $\tmset^*$.
\end{corollary}

\begin{proof}
Notice that, since all demand pairs in $\mset^{**}$ are routable by node-disjoint paths, all vertices of $S(\mset^{**})$ are distinct.  

We perform $\rho$ iterations. The input to the $j$th iteration is a subset $\mset_{j-1}\subseteq \mset^{**}$ of demand pairs, where the input to the first iteration is $\mset_0= \mset^{**}$. In order to execute the $j$th iteration, we apply Theorem~\ref{thm: shadow property for one set of squares} to the set $\qset_j$ of squares and the set $\mset_{j-1}$ of demand pairs, to obtain a subset $\mset_j\subseteq \mset_{j-1}$ of demand pairs of cardinality at least $|\mset_{j-1}|/\log^6 n$, such that all squares in $\qset_j$ have the $1$-shadow property with respect to $\mset_j$. Consider the final set $\mset_{\rho}$ of demand pairs. Clearly, 

\[|\mset_{\rho}|\geq \frac{|\mset_0|}{\log^{6\rho}n}\geq\frac{\opt}{2^{2\sqrt{\log n}}\cdot 2^{6\sqrt{\log n}\cdot\log\log n}}\geq \frac{\opt}{2^{7\sqrt{\log n}\cdot\log\log n}},\]

 and for every $1\leq h\leq \rho$, every square in $\qset_h$ has the $1$-shadow property with respect to $\mset_{\rho}$. Applying Observation~\ref{obs: boosting shadow} to $\mset_{\rho}$ and the set $\qset=\bigcup_{h=1}^\rho\qset_h$, we obtain a subset $\tmset^*\subseteq \mset_{\rho}$ of at least $\floor{\frac{|\mset_{\rho}|}{4\eta^2}}\geq \frac{|\opt|}{8\eta^2\cdot 2^{7\sqrt{\log n}\cdot\log\log n}}\geq \frac{|\opt|}{2^{8\sqrt{\log n}\cdot\log\log n}}$ demand pairs, that has the required properties. 
%
\end{proof}

\subsection{Hierarchical Decomposition of Row $R^*$}\label{subsec: hierarchical partition of R1}
Recall that $R^*$ is the first row of the grid $G'$.
We start with intuition. Recall that $\tmset^*$ is the set of the demand pairs from Corollary~\ref{cor: shadow property of all squares}. In general, we would like to define a hierarchical partition $(\jset_1,\jset_2,\ldots,\jset_{\rho})$ of $R^*$, that has the following property: for each level $1\leq h\leq \rho$, every interval $I\in \jset_h$ contains either $0$, or roughly $d_h$ source vertices of the demand pairs in $\tmset^*$. Being able to find such a partition is crucial to our algorithm. If we knew the set $\tmset^*$ of demand pairs, finding such a partition would have been trivial. But unfortunately, set $\tmset^*$ depends on the optimal solution and is not known to us. We show instead that, at the cost of losing an additional $2^{O(\sqrt{\log n}\log\log n)}$-factor in the approximation ratio, we can define a small collection of such hierarchical partitions of $R^*$, one of which is guaranteed to have the above property. In each such hierarchical partition that we construct, for each level $1\leq h\leq \rho$, all intervals in $\jset_h$ will have the same length, that we denote by $\ell_h$. The lengths $\ell_h$ are not known to us a-priori, and differ from partition to partition, but we will be able to guess them from a small set of possibilities.

Suppose we are given a sequence $L=(\ell_1,\ell_2,\ldots,\ell_r)$ of integers, with $\ell_1>\ell_2>\ldots>\ell_r$, where $r>0$ is some integer, and for each $1\leq h\leq r$, $\ell_h$ is an integral power of $\eta$. A hierarchical $L$-decomposition of the first  row $R^*$ of the grid $G'$ is defined as follows. First, we partition $R^*$ into a collection $\jset_1$ of intervals, each of which contains  exactly $\ell_1$ consecutive vertices (recall that $R^*$ contains $\ell'$ vertices, and $\ell'$ is an integral multiple of $d_1=\eta^{\rho+2}$, which in turn is an integral power of $\eta$). We call the intervals in $\jset_1$ \emph{level-$1$ intervals}. Assume now that we are given, for some $1\leq h<r$, a partition $\jset_h$ of $R^*$ into level-$h$ intervals. We now define a partition $\jset_{h+1}$ of $R^*$ into level-$(h+1)$ intervals as follows. Start with $\jset_{h+1}=\emptyset$. For each level-$h$ interval $I\in \jset_h$, partition $I$ into consecutive intervals containing exactly $\ell_{h+1}$ vertices each. Add all resulting intervals to $\jset_{h+1}$. Note that for a fixed sequence $L=(\ell_1,\ell_2,\ldots,\ell_r)$ of integers, the corresponding hierarchical $L$-decomposition $(\jset_1,\ldots,\jset_r)$ of $R^*$ is unique.

Assume now that we are also given a collection $\mset'\subseteq \mset$ of demand pairs. We say that $\mset'$ is \emph{compatible} with the sequence $L=(\ell_1,\ldots,\ell_{\rho})$ of integers, iff for each $1\leq h\leq \rho$, for every interval $I\in \jset_h$, either $I\cap S(\mset')=\emptyset$, or:

\[\frac{d_h}{16\eta}\leq |I\cap S(\mset')|\leq \frac{d_h}{4}.\]

The following theorem uses the hierarchical system $\thset\in \fset$ of squares we have defined above and the set $\tmset^*\subseteq \tmset$ of demand pairs from Corollary~\ref{cor: shadow property of all squares}. Recall that all demand pairs in $\tmset^*$ have their destinations in $\thset$, and they can all be routed via node-disjoint pairs in $G$. Moreover, all squares that lie in the system $\thset$ have the $1/\eta^2$-shadow property with respect to $\tmset^*$.

\begin{theorem}\label{thm: partition of R1}
There is a sequence $L=(\ell_1,\ell_2,\ldots,\ell_{\rho})$ of integers, with $\ell_1>\ell_2>\ldots>\ell_{\rho}$, where each $\ell_i$ is an integral power of $\eta$, and a subset $\tmset^{**}\subseteq \tmset^*$ of demand pairs, such that:

\begin{itemize}
\item $|\tmset^{**}|\geq \opt/2^{O(\sqrt{\log n}\log \log n)}$;
\item the set $\tmset^{**}$ of demand pairs is compatible with the sequence $L$.
\end{itemize}
\end{theorem}

\begin{proof}
Let $\rho'$ be the largest integer such that $\ell'\geq \eta^{\rho'}$, and note that $\rho'\leq \sqrt{\log n}$.
We next gradually modify the set $\tmset^*$ of demand pairs, by starting with $\mset_{\rho'}=\tmset^*$, and then gradually obtaining smaller and smaller subsets $\mset_{\rho'-1},\mset_{\rho'-2},\ldots,\mset_0$ of $\mset_{\rho'}$. Along the way, we will also define a sequence $\gamma_{\rho'},\ldots,\gamma_1$ of integers that will be helpful for us later.

Consider first the sequence $L'=(\ell_1',\ell_2',\ldots,\ell_{\rho'})$ of integers, where for $1\leq h\leq \rho'$, $\ell_h'=\eta^{\rho'-h}$. Let $(\jset_1',\jset_2',\ldots,\jset'_{\rho'})$ be the corresponding hierarchical $L'$-decomposition of $R^*$. We start with $\mset_{\rho'}=\tmset^*$, and then perform $\rho'$ iterations. An input to iteration $j$ is some subset $\mset_{\rho'-j+1}$ of $\tmset^*$. In order to execute iteration $j$,  we consider the set $\jset_{\rho'-j+1}$ of all level-$(\rho'-j+1)$ intervals of $R^*$ in the $L'$-hierarchical decomposition of $R^*$. We partition these intervals into $\ceil{\log n}$ groups: an interval $I\in \jset_{\rho-j+1}$ belongs to group $U_z$ iff:

\[2^z<|I\cap S(\mset_{\rho'-j+1})| \leq 2^{z+1}.\]

We say that a demand pair $(s,t)$ belongs to group $U_z$ iff the interval $I\in \jset_{\rho'-j+1}$ containing $s$ belongs to group $U_z$. Then there is an index $z$, such that $U_z$ contains at least a $(1/\ceil{\log n})$-fraction of the demand pairs of $\mset_{\rho'-j+1}$. We then set $\gamma_{\rho'-j+1}=2^z$, and we let $\mset_{\rho'-j}$ contain the set of all demand pairs of  $\mset_{\rho'-j+1}$ that belong to $U_z$. Let $\mset_0$ be the set of demand pairs obtained after the last iteration of the algorithm. We then set $\tmset^{**}=\mset_0$. It is easy to verify that:

 \[|\tmset^{**}|\geq \frac{|\tmset^*|}{\ceil{\log n}^{\rho'}}\geq \frac{\opt}{2^{8\sqrt{\log n}\cdot\log\log n}\cdot \ceil{\log n}^{\sqrt{\log n}}}\geq \frac{\opt}{2^{10\sqrt{\log n}\cdot\log\log n}}.\]

Recall that we have defined integers $\gamma_1,\gamma_2,\ldots,\gamma_{\rho'}$, and for each $1\leq h\leq \rho'$, for every level-$h$ interval $I_h\in \jset'_h$, either $I_h$ contains no vertices of $S(\tmset^{**})$, or:

\[\gamma_h<|I\cap S(\tmset^{**})|\leq 2\gamma_h.\]

Since for each $1< h\leq \rho'$, a level-$h$ interval is a sub-interval of a level-$(h-1)$ interval, it is easy to verify that $\gamma_1\geq \gamma_2\geq \cdots\geq \gamma_{\rho'}$. Moreover, we claim that for each  $1< h\leq \rho'$, $\gamma_h\geq \gamma_{h-1}/(2\eta)$. This is since a level-$(h-1)$ interval $I$ is partitioned into exactly $\eta$ level-$h$ intervals. If $\gamma_h< \gamma_{h-1}/(2\eta)$, then each such interval contributes fewer than $\gamma_{h-1}/\eta$ source vertices to 
$I$, and so $I$ contains fewer than $\gamma_h$ vertices of $S(\tmset)$, a contradiction.
It is also easy to verify that $\gamma_{\rho'}=1$, while, since $|\iset'_1|\leq \eta^2$, $\gamma_1\geq |\tmset^*|/(\eta^2 \ceil{\log n})\geq \opt/(\ceil{\log n}\cdot 2^{9\sqrt{\log n}\cdot\log\log n})$. We set the parameter $c^*$, that was used in the definition of the parameter $\rho$, so that $\gamma_1\geq d_1$ holds. For example, $c^*\geq 11$ is sufficient.

We are now ready to define the sequence $L=(\ell_1,\ldots,\ell_{\rho})$ of integers. For each $1\leq j\leq \rho$, we let $\ell_j$ be the largest integer $\ell_h'\in L'$, for which $\gamma_h\leq d_j/8$. From the above discussion, such an integer exists for all $1\leq j\leq \rho$, and we have that $d_j/(16\eta)\leq \gamma_h\leq d_j/8$. In particular for each interval $I\in \jset_h$, if $I\cap S(\tmset^{**})\neq\emptyset$, then $d_j/(16\eta)\leq |I\cap S(\tmset^{**})|\leq d_j/4$. 

It is now easy to verify that $\tmset^{**}$ is compatible with the sequence $L$, since the sets  of intervals of the $L$-hierarchical partition $(\jset_1,\jset_2,\ldots,\jset_{\rho})$ all belong to the collection $(\jset'_1,\jset_2',\ldots,\jset_{\rho'}')$ corresponding to the $L'$-hierarchical partition, as $L\subseteq L'$.
\end{proof}

For our algorithm, we need to assume that we know the sequence $L$ of integers given by Theorem~\ref{thm: partition of R1}. As sequence $L$ consists of at most $\sqrt{\log n}$ integers, each of which is an integral power of $\eta$, there are at most $\sqrt{\log n}$ possible choices for each integer of $L$, and at most $(\sqrt{\log n})^{\sqrt{\log n}}$ possible choices for the sequence $L$. Therefore, we can go over all such choices, and apply our algorithm to each of them separately. It is now sufficient to show that our algorithm succeeds in producing the right output when the sequence $L$ is guessed correctly --- that is, it has the properties guaranteed by Theorem~\ref{thm: partition of R1}. We assume from now on that the correct sequence $L=(\ell_1,\ldots,\ell_{\rho})$ is given, and we denote by $(\jset_1,\ldots,\jset_{\rho})$ the corresponding $L$-hierarchical decomposition of $R^*$, that can be computed efficiently given $L$. Intervals lying in set $\jset_h$ are called \emph{level-$h$ intervals}. Recall that we have also assumed that we are given a lower bound $\opt$ on the value of the optimal solution, and that we have guessed  the hierarchical system $\thset\in \fset$ of squares correctly. We denote $\thset=(\qset_1,\ldots,\qset_{\rho})$. Recall also that $\tmset\subseteq \mset$ is the set of all demand pairs whose destinations lie in $\thset$.
\subsection{Hierarchical System of Colors}\label{subsec: hierarchical squares}

For every level $1\leq h\leq \rho$, for every level-$h$ interval $I\in \jset_h$, we introduce a distinct level-$h$ color $c_h(I)$, and we color all vertices of $I$ with this color.  We let $\chi_h$ be the set of all level-$h$ colors, so $\chi_h=\set{c_h(I)\mid I\in \jset_h}$. Every vertex of $R^*$ is now assigned $\rho$ colors - one color for every level. For convenience, we define a set $\chi_0$ of level-$0$ colors, that contains a unique color $c_0$.  This naturally defines a hierarchical structure of the colors. Consider any level $0\leq h<\rho$, and some level-$h$ interval $I\in \jset_h$, with the corresponding color $c_h(I)$. For every level-$(h+1)$ interval $I'$ contained in $I$, we say that the corresponding color $c_{h+1}(I')$ is the child-color of $c_h(I)$. This also naturally defines a descendant-ancestor relation between colors: for $1\leq h\leq h'\leq \rho$, we say that $c_h(I)$ is an ancestor-color of $c_{h'}(I')$ iff $I'\subseteq I$. In particular, each color is an ancestor-color of itself. If a color $c$ is an ancestor-color of color $c'$, then we say that $c'$ is a descendant-color of $c$. We let $\tilde{\chi}(c)\subseteq \chi$ denote the set of all descendant-colors of $c$, and by $\tilde{\chi}_h(c)$ the set of all descendants of $c$ that belong to level $h$. 

Given the hierarchical system of colors  $\cset=(\chi_0,\chi_1,\ldots,\chi_{\rho})$ as above, a coloring of $\thset$ by $\cset$ is an assignment of colors to squares, $f:\bigcup_{i=1}^{\rho}\qset_i\rightarrow\bigcup_{i=1}^{\rho}\chi_i$, such that:

\begin{itemize}
\item For each $1\leq h\leq \rho$, every level-$h$ square $Q\in \qset_h$ is assigned a single  level-$h$ color $c\in \chi_h$; and 

\item For all $1<h\leq \rho$, if $Q\in \qset_h$ is a level-$h$ square, that is assigned a level-$h$ color $c$, and its parent-square $Q'\in \qset_{h-1}$ is assigned a level-$(h-1)$ color $c'$, then $c$ is a child-color of $c'$.
\end{itemize}

Let $U$ denote the set of all vertices of $\tilde G$ that belong to the hierarchical system $\thset$ of squares, that is, $U=\bigcup_{Q\in \qset_{\rho}}V(Q)$.
Recall that $R^*\cap U=\emptyset$. Consider any valid assignment of colors, and let $v\in U$ be some vertex. For $1\leq h\leq \rho$, let $Q_h$ be the level-$h$ square containing $v$, and let $c_h$ be the color assigned to $Q_h$. Then we say that the level-$h$ color of $v$ is $c_h$. We also say that the level-$0$ color of $v$ is $c_0$. Therefore, every vertex $v\in U(\hset)\cup R^*$ is associated with a $(\rho+1)$-tuple of colors $(c_0(v),\ldots,c_{\rho}(v))$, where $c_h(v)$ is its level-$h$ color  (recall that we have already assigned colors to the vertices of $R^*$). Moreover, for $0\leq h<\rho$, $c_{h+1}(v)$ is a child-color of $c_h(v)$.

Assume now that we are given some coloring $f$ of $\thset$ by $\cset$, and some subset $\tmset'\subseteq \tmset$ of demand pairs. We say that $\tmset'$ is a \emph{perfect set of demand pairs for $f$}, iff:

\begin{itemize}

\item All sources and all destinations of the pairs in $\tmset'$ are distinct; 
\item For each demand pair $(s,t)\in \tmset'$, both $s$ and $t$ are assigned the same level-$\rho$ color (notice that this means that for every level $1\leq h\leq \rho$, the level-$h$ color assigned to $s$ and $t$ is the same); and

\item For every level $1\leq h\leq \rho$, for every level-$h$ color $c_h$, the number of demand pairs  $(s,t)\in \tmset'$, for which $s$ and $t$ are assigned the color $c_h$ is at most $d_h$.
\end{itemize}

The rest of the argument consists of three parts. First, we prove that there exists a coloring $f$ of $\thset$, and a perfect set $\tmset'\subseteq \tmset$ of demand pairs for $f$, so that $|\tmset'|$ is quite large compared to $\opt$; set $\tmset'$ is obtained by appropriately modifying the set $\tmset^{**}$ of demand pairs given by Theorem~\ref{thm: partition of R1}. Next, we show that, given any coloring $f$ of $\thset$ by $\cset$ and any perfect set $\tmset'$ of demand pairs for $f$, we can efficiently route a large fraction of the demand pairs in $\tmset'$ in graph $G$. Finally, we show an efficient  algorithm for (approximately) computing the largest set $\tmset'\subseteq \tmset$ of demand pairs, together with a coloring $f$ of $\thset$ by $\cset$, such that $\tmset'$ is perfect for $f$. The following three theorems summarize these three steps, and their proofs appear in the following sections.

\begin{theorem}\label{thm: existence of perfect set}
There is a coloring $f$ of $\thset$ by $\cset$, and a set $\tmset'\subseteq\tmset$ of demand pairs, such that $\tmset'$ is perfect for $f$, and $|\tmset'|\geq \opt/2^{O(\sqrt{\log n}\cdot \log\log n)}$.
\end{theorem}

\begin{theorem}\label{thm: snake-like routing} There is an efficient algorithm, that, given a coloring $f$ of $\thset$ by $\cset$, and a set $\tmset'\subseteq \tmset$ of demand pairs with $|\tmset'|\leq d_1$, such that $\tmset'$ is perfect for $f$, routes at least $|\tmset'|/2^{O(\sqrt{\log n})}$ demand pairs of $\tmset'$ in $G$. 
\end{theorem}

\begin{theorem}\label{thm: finding the coloring}
There is an efficient randomized algorithm, that, given:

\begin{itemize}
\item a hierarchical system $\thset=(\qset_1,\ldots,\qset_{\rho})\in \fset$ of squares, together with a subset $\tmset\subseteq \mset$ of demand pairs whose destinations belong to $\thset$; and
\item a sequence $L=(\ell_1,\ldots,\ell_{\rho})$ of integers, with $\ell_1\geq \ell_2\geq \cdots\geq \ell_{\rho}$, such that each integer $\ell_h$ is an integral power of $\eta$, together with the corresponding $L$-hierarchical decomposition of $R^*$ and the corresponding  hierarchical system  $\cset$ of colors,
\end{itemize}

computes a coloring $f$ of $\thset$ by $\cset$ and a set $\tmset'\subseteq \tmset$ of demand pairs that is perfect for $f$, such that with high probability, $|\tmset'|\geq \opt'/2^{O(\sqrt{\log n})}$, where $\opt'$ is the cardinality of the largest set $\tmset''\subseteq \tmset$ of demand pairs, such that there is some coloring $f'$ of $\thset$ by $\cset$, for which $\tmset''$ is a perfect set.
\end{theorem}

It is now easy to complete the proof of Theorem~\ref{thm: main w destinations far from boundary}. We assume that our algorithm is given a lower bound $\opt$ on the value of the optimal solution. It starts by constructing the family $\fset$ of hierarchical systems of squares using Claim~\ref{claim: hierarchical system of squares}. It then guesses one of the systems $\thset\in \fset$, and a sequence $L=(\ell_1,\ldots,\ell_{\rho})$ of integers as described above. We can now efficiently compute the $L$-hierarchical decomposition $(\jset_1,\ldots,\jset_{\rho})$ of row $R^*$ of $G'$, and the corresponding hierarchical system $\cset$ of colors. We apply Theorem~\ref{thm: finding the coloring} in order to compute a coloring $f$ of $\thset$ by $\cset$ and the corresponding set $\tmset'\subseteq \tmset$ of demand pairs that is perfect for $f$. From Theorems~\ref{thm: existence of perfect set} and~\ref{thm: finding the coloring}, if $\thset$ and $L$ were guessed correctly, then with high probability, $|\tmset'|\geq \opt/\approxfactor$. If $|\tmset'|>d_1$, then we discard demand pairs from $\tmset'$ until $|\tmset'|=d_1$ holds. As $d_1=\eta^{\rho+2}\geq \opt/\approxfactor$, we still have that $|\tmset'|\geq \opt/\approxfactor$.
Finally, we use Theorem~\ref{thm: snake-like routing} in order to route at least $|\tmset'|/2^{O(\sqrt{\log n})}\geq \opt/\approxfactor$ demand pairs of $\tmset'$ in $G$.  We now turn to prove Theorem~\ref{thm: shadow property for one set of squares} and Theorems~\ref{thm: existence of perfect set} --- \ref{thm: finding the coloring}.

\label{---------------------------sec: shadow property-------------------------------}
\section{The Shadow Property}\label{sec: shadow property}
This section is dedicated to proving Theorem~\ref{thm: shadow property for one set of squares}. 

  
For every demand pair $(s,t)\in \hmset$, we let $P(s,t)$ denote the path routing this pair in $\pset$. We assume that $G$ is embedded into the plane in the natural manner, and we construct a number of discs in the plane. Whenever we refer to a disc, we mean a simple disc, whose boundary is a simple cycle. Let $D^0$ denote the disc whose outer boundary is the boundary of the grid $G$. 

\begin{definition}
We say that a disc $D$ in the plane is \emph{canonical}, if either (i)  $D=D^0$, or (ii) we can partition the boundary $\sigma(D)$ of $D$ into four contiguous segments, $\sigma_1(D),\sigma_2(D),\sigma_3(D),\sigma_4(D)$, such that $\sigma_1(D)\subseteq R^*$; $\sigma_3(D)$ is contained in the boundary of some square $Q\in \qset$; $\sigma_2(D)$ is contained in some path $P(s,t)\in \pset$, and $\sigma_4(D)$ is contained in some path $P(s',t')\in \pset$ (where possibly $(s,t)=(s',t')$).
\end{definition}

\begin{definition}
We say that a collection $\dset$ of canonical discs is \emph{nested} iff for every pair $D,D'\in \dset$, either one of the discs is contained in the other, or they are disjoint. Moreover, if the boundaries of two distinct discs $D,D'\in \dset$ intersect, then $\sigma(D)\cap \sigma(D')\subseteq R^*$ must hold.
\end{definition}

 Assume now that we are given a nested collection $\dset$ of canonical discs. Consider some disc $D\in \dset$, and let $\dset'\subseteq \dset$ be the set of all discs $D'$ with $D'\subseteq D$. We define the region of the plane associated with $D$ as $R_{\dset}(D)=D\setminus (\bigcup_{D'\in \dset'}D')$.

The idea is to compute a collection $\dset=\set{D^0,D^1,\ldots,D^z}$ of nested canonical discs, and for every disc $D^i$ a set $\mset_i\subseteq \hmset$ of demand pairs, such that all demand pairs in $\mset_i$ are routed inside the region $R_{\dset}(D^i)$ by $\pset$. Clearly, for $i\neq i'$, $\mset_i\cap \mset_{i'}=\emptyset$ will hold. We denote by $\overline{\mset}=\hmset\setminus \bigcup_i\mset_i$ the set of the demand pairs we discard, and we will ensure that $|\overline{\mset}|$ is small enough compared to $|\hmset|$. 

Given a set $\overline{\mset}$ of discarded demand pairs, we say that a square $Q\in \qset$ is \emph{active} iff at least one demand pair in $\hmset\setminus\overline{\mset}$ has a destination vertex lying in $Q$. We say that the set $\dset$ of discs \emph{respects} the active squares, iff for every active square $Q\in \qset$, there is a unique disc $D^i\in \dset$ with $Q\subseteq R_{\dset}(D^i)$. We will ensure that the set $\dset$ of discs we construct respects all active squares.

\paragraph*{Heavy demand pairs.}
Consider any demand pair $(s,t)\in \hmset$ and the path $P(s,t)$ routing $(s,t)$ in $\pset$. We say that $(s,t)$ \emph{covers} the demand pair $(s',t')\in \hmset$ iff $P(s,t)$ intersects the unique square $Q\in \qset$ with $t'\in Q$. Notice that $(s',t')$ may be covered by at most $4d$ demand pairs. This is since $t'$ belongs to exactly one square of $\qset$; for every pair $(s,t)$ covering $(s',t')$, path $P(s,t)$ must contain a vertex from the boundary of that square; and the length of the boundary of each such square is at most $4d$. We say that demand pair $(s,t)$ is \emph{heavy} if it covers more than $2^{9}d\log n$ other pairs in $\hmset$.

\begin{observation}
The number of heavy demand pairs is at most $|\hmset|/(2^7\log n)$.
\end{observation}
\begin{proof}
Assume otherwise. Then the total number of demand pairs covered by heavy pairs (counting multiplicities) is more than $4|\hmset|d$. But every demand pair can be covered by at most $4d$ demand pairs, a contradiction.
\end{proof}

\paragraph*{Partitioning algorithm.}
Initially, we let $\overline{\mset}$ contain all heavy demand pairs, so $|\overline{\mset}|\leq |\hmset|/(128\log n)$. We start with $\dset=\set{D^0}$, and let the corresponding set of the demand pairs be $\mset_0=\hmset\setminus \overline{\mset}$. We then iterate. Let $D^i\in \dset$ be any disc and let $\mset_i$ be the corresponding set of demand pairs. Consider any pair $(s,t),(s',t')\in \mset_i$ of demand pairs. Let $N_{\mset_i}(t,t')$ be the number of the demand pairs $(s'',t'')\in \mset_i$ whose source lies between $s$ and $s'$ on $R^*$. Let $N'_{\mset_i}(t,t')$ denote the number of the remaining demand pairs in $\mset_i$ -- those pairs whose sources do not lie between $s$ and $s'$ on $R^*$.

\begin{definition}
 We say that $(s,t),(s',t')$ is a \emph{candidate pair for $\mset_i$} iff:

\begin{itemize}
\item there is a square $Q\in \qset$ with $t,t'\in Q$; and

\item both $N_{\mset_i}(t,t'),N'_{\mset_i}(t,t')>2^{20}d\log^3n$.
\end{itemize}
\end{definition}

If a candidate pair $(s,t),(s',t')$ exists in $\mset_i$, then we do the following.
Let $Q\in \qset$ be the square that contains both $t$ and $t'$.  We define a new canonical disc $D'\subseteq D^i$. Let $v$ be the first vertex on path $P(s,t)$ that lies in $Q$ (we view the path as directed from $s$ to $t$), and let $P'(s,t)$ be the sub-path of $P(s,t)$ between $s$ and $v$. Similarly, let $v'$ be the first vertex on path $P(s',t')$ that lies in $Q$, and let $P'(s',t')$ be the sub-path of $P(s',t')$ between $s'$ and $v'$. We let $\sigma_3(D')$ be a path connecting $v$ to $v'$, such that $\sigma_3(D')$ is contained in the boundary of $Q$, and the length of $\sigma_3(D')$ is at most $2d$.  Let $\sigma_2(D')=P'(s,t)$, $\sigma_4(D')=P'(s',t')$, and let $\sigma_1(D')$ be the subpath of $R^*$ between $s$ and $s'$. Finally, let $\sigma(D')$ be the concatenation of these four curves, and let $D'$ be the disc whose boundary is $\sigma(D')$. Let $\overline{\mset}_i\subseteq\mset_i$ contain all demand pairs $(s'',t'')$, such that either (i) $P(s'',t'')$ crosses $\sigma(D')$; or (ii) $(s,t)$ covers $(s'',t'')$; or (iii) $(s',t')$ covers $(s'',t'')$. Recall that there are at most $2d$ pairs of the first type (as all their paths must cross $\sigma_3(D')$), and at most $2^{9}d\log n$ pairs of each of the remaining types (as we have discarded all heavy demand pairs).
 Altogether, $|\overline{\mset}_i|\leq 2d+2^{10}d\log n$.
  Notice that $\overline{\mset}_i$ contains the pairs $(s,t)$ and $(s',t')$, and also all pairs whose destinations lie in the squares of $\qset$ that $\sigma(D')$ visits, including the pairs whose destinations lie in $Q$.

We  create two new subsets of the demand pairs: $\mset'\subseteq \mset_i\setminus \overline{\mset}_i$ contains all demand pairs whose corresponding path in $\pset$ is contained in the new disc $D'$, and we update $\mset_i$ to contain all remaining demand pairs of the original set $\mset_i$ (excluding the pairs in $\mset'$ and $\overline{\mset}_i$). Therefore, for every demand pair in the new set $\mset_i$, the path routing the pair in $\pset$ is contained in $R_{\dset}(D^i)\setminus D'$. We then add $D'$ to $\dset$, with $\mset'$ as the corresponding set of the demand pairs, and we add all demand pairs in $\overline{\mset}_i$ to $\overline {\mset}$. It is easy to verify that the resulting set of discs remains nested and canonical. From the above discussion, it also respects the currently active squares.

We have discarded at most $2d+2^{10}d\log n\leq 2^{11}d\log n$ demand pairs in the current iteration, but both $\mset'$ and the new set $\mset_i$ contain at least $2^{20}d\log^3n-2^{10}d\log n -2d\geq 2^{19}d\log^3n$ demand pairs (after excluding the discarded demand pairs).

We continue this partitioning procedure, as long as there is some disc $D^i\in \dset$, whose corresponding set $\mset_i$ of demand pairs contains a candidate pair $(s,t),(s',t')$. Let $\dset=\set{D^0,\ldots,D^z}$ be the final set of discs, and for each $i$, let $\mset_i$ be the corresponding set of the demand pairs.

\begin{claim}\label{claim: charging}
At the end of the algorithm,  $|\overline{\mset}|< |\hmset|/(64\log n)$.
\end{claim}

\begin{proof}
Recall that at the beginning of the algorithm, $|\overline{\mset}|\leq |\hmset|/(128\log n)$ held. Over the course of the algorithm, for every demand pair $(s,t)\in \hmset\setminus \overline{\mset}$, we define a budget of $(s,t)$ as follows. Assume that $(s,t)$ belongs to some set $\mset_i$, corresponding to some disc $D^i\in \dset$. We set the budget of $(s,t)$ to be $B(s,t)=\frac{\log(|\mset_i|)}{128\log^2 n}$. Notice that this budget may change as the algorithm progresses. At the beginning of the algorithm, the total budget of all demand pairs is at most $\frac{|\hmset|}{128\log n}$, while $|\overline{\mset}|\leq \frac{|\hmset|}{128\log n}$. It is now enough to prove that throughout the algorithm, the total budget of all demand pairs plus the cardinality of $\overline{\mset}$ do not increase.

Consider some iteration of the algorithm, where we processed some set $\mset_i$ of demand pairs as above, adding a new set $\overline{\mset}_i$ of demand pairs to $\overline{\mset}$, and creating two new sets of demand pairs: $\mset'$ and $\mset'_i$ (that replaces $\mset_i$). Recall that we have guaranteed that $|\overline{\mset}_i|\leq 2^{11}d\log n$, while $|\mset'|,|\mset'_i|\geq 2^{19}d\log^3n$. Assume without loss of generality that $|\mset'|\leq |\mset'_i|$. Then the budget of every demand pair in $\mset'$ decreases by at least $\frac{1}{128\log^2n}$ (as $|\mset'|\leq |\mset_i|/2$), and the total decrease in the budgets is therefore at least $\frac{|\mset'|}{128\log^2n}\geq \frac{2^{19}d\log^3n}{128\log^2n}\geq 2^{12}d\log n$, while $|\overline{\mset}|$ increases by at most $2^{11}d\log n$. Therefore, the total budget of all demand pairs plus $|\overline{\mset}|$ do not increase.
\end{proof}

From the above discussion, the final set $\dset$ of discs is canonical, nested, and respects all remaining active squares.

\paragraph*{Boundary Pairs.}
Consider some disc $D^i\in \dset$, and assume first that $|\mset_i|\geq 2^{21}d\log^3n$.
Let $a_i,b_i$ be the two endpoints of $\sigma_1(D^i)$, and let $\zset^i_0\subseteq \mset_i$ contain $2^{20}d\log^3n$ demand pairs whose sources lie closest to $a_i$ and $2^{20}d\log^3n$ demand pairs whose sources lie closest to $b_i$. If $|\mset_i|< 2^{21}d\log^3n$, then we set $\zset^i_0=\mset_i$.
In either case, $|\zset^i_0|\leq 2^{21}d\log^3n$. We call the demand pairs in $\zset^i_0$ \emph{boundary pairs}. If $|\mset_i|>2^{23}d\log^4n$, then we discard the boundary pairs from $\mset_i$, and add them to $\overline{\mset}$. After this procedure, we are still guaranteed that $|\overline{\mset}|\leq \frac{|\hmset|}{64\log n}+\frac{|\hmset|}{4\log n}\leq \frac{|\hmset|}{2\log n}$. So we assume from now on that whenever the boundary pairs are present in $\mset_i$, $|\mset_i|\leq 2^{23}d\log^4n$.

\paragraph*{Forest of Discs.}
It would be convenient for us to organize the discs in $\dset$ into a directed forest $F$, in a natural way. In each arborescence of the forest, the edges will be directed towards the root.
The set of vertices of $F$ is $v_0,v_1,\ldots,v_z$, where $v_i$ represents disc $D^i$. There is a directed edge $(v_i,v_{i'})$ iff $D^i\subsetneq D^{i'}$, and there is no disc $D^x\in \dset$ with $D^i\subsetneq D^x\subsetneq D^{i'}$. Since the discs in $\dset$ are nested, it is immediate to verify that $F$ is indeed a directed forest. For every arborescence $\tau$ of $F$, the root of $\tau$ is the vertex $v_i$ corresponding to a disc $D^i$ that has the largest area. For every vertex $v_j\in V(F)$, we let its weight $w_j$ be $|\mset_j|$.
We denote by $\mset'=\bigcup_{i=0}^z \mset_i$, so $|{\mset'}|\geq |\hmset|/2$.

We use the following well-known result. For completeness, its proof appears in Appendix.

\begin{claim}\label{claim: partition the forest} There is an efficient algorithm, that, given a directed forest $F$ with $n$ vertices, computes a partition $\yset=\set{Y_1,\ldots,Y_{\ceil{\log n}}}$ of $V(F)$ into subsets, such that for each $1\leq j\leq \ceil{\log n}$, $F[Y_j]$ is a collection of disjoint directed paths, that we denote by $\pset_j$. Moreover, for all $v,v'\in Y_j$, if there is a directed path from $v$ to $v'$ in $F$, then they both lie on the same path in $\pset_j$.
\end{claim}

Let $\set{Y_1,\ldots,Y_{\ceil{\log z}}}$ be the partition of $V(F)$ given by Claim~\ref{claim: partition the forest}. Then there is some index $1\leq x\leq \ceil{\log z}$, such that the total weight of all vertices in $Y_x$ is at least $|{\mset'}|/\ceil{\log z}\geq |\hmset|/(2\log n)$. 
Let $\dset'\subseteq \dset$ be the set of all discs $D^i$, whose corresponding vertex $v_i\in Y_x$, and let $\mset''\subseteq \mset'$ be the set of all demand pairs that belong to all such discs, that is, $\mset''=\bigcup_{D^i\in\dset'}\mset_i$. Consider now some vertex $v_i\in Y_x$. Notice that at most one child-vertex $v_{i'}$ of $v_i$ in the forest $F$ may belong to $Y_x$. If such a vertex $v_{i'}$ does not exist, then all demand pairs in $\mset_i$ are called \emph{right pairs}. Otherwise, for each demand pair $(s,t)\in \mset_i$, if $s$ appears to the left of $\sigma_1(\mset_{i'})$ on $R^*$, then we call it a \emph{left pair}, and otherwise we call it a right pair. Therefore, all demand pairs in $\mset''$ are now partitioned into right pairs and left pairs. We assume without loss of generality that at least half of the demand pairs in $\mset''$ are left pairs; the other case is symmetric. Let $\mset'''\subseteq \mset''$ be the set of all left pairs, so $|\mset'''|\geq |\mset''|/2\geq |\hmset|/(4\log n)$.

\begin{claim}\label{claim: the final shadow property}
The squares in $\qset$ have the $r$-shadow property with respect to $\mset'''$, for $r=2^{23}\log^4n$.
\end{claim}

Notice that, assuming the claim is correct, from Observation~\ref{obs: boosting shadow}, there is a subset $\hmset'\subseteq \mset'''$ of at least $\frac{|\mset'''|}{2^{25}\log^4n}\geq\frac{|\hmset|}{2^{27}\log^5n}\geq\frac{|\hmset|}{\log^6n}$ demand pairs, such that all squares in $\qset$ have the $1$-shadow property with respect to $\hmset'$ (we have used the assumption that $n$ is large enough, so $\log n\geq 2^{27}$). It now remains to prove Claim~\ref{claim: the final shadow property}.

\begin{proofof}{Claim~\ref{claim: the final shadow property}}
Assume otherwise, and let $Q\in \qset$ be some square that does not have the $r$-shadow property with respect to $\mset'''$. Since $Q$ must be an active square, there is some disc $D^i\in \dset'$, such that the square $Q$ belongs to $R_{\dset}(D^i)$. Consider the shadow $J_{\tmset'''}(Q)$ of $Q$, and let $s,s'$ be its left and right endpoints respectively. Then there are two demand pairs $(s,t),(s',t')\in \mset'''$, with $t,t'\in Q$, such that at least $rd$ vertices of $S(\mset''')$ lie between $s$ and $s'$ on $R^*$. Let $S'$ be the set of all these source vertices, and let $\nset\subseteq\mset'''$ be their corresponding demand pairs. We claim that $\nset\subseteq \mset_i$. Indeed, assume that there is some demand pair $(s^*,t^*)\in \nset$ that does not belong to $\mset_i$. Then it must belong to some set $\mset_j$, where $v_j$ is a descendant of $v_i$ in the forest $F$. This can only happen if $v_j\in Y_x$. Moreover, the segment $\sigma_1(D^j)$ must be contained in the shadow $L_{\tmset''}(Q)$. Let $v_{i'}$ be the child vertex of $v_i$ that lies in $Y_x$ (such a vertex must exist due to the vertex $v_j\in Y_x$). Then $v_j$ is also the descendant of $v_{i'}$, and $D^j\subseteq D^{i'}\subseteq D^i$. In particular, the segment $\sigma_1(D^{i'})$ must also be contained in the shadow $L_{\tmset''}(Q)$. But then $(s,t)$ is a left pair, while $(s',t')$ is a right pair, which is impossible, since we have discarded all the right pairs. We conclude that $\nset\subseteq \mset_i$.

Since $|\mset_i|\geq |\nset|>rd=2^{23}d\log^4n$, we have deleted the boundary pairs from the set $\mset_i$, and so in particular, $(s,t)$ and $(s',t')$ are not boundary pairs. It follows that $(s,t)$ and $(s',t')$ was a valid candidate pair, and we should have continued our partitioning algorithm.
\end{proofof}

\label{---------------------------------------sec: existence of coloring----------------------}
\section{The Existence of Coloring}\label{sec: existence of coloring}
In this section we prove Theorem~\ref{thm: existence of perfect set}. Recall that we assume that we are given a lower bound $\opt$ on the value of the optimal solution; a hierarchical system $\thset=(\qset_1,\ldots,\qset_{\rho})\in \fset$, and a sequence $L=(\ell_1,\ldots,\ell_{\rho})$ of integers as in Theorem~\ref{thm: partition of R1}. As before, we denote by $(\jset_1,\ldots,\jset_{\rho})$ the $L$-hierarchical decomposition of $R^*$. From Theorem~\ref{thm: partition of R1}, there is a set $\tmset^{**}$ of demand pairs, with $|\tmset^{**}|\geq \opt/2^{O(\sqrt{\log n}\log \log n)}$, such that for all demand pairs in $\tmset^{**}$, their destinations belong to $\thset$, so in particular $\tmset^{**}\subseteq \tmset$. Moreover, all demand pairs in $\tmset^{**}$ can be routed via node-disjoint paths in $G$,  
 all squares in $\bigcup_{h=1}^{\rho}\qset_h$ have the $1/\eta^2$-shadow property with respect to $\tmset^{**}$, and
 the set $\tmset^{**}$ of demand pairs is compatible with the sequence $L$. Recall that the latter means that for each $1\leq h\leq \rho$, for every level-$h$ interval $I\in \jset_h$ of $R^*$, either $I\cap S(\tmset^{**})=\emptyset$, or $d_h/(16\eta)\leq |S(\tmset^{**})\cap I|\leq d_h/4$.
 
 For each $1\leq h\leq \rho$, for every level-$h$ square $Q\in \qset_h$, let $\hmset(Q)\subseteq \tmset^{**}$ denote the set of all demand pairs whose destinations lie in $Q$.
 We use the following simple observation.
 
 \begin{observation}\label{obs: squares spanning few intervals}
 For each $1\leq h\leq \rho$, for every level-$h$ square $Q\in \qset_h$, the source vertices of the demand pairs in $\hmset(Q)$ belong to at most two level-$h$ intervals in $\jset_h$.
 \end{observation}
 
 \begin{proof}
 Assume otherwise, and let $I',I'',I'''\in \jset_h$ be three distinct level-$h$ intervals that  each contain at last one source vertex of $\hmset(Q)$. Assume that $I',I'',I'''$ lie on $R^*$ in this left-to-right order. Then $I''$ is contained in the shadow $L_{\hmset}(Q)$. From the $1/\eta^2$-shadow property of $Q$, the length of this shadow must be at most $d_h/\eta^2$. But from our construction, the length of the shadow is at least $|S(\tmset^{**})\cap I''|\geq d_h/(16\eta)> d_h/\eta^2$, a contradiction.
 \end{proof}
 
 We start with the set $\hmset_0=\tmset^{**}$ of demand pairs, and perform $\rho$ iterations. In iteration $h$, the input is some subset $\hmset_{h-1}\subseteq \hmset_0$ of demand pairs, and an assignment of level-$(h-1)$ colors to all level-$(h-1)$ squares, that has the following property: for each demand pair $(s,t)\in \hmset_{h-1}$, if $Q\in \qset_{h-1}$ is the level-$(h-1)$ square containing $t$, and $I\in \jset_{h-1}$ is the level-$(h-1)$ interval containing $s$, then $Q$ and $I$ are assigned the same level-$(h-1)$ color, or in other words, $Q$ is assigned the level-$(h-1)$ color $c_{h-1}(I)$.
 Our starting point is the set $\hmset_0$ of demand pairs. We assume that the level-$0$ color of all vertices is $c_0$, so, letting $\tilde G$ be a unique level-$0$ square, the invariant holds.
 
 Consider now some iteration $h$, together with the corresponding set $\hmset_{h-1}\subseteq \tmset^{**}$ of demand pairs. Consider some level-$h$ square $Q\in \qset_h$, and assume that its parent-square is assigned the level-$(h-1)$ color $c_{h-1}$. Let $\hmset'(Q)\subseteq \hmset_{h-1}$ be the set of all demand pairs, whose destinations lie in $Q$. From the above observation, there are at most two level-$h$ intervals of $R^*$ (that we denote by $I$ and $I'$, where possibly $I=I'$) that contain the source vertices of $\hmset'(Q)$. Assume w.l.o.g. that $I$ contains at least half of these source vertices. From our invariant, all vertices of the parent-interval of $I$ are also assigned the level-$(h-1)$ color $c_{h-1}$ - same as the color assigned to the parent-square of $Q$. We then assign to $Q$ the level-$h$ color $c_h(I)$ - that is, the color associated with the interval $I$. We also discard from $\hmset'(Q)$ all demand pairs except those whose sources are contained in $I$, obtaining a new set $\hmset''(Q)\subseteq \hmset'(Q)$ of demand pairs, containing at least $|\hmset''(Q)|/2$ demand pairs. We then set $\hmset_h=\bigcup_{Q\in \qset_h}\hmset''(Q)$.
 
 Let $\tmset'=\hmset_{\rho}$ be the set of demand pairs obtained after the last iteration. Then it is immediate to verify that:
 
 \[|\tmset'|\geq \frac{|\tmset^{**}|}{2^{\rho}}\geq \frac{|\tmset^{**}|}{2^{\sqrt{\log n}}}\geq \frac{\opt}{2^{O(\sqrt{\log n}\log \log n)}}.\]
 
 Clearly, for each demand pair $(s,t)\in \tmset'$, both $s$ and $t$ are assigned the same level-$\rho$ color.
 Finally, fix some level $1\leq h\leq \rho$, and some level-$h$ color $c_h(I)$, that is associated with some level-$h$ interval $I$. It remains to verify that at most $d_h$ demand pairs are assigned color $c_h(I)$. This is clearly true, since all such demand pairs have their source vertices lying on $I$, and, since $\tmset^{**}$ is compatible with $L$ from Theorem~\ref{thm: partition of R1}, we are guaranteed that the number of such demand pairs is at most $d_h/4$.

\label{------------------------------------sec: snakes--------------------------------------}
\section{Finding the Routing}\label{sec: snakes}

The goal of this section is to prove Theorem~\ref{thm: snake-like routing}. 
Recall that we are given as input a hierarchical system $\thset=(\qset_1,\ldots,\qset_{\rho})\in \fset$ of squares and a hierarchical L-partition $(\jset_1,\ldots,\jset_{\rho})$ of the row $R^*$ of $G'$.
We are also given a coloring $f$ of $\thset$ by the corresponding hierarchical system  $\cset$ of colors, and a set $\tmset'\subseteq \tmset$ of at most $d_1$ demand pairs that are perfect for $f$.
In particular, we are guaranteed that for each demand pair $(s,t)\in \tmset'$, both $s$ and $t$ are assigned the same level-$\rho$ color, and, moreover, for every $1 \leq h \leq \rho$, for every level-$h$ color $c_h$, at most $d_h$ demand pairs in $\tilde \mset'$ are assigned color $c_h$.
Recall that all destination vertices appear within distance at least $\opt/\eta>4d_1$ from the grid boundary. For consistency, we define a single level-$0$ square $Q_0$, as follows: let $\rset$ contain all but the $d_1$ top and the $d_1$ bottom rows of the grid $G$, and let $\wset$ contain all but the first $d_1$ and the last $d_1$ columns of the grid $G$. Then $Q_0$ is the sub-grid of $G$ spanned by the rows in $\rset$ and the columns in $\wset$. We discard from $\qset_1$ all level-$1$ squares that are not contained in $Q_0$: such squares do not contain any destination vertices. We also discard from all other sets $\qset_h$ all descendant-squares of the level-$1$ squares that were discarded. Abusing the notation, we denote the resulting hierarchical system of squares by $\thset$. For convenience, we define a unique level-$0$ color $c_0$, that serves as a parent of every level-$1$ color, and we assign this color to $Q_0$.

Assume w.l.o.g. that $\tmset'=\set{(s_1,t_1),\ldots,(s_{z},t_{z})}$, where $s_1,\ldots,s_{z}$ appear on $R^*$ in this left-to-right order.
We define a set $\hmset=\set{(s_i,t_i)\mid i\equiv 1 \mod 2 \eta^3}$, of $\Omega(|\tilde \mset'|/\eta^3)$ demand pairs. In the remainder of this section, we show an efficient algorithm for routing all demand pairs in $\hmset$.
Consider a level-$h$ color $c_h$, where $1 \leq h \leq \rho$.
Let $N(c_h)$ be the number of demand pairs $(s,t)\in \hmset$, such that $s$ and $t$ have level-$h$ color $c_h$. 
Since $d_h = \eta^{\rho-h+3}$, $N(c_h) \leq \ceil{d_h/(2\eta^3)} \leq \eta^{\rho-h}$, and $|\hmset|\leq d_1/(2\eta^3)$.
Notice that for every level-$\rho$ color $c_\rho$, $N(c_\rho) \leq 1$ holds.

We say that a level-$h$ square $Q\in \qset_h$ is \emph{empty} iff no demand pair in $\hmset$ has a destination in $Q$.
As before, for each square $Q$, we let $\Gamma(Q)$ denote its boundary. 
For each level $1 \leq h \leq \rho$ and for each level-$h$ square $Q_h \in \qset_{h}$, let $Q_h^+$ be an extended square, obtained by adding a margin of $d_{h}/\eta$ on all sides of $Q_h$; we also let $Q_0^+$ be obtained from $Q_0$ by adding a margin of $d_1/\eta$ around it. Notice that the distance between $\Gamma(Q_0^+)$ and $\Gamma(G)$ remains at least $d_1/2$. Notice also that for each $1\leq h\leq \rho$, for each pair $Q,\hat Q\in \qset_h$ of level $h$ squares, $d(Q^+,\hat Q^+)\geq d_h/2$.

\begin{figure}[h]
\center
\scalebox{0.28}{\includegraphics{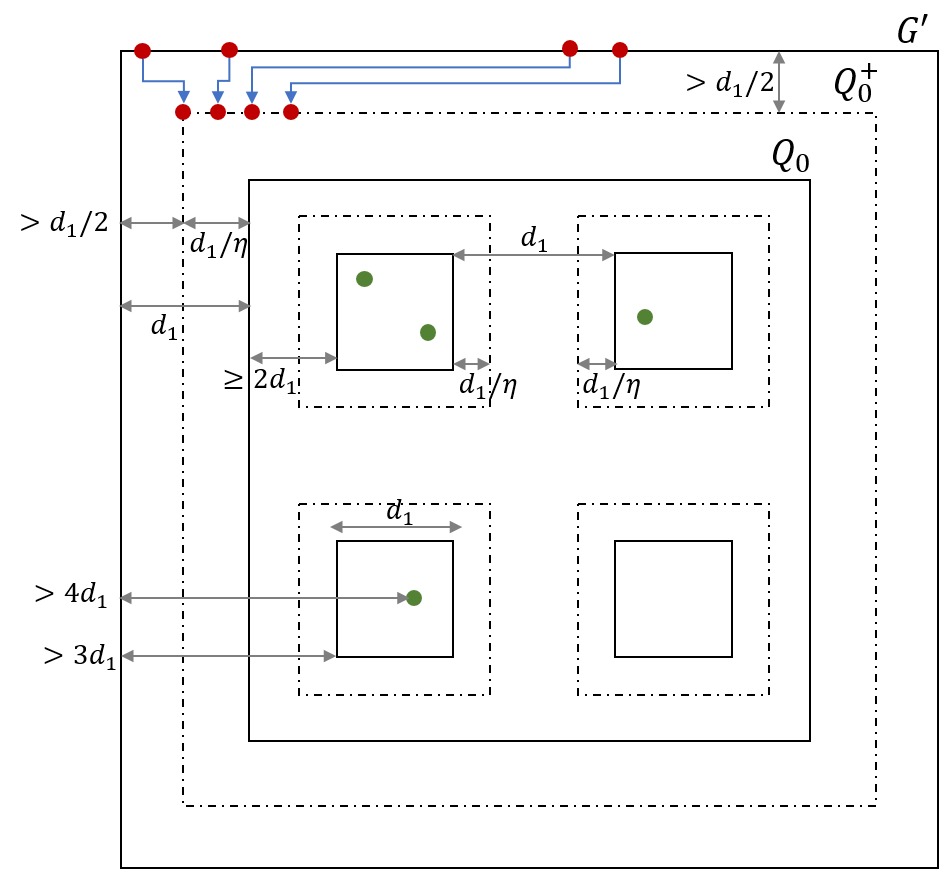}}
\caption{Square $Q_0$ inside $G$. Original squares $Q$ of $\hset$ are represented by solid lines and the corresponding augmented squares $Q^+$ are represented by dashed lines.
The sources $s_j$ of $\hat \mset$ are represented by red nodes on the top boundary of $G'$, and their copies $s'_j$ are on the top row of $Q_0^+$. The paths of $\pset_0$ are shown in blue.
\label{fig: q0 and g}}
\end{figure}

Let $I'_0$ be the set of the $3 \cdot N(c_0)$ leftmost vertices on the top row of $Q_0^+$.
We map the vertices of $S(\hmset)$ to the vertices of $I'_0$ as follows: if $s$ is the $j$th vertex from the left in $S(\hmset)$, then it is mapped to the $(3j)$th leftmost vertex of $I'_0$, that we denote by $s'$.
For every level $1\leq h\leq \rho$, for every level-$h$ interval $I_h\in \iset_h$ of $R^*$, we can now naturally define a corresponding interval $I'_h$ of $I'_0$ as the smallest interval containing all vertices in $\set{s'_j\mid s_j\in I_h}$. We then let $\iset'_h=\set{I'_h\mid I_h\in \iset_h}$, so that $(\iset'_1,\ldots,\iset'_{\rho})$ can be viewed as a hierarchical partition of $I'_0$.
The parent-child relationship between the new intervals is defined exactly as before. Let $\iset'_0=\set{I'_0}$.

The following observation is immediate (see Figure \ref{fig: q0 and g}):
\begin{observation} \label{obs: paths outside q0+}
  There is an efficient algorithm to compute a set $\pset_0$ of node-disjoint paths, connecting every source $s_j \in \sset(\hat \mset)$ to its corresponding copy $s'_j \in I'_0$, such that the paths in $\pset_0$ are internally disjoint from $Q_0^+$.
\end{observation}

In order to define the routing of all demand pairs in $\hat \mset$, we use special a structure called a \emph{snake}, which we define next.
Recall that,
given a set $\rset$ of consecutive rows of $G$ and a set $\wset$ of consecutive columns of $G$, we denoted by $\Y(\rset,\wset)$ the subgraph of $G$ spanned by the rows in $\rset$ and the columns in $\wset$; we refer to such a graph as a \emph{corridor}.

Let $\Y=\Y(\rset,\wset)$ be any corridor.
Let $R'$ and $R''$ be the first and the last row of $\rset$ respectively, and let $W'$ and $W''$ be the first and the last column of $\wset$ respectively.
Each of the four paths $\Y\cap R',\Y\cap R'',\Y\cap W'$ and $\Y\cap W''$ is called a \emph{boundary edge} of $\Y$, and their union is called the \emph{boundary} of $\Y$. The width of the corridor $\Y$ is $w(\Y)=\min\set{|\rset|,|\wset|}$.
We say that two corridors $\Y,\Y'$ are \emph{internally disjoint}, iff every vertex $v\in \Y\cap \Y'$ belongs to the boundaries of both corridors.
We say that two internally disjoint corridors $\Y,\Y'$ are \emph{neighbors} iff $\Y\cap \Y'\neq \emptyset$.
 
 We are now ready to define snakes. A snake $\yset$ of length $z$ is a sequence $(\Y_1,\Y_2,\ldots,\Y_{z})$ of $z$ corridors that are pairwise internally disjoint, such that for all $1\leq z',z'' \leq z$, $\Y_{z'}$ is a neighbor of $\Y_{z''}$ iff $|z'-z''|=1$. The width of the snake is defined to be the minimum of $\min_{1\leq z'<z}\set{|\Y_{z'}\cap \Y_{z'+1}|}$, and $\min_{1\leq z'\leq z}\set{w(\Y_{z'})}$.
 We say that a node $u \in \yset$ iff $u \in \bigcup_{z'=1}^z \Y_{z'}$.
 Notice that given a snake $\yset$, there is a unique simple cycle $\sigma(\yset)$ in $G$ such that all the nodes $u \in \yset$ lie on or inside $\sigma(\yset)$ and all other nodes of $G$ lie outside it.
We call $\sigma(\yset)$ the \emph{boundary} of $\yset$.
We say that a pair of snakes $\yset_1$ and $\yset_2$ are \emph{internally disjoint} iff there is a boundary edge $\Sigma_1$ of some corridor of $\yset_1$ and a boundary edge $\Sigma_2$ of some corridor of $\yset_2$ such that for every vertex $u$ with $u \in \yset_1$ and $u \in \yset_2$, $u$ must belong to $V(\Sigma_1) \cap V(\Sigma_2)$.
 We use the following simple claim, whose proof can be found, e.g. in~\cite{NDP-hardness}.

 \begin{claim}\label{claim: routing in a snake}
Let $\yset=(\Y_1,\ldots,\Y_{z})$ be a snake of width $w$, and let $A$ and $A'$ be two sets of vertices with $|A|=|A'|\leq w-2$, such that the vertices of $A$ lie on a single boundary edge of $\Y_1$ and the vertices of $A'$ lie on a single boundary edge of $\Y_{z}$. There is an efficient algorithm, that, given the snake $\yset$, and the sets $A$ and $A'$ of vertices as above, computes a set $\qset$ of node-disjoint paths contained in $\bigcup_{z'=1}^{z}\Y_{z'}$, that  connect every vertex of $A$ to a distinct vertex of $A'$. 
 \end{claim}

\begin{definition}
Let $0\leq h\leq \rho$ be a level, let $c_h\in \chi_h$ be a level-$h$ color, and let $\yset$ be a snake.  We say that a snake $\yset$ is a \emph{valid level-$h$ snake for color $c_h$} iff:

\begin{itemize}
\item $\yset\subseteq Q_0^+$;
\item $\yset$ has width at least $3 \cdot N(c_h)$;
\item $\yset$ contains the interval $I' \in \iset_h'$, corresponding to the interval $I\in \iset_h$ that was used in defining color $c_h$, as a part of the top boundary of its first corridor; and
\item for each non-empty level-$h$ square $Q\in \qset_h$, whose level-$h$ color is $c_h$, $Q^+$ is a corridor of $\yset$. 
\end{itemize}
\end{definition}

The following lemma is central to proving Theorem~\ref{thm: snake-like routing}.

\begin{lemma}\label{lemma: building snakes}
There is an efficient algorithm, that constructs, for each level $1\leq h\leq \rho$, for each level-$h$ color $c_h$, a valid level-$h$ snake $\yset(c_h)$ for this color, such that all snakes of the same level are disjoint.
\end{lemma}

We provide the proof of the lemma below, after completing the proof of Theorem~\ref{thm: snake-like routing} using it.
Fix some level-$\rho$ color $c_{\rho}$.
Recall that $\hmset$ contains at most one demand pair $(s,t)$, such that $s$ and $t$ are assigned the level-$\rho$ color $c_{\rho}$.
Recall that $s'\in I'_0$ is the vertex on the top boundary of $Q_0^+$ to which $s$ is mapped.
We route the pair $(s',t)$ inside the valid level-$\rho$ snake $\yset=\yset(c_{\rho})$ given by Lemma~\ref{lemma: building snakes}.
This snake is guaranteed to contain both the level-$\rho$ interval $I'\in \iset'(\rho)$ to which $s'$ belongs, and the non-empty level-$\rho$ square
$Q_{\rho}$ that contains $t$ (as the level-$\rho$ color of $Q_{\rho}$ must be $c_{\rho}$).
We select an arbitrary path $P$ connecting $s'$ to $t$ inside $\yset$.
In order to complete the routing, we combine these paths with the set $\pset_0$ of paths computed in Observation \ref{obs: paths outside q0+}.
As all level-$\rho$ snakes are disjoint, we obtain a valid routing of at least $|\hmset|\geq |\tmset'|/\eta^3 = |\tilde \mset'|/2^{O(\sqrt{\log n})}$ demand pairs. It now remains to prove Lemma~\ref{lemma: building snakes}.

 \begin{proofof}{Lemma~\ref{lemma: building snakes}}
The proof is by induction on $h$. For level $0$, we construct a single level-$0$ snake $\yset(c_0)$, consisting of a single corridor $Q_0^+$. Clearly, this is a valid level-$0$ snake for color $c_0$.
Assume now that the claim holds for some level $0\leq h<\rho$.
We prove that it holds for level $(h+1)$.

Fix a level-$h$ interval $I_h\in \iset_h$, and its corresponding level-$h$ color $c_h$.
Let $\yset=\yset(c_h)$ be the valid level-$h$ snake for $c_h$, given by the induction hypothesis.
Recall that for each non-empty level-$h$ square $Q$ of color $c_h$, $Q^+$ is a corridor of $\yset$.
The following claim completes the induction step and Lemma \ref{lemma: building snakes} follows.

\begin{claim} \label{clm: construct child snakes master}
 There is an efficient algorithm, that, given a valid level-$h$ snake $\yset(c_h)$ for a level-$h$ color $c_h$, constructs, for every child-color $c_{h+1}$ of $c_h$, a valid level-$(h+1)$ snake $\yset(c_{h+1})$, such that all resulting level-$(h+1)$ snakes are disjoint from each other and are contained in $\yset(c_h)$.
\end{claim}

\begin{proof}
 We will construct a number of \emph{sub-snakes} \footnote{Each of the sub-snake that we construct is a snake as defined earlier. We use the term `sub-snake' to distinguish them from the final snake $\yset(c_{h+1})$, that is obtained by concatenating a number of sub-snakes.}
  for each color $c_{h+1}$, and concatenate them together to construct a snake $\yset(c_{h+1})$ with claimed properties.
  To this end, we will construct sub-snakes of two types: Sub-snakes of type $1$ will be contained in $\yset$, but they will be internally disjoint from all non-empty level-$h$ squares of color $c_h$. Each sub-snake of type $2$ will be contained inside some non-empty level-$h$ square of color $c_h$.
In order to coordinate between these sub-snakes, we use \emph{interfaces} on the top and bottom boundaries of such squares, that we define next.

  For each non-empty level-$h$ square $Q$ of color $c_h$, let $I_Q$ be the shortest sub-path of the top boundary of $Q^+$ that contains the $3 \cdot N(c_h)$ leftmost  vertices of the top boundary of $Q^+$.
  Similarly, let $J_Q$ be the shortest sub-path of the bottom boundary of $Q^+$ that contains the  $3 \cdot N(c_h)$ rightmost vertices of the bottom boundary of $Q^+$. 
  Notice that $I_Q$ and $J_Q$ are well-defined since $3 \cdot N(c_h) \leq 3 \cdot \ceil {d_h/2\eta^3} \leq d_h$.

  Let $c_h^1, c_h^2, \ldots, c_h^z$ be the child-colors of $c_h$, indexed in the left-to-right order of their corresponding intervals on $R^*$.
 For a non-empty level-$h$ square $Q$,  let $\set{I^i_Q \subseteq I_Q \> | \> i \in \set{1, \ldots, z}}$ be a collection of disjoint sub-paths of $I_Q$ such that the following holds:
  (i) each sub-path $I^i_Q$ contains $3 \cdot N(c_h^i)$ vertices; and
  (ii) for all $1 \leq i < i' \leq z$, $I_Q^i$ appears to the left of $I_Q^{i'}$ on the top boundary of $Q^+$.
  Similarly, let $\set{J^i_Q \subseteq J_Q \> | \> i \in \set{1, \ldots, z}}$ be an analogous collection of disjoint sub-paths of $J_Q$.
  Notice that there is a unique set of such sub-paths, since $\sum_{i=1}^{z} 3 \cdot N(c_h^i) = 3 \cdot N(c_h)$.
  The sub-paths $I^i_Q$ and $J^i_Q$ will act as interfaces  between the sub-snakes corresponding to color $c_h^i$, and will be used to combine them together.
  Having described the interfaces, we now focus on constructing type-$2$ sub-snakes inside each non-empty square of level $h$.
  We will then construct type-$1$ sub-snakes, and finally combine them all together to obtain level-$(h+1)$ snakes as claimed. 

  \begin{definition}
    Fix a non-empty level-$h$ square $Q$ of color $c_h$, and let $c_h^i$ be a child-color of $c_h$. 
    Given a snake $\hat \yset$, we say that $\hat \yset$ is a \emph{valid type-$2$ sub-snake of color $c_h^i$ in square $Q^+$} iff
  \begin{itemize}
    \item $\hat \yset \subseteq Q^+$;
    \item  $\hat \yset$ has width at least $3 \cdot N(c_h^i)$;
    \item $\hat \yset$ contains $I_Q^i$ as the top boundary of its first corridor and $J_Q^i$ as the bottom boundary of its last corridor; and
    \item for each non-empty level-$(h+1)$ child square $\hat Q$ of $Q$ of color $c_h^i$, $\hat Q^+$ is a corridor of $\hat \yset$.
  \end{itemize}
  \end{definition}

  \begin{claim} \label{clm: construct child snakes inside Q}
    There is an efficient algorithm, that, given a non-empty level-$h$ square $Q$ of color $c_h$ constructs, for every child-color $c_h^i$ of $c_h$, a valid type-$2$ sub-snake $\yset_Q(c_h^i)$ of color $c_h^i$ in square $Q^+$, such that the resulting sub-snakes are pairwise disjoint.    
  \end{claim}

  We prove the above claim after completing the proof of Claim \ref{clm: construct child snakes master}, using the claim. We now turn to construct  type-$1$ sub-snakes.
  We starts with the following simple observation.

  \begin{observation}
    Let $\yset = (\Y_1, \Y_2, \ldots, \Y_r)$ be a snake of width $w$.
    Let $\Sigma_1$ be some boundary edge of $\Y_1$ that is contained in the boundary of $\yset$.
    Similarly, let $\Sigma_r \neq \Sigma_1$ be some boundary edge of $\Y_r$ that is contained in the boundary of $\yset$.
    Let $(A_1, A_2, \ldots A_{j})$ be disjoint sub-paths on $\Sigma_1$ and let $(B_1, B_2, \ldots, B_{j})$ be disjoint sub-paths on $\Sigma_r$ such that the sub-paths $(A_1, \ldots, A_j, B_j, B_{j-1}, \ldots, B_1)$ appear in this clockwise order on the boundary of $\yset$ and $\left| \bigcup_{i=1}^j V(A_i) \right|, \left| \bigcup_{i=1}^j V(B_i) \right| \leq w$.
    Then there is an efficient algorithm to construct pairwise disjoint snakes $\yset_1, \ldots, \yset_{j}$ such that for each $1 \leq i \leq j$, (i) $\yset_i$ is contained in $\yset$; 
    (ii) $A_i$ is contained as a part of the boundary of the first corridor of $\yset_i$;
    (iii) $B_i$ is contained as a part of the boundary of the last corridor of $\yset_i$; and
    (iv) the width of the snake $\yset_i$ is at least $\min{(|A_i|, |B_i|)}$.
  \end{observation}

  Recall that we are given a valid level-$h$ snake $\yset$ for a level-$h$ color $c_h$.
  Also, recall that $\yset$ is a sequence of corridors, containing the squares $Q^+$ corresponding to non-empty level-$h$ color $c_h$ squares $Q$ as corridors.
  Let $Q_1, \ldots, Q_r$ be such squares, in the order of their appearance in $\yset$.
  From the above observation, it is now immediate to find pairwise disjoint sub-snakes $\yset^1(c_h^i), \ldots, \yset^r(c_h^i)$ for each child-color $c_h^i$ of $c_h$ such that:
  \begin{itemize}
    \item the top boundary of the first corridor of $\yset^1(c_h^i)$ contains the interval $I'(c_h^i) \in \iset'_{h+1}$ corresponding to the interval $I(c_h^i) \in \iset_{h+1}$ that was used in defining color $c_h^i$;
    \item for each $2 \leq j \leq r$, the top boundary of the first corridor of $\yset^j(c_h^i)$ contains the interval $J_{Q_{j-1}}^i$;
    \item for each $1 \leq j \leq r$, the bottom boundary of the last corridor of $\yset^j(c_h^i)$ contains the interval $I_{Q_j}^i$;
    and
    \item all the sub-snakes are pairwise disjoint, and are internally disjoint from the squares $Q_1, \ldots, Q_r$.
  \end{itemize}

  For each child color $c_h^i$ of $c_h$, we say that the sub-snakes $\yset^1(c_h^i), \ldots, \yset^{r}(c_h^i)$ are sub-snakes of type $1$.
  Now we focus on combining the snakes of type $1$ and type $2$.

  \begin{definition}
    Let $\yset^1 = (\Y^1_1, \ldots, \Y_{\ell_1}^1)$ and $\yset^2 = (\Y^2_1, \ldots, \Y_{\ell_2}^2)$ be a pair of snakes of width $w_1$ and $w_2$ respectively.
    We say that $\yset_1$ and $\yset_2$ are \emph{composable} iff they are internally disjoint and there are boundary edges $\Sigma^1$ of $\Y_{\ell_1}^1$ and $\Sigma^2$ of $\Y_{1}^2$ such that $|V(\Sigma^1) \cap V(\Sigma^2)| \geq \min{(w_1, w_2)}$.
    We denote by $\yset_1 \oplus \yset_2$ the sequence $(\Y^1_1, \ldots, \Y_{\ell_1}^1, \Y^2_1, \ldots, \Y_{\ell_2}^2)$ of corridors.
  \end{definition}

  We begin with a simple observation that follows directly from the definition of snakes.

  \begin{observation}
    Let $\yset_1$ and $\yset_2$ be a pair of composable snakes, each of width at least $w$.
    Then $\hat \yset = \yset_1 \oplus \yset_2$ is also a snake of width at least $w$.      
  \end{observation}

  Notice that for each child-color $c_h^i$ of $c_h$ and for each $1 \leq j \leq r$, the pair of sub-snakes $\yset^j(c_h^i)$ and $\yset_{Q_j}(c_h^i)$ are composable.
  Similarly, for each $1 \leq j < r$, the pair of sub-snakes $\yset_{Q_j}(c_h^i)$ and $\yset^{j+1}(c_h^i)$ are composable.
  Claim \ref{clm: construct child snakes master} now follows by considering the snakes 
  \[ \yset(c_h^i) = \yset^1(c_h^i) \oplus \yset_{Q_1}(c_h^i) \oplus \yset^2(c_h^i) \oplus \ldots \oplus \yset^r(c_h^i) \oplus \yset_{Q_r}(c_h^i)\]
  for all child-colors $c_h^i$ of $c_h$.
  \end{proof}
  \endproofof

  In the remainder of the section, we prove Claim \ref{clm: construct child snakes inside Q}.

\proofof{Claim \ref{clm: construct child snakes inside Q}}

Recall that we are given a non-empty level-$h$ square $Q$ of color $c_h$.
Let $\dset$ be the set of all non-empty level-$(h+1)$ child squares of $Q$.
Recall that given a pair of intervals $I,I'$ of $[\ell']$, we denoted by $Q(I,I')$ the sub-graph induced by all vertices $\set{v(i,j) \> | \> i \in I, j \in I'}$.
We say that a pair of squares $Q_1, Q_2 \in \dset$ are \emph{vertically aligned} iff $Q_1 = Q(I, I_1)$ and $Q_2 = Q(I, I_2)$ for some intervals $I, I_1, I_2$ of $[\ell']$.
Let $\dset_1, \dset_2, \ldots, \dset_r$ be a partition of $\dset$ such that the following holds:
\begin{itemize}
  \item for each $1 \leq i \leq r$, the squares in $\dset_i$ are vertically aligned; and
  \item for each $1 \leq i < i' \leq r$, all squares in $\dset_i$ appear to the left of those in $\dset_{i'}$.
  \end{itemize}

For each $1 \leq i \leq r$, we denote $\dset_i = (Q_i^1, \ldots, Q_i^{k_i})$, where for each $1 \leq j < j' \leq k_i$, the square $Q_i^j$ appears above $Q_i^{j'}$.

Our goal, as before, is to construct a number of child-snakes\footnote{As in the case of sub-snakes, each of the child-snake that we construct, is a snake as defined earlier. We use the term `child-snake' to distinguish them from the final combined sub-snake $\yset_Q(c_h^i)$ that we construct.}
 for each child-color $c_h^i$ of $c_h$, which can be combined together to obtain the desired sub-snake $\yset_Q(c_h^i)$.
To this end, we need interfaces, which we define next.

For each child-square $\hat Q \in \dset$ of $Q$, let $R(\hat Q^+)$ and $R'(\hat Q^+)$ be the rows containing the top and bottom boundaries of $\hat Q^+$ respectively.
Let $\alpha(\hat Q)$ be the sub-path containing all nodes of the top boundary of $\hat Q^+$ along with $3 \cdot N(c_h)$ nodes of $R(\hat Q^+)$ lying immediately to the left and $3 \cdot N(c_h)$ nodes of $R(\hat Q^+)$ lying immediately to the right of the top boundary of $\hat Q^+$.
Similarly, define $\beta(\hat Q)$ analogously for the bottom boundary of $\hat Q^+$.
We say that the vertices in paths $\alpha(\hat Q)$ and $\beta(\hat Q)$ are the \emph{portals} of $\hat Q$.
Let $\set{\alpha_i(\hat Q) \subseteq \alpha(\hat Q) \> | \> i \in \set{1, \ldots, z}}$ be a set of disjoint intervals in the natural left-to-right ordering with following two properties:
(i) there are at least $3 \cdot N(c_i)$ nodes in the interval $\alpha_i(\hat Q)$; and
(ii) if the square $\hat Q$ is assigned level-$(h+1)$ color $c_h^i$, then $\alpha_i(\hat Q)$ is the top boundary of $\hat Q^+$.
We define the intervals $\beta_i(\hat Q) \subseteq \beta(\hat Q)$ corresponding to the bottom boundary of $\hat Q^+$ for each $1 \leq i \leq z$ in a similar fashion.
The sub-paths $\alpha_i(\hat Q)$ and $\beta_i(\hat Q)$ will act as interfaces to connect child-snakes corresponding to color $c_h^i$.
Now, we focus on constructing the child-snakes inside $Q^+$.

Notice that $N_h(c_h) \leq \ceil{d_h/2\eta^3} < d_h/\eta^3 \leq d_{h+1}/18$, since $\eta \geq 3$.
Recall that for each pair $\hat Q, \tilde Q \in \dset$ of distinct level-$(h+1)$ squares contained inside $Q$, $d(\hat Q^+,\tilde Q^+), d(\hat Q^+, \Gamma(Q^+)), d(\tilde Q^+, \Gamma(Q^+)) \geq d_{h+1}/2 > 9 \cdot N(c_h)$.
Hence, portals of any pair of distinct squares are separated by at least $3 \cdot N(c_h)$ nodes from each other.
Exploiting this spacing, it is now immediate to construct the child-snakes of the following five types (see Figure \ref{fig: q and children} for illustration).

\begin{itemize}
  \item The set of child-snakes of type $A$ contains, for each child-color $c_h^i$ of $c_h$, a single snake $\yset^A(c_h^i)$ of width $3 \cdot N(c_h^i)$ with $I_Q^i$ as the top boundary of its first corridor and $\alpha(Q^1_1)$ as the bottom boundary of its last corridor.

  \item The set of child-snakes of type $B$ contains, for each child-color $c_h^i$ of $c_h$, a single snake $\yset^B(c_h^i)$ of width $3 \cdot N(c_h^i)$ with $\beta_i(Q_r^{k_r})$ as the top boundary of its first corridor and $J_Q^i$ as the bottom boundary of its last corridor.
 
  \item The set of child-snakes of type $C$ contains, for each child-color $c_h^i$ of $c_h$,  for each $1 \leq j \leq r$ and $1 \leq j' < k_j$, a snake $\yset_j^{j'}(c_h^i)$ of width $3 \cdot N(c_h^i)$ with $\beta_i(Q_j^{j'})$ as the top boundary of its first corridor and $\alpha(Q_j^{j'+1})$ as the bottom boundary of its last corridor.
  \item The child-snakes of type $D$ consists of, for each child color $c_h^i$ of $c_h$ and $1 \leq j < r$, a snake $\yset_j^{k_j}(c_h^i)$ of width $3 \cdot N(c_h^i)$ with $\beta_i(Q_j^{k_j})$ as the top boundary of its first corridor and $\alpha_i(Q_{j+1}^{1})$ as the bottom boundary of its last corridor.

  \item The child-snakes of type $E$ consists of, for each child-color $c_h^i$ of $c_h$ and square $Q_j^{j'} \in \dset$, a snake $\hat \yset_j^{j'}(c_h^i)$ of width $3 \cdot N(c_h^i)$ with $\alpha_i(Q_j^{j'})$ as the top boundary of its first corridor and $\beta_i(Q_j^{j'})$ as the bottom boundary of its last corridor.
  Moreover, if the level-$(h+1)$ color of $Q_j^{j'}$ is $c_h^i$, then $\hat \yset_j^{j'}(c_h^i)$ is a snake contains the corresponding extended square $Q_j^{j'+}$ as a corridor.
\end{itemize}

Moreover, we also ensure that the child-snakes corresponding to different colors $c_h^i$ are disjoint, and for each child color $c_h^i$ the snakes are composable.
 Claim \ref{clm: construct child snakes inside Q} now follows by considering the snakes 
  \[ \yset_Q(c_h^i) = \yset^A(c_h^i) \oplus \hat \yset_1^1(c_h^i) \oplus \yset_1^1(c_h^i) \oplus \hat \yset_1^2(c_h^i) \oplus \yset_1^2(c_h^i) \oplus \ldots \oplus \yset_{r}^{k_r-1}(c_h^i) \oplus \hat \yset_r^{k_r}(c_h^i) \oplus \yset^B(c^i_{h})\]
  for all child-colors $c^i_h $ of $c_h$.
\end{proofof}

\begin{figure}[h]
\center
\scalebox{0.26}{\includegraphics{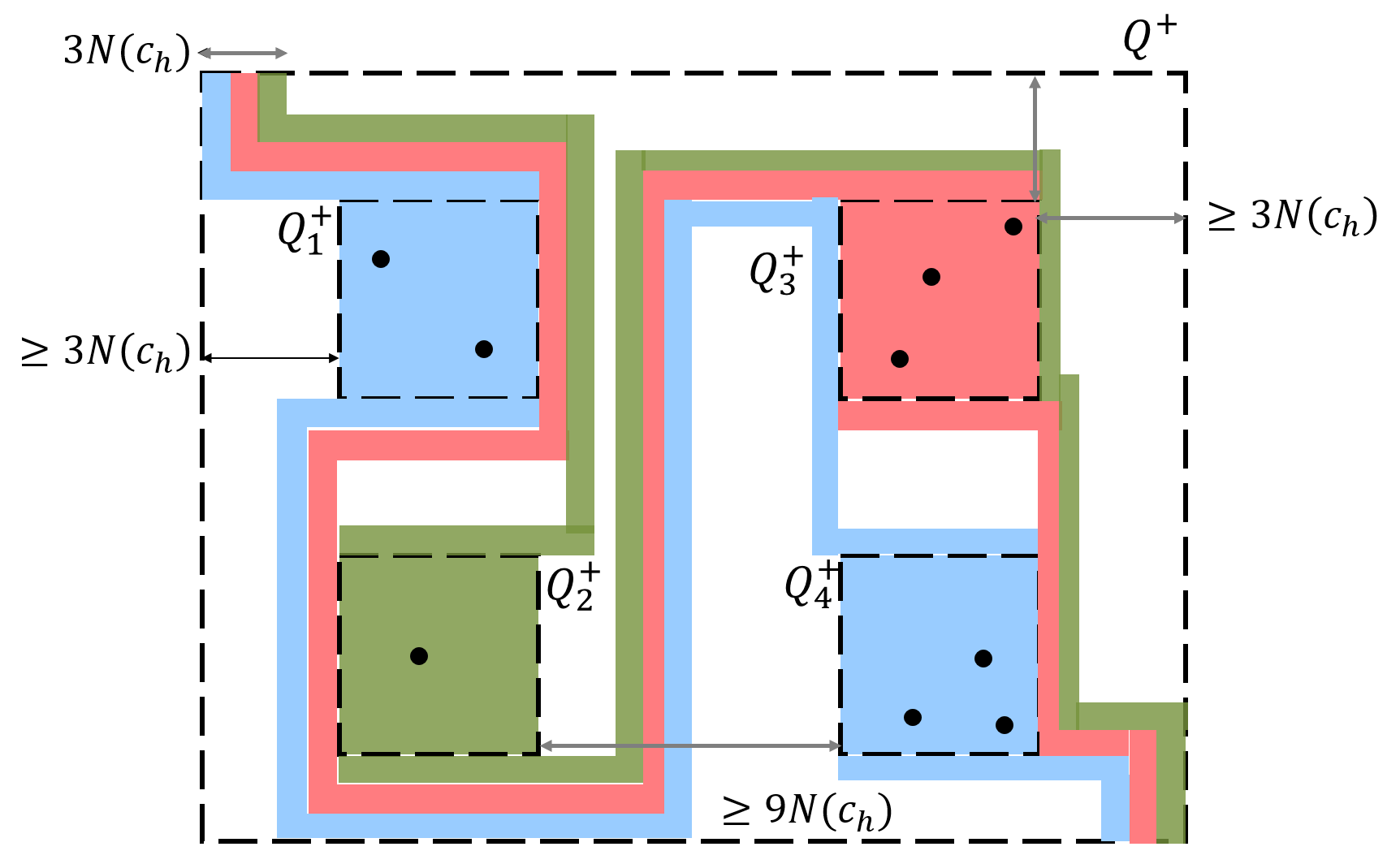}}
\caption{Squares $Q_1, Q_2, Q_3, Q_4$ are children of level $h$ square $Q$ of color $c_h$.
We label them in this way for simplicity. Here, $Q_1$ and $Q_2$, as well as $Q_3$ and $Q_4$ are vertically aligned.
The width of the snake responsible to route $N(c_h)$ demand pairs, containing $Q^+$ as a corridor, is $3N(c_h)$.
The figure shows the corridors of snakes corresponding to the child-colors of $c_h$ inside square $Q^+$.
\label{fig: q and children}}
\end{figure}

\label{-----------------------------sec: coloring the squares-------------------------------}
\section{Finding the Coloring of Squares and a Perfect Set of Demand Pairs}\label{sec: coloring the squares}
The goal of this section is to prove Theorem~\ref{thm: finding the coloring}. It would be convenient for us to reformulate the problem slightly differently. Recall that we are given a hierarchical system $\cset$ of colors, and that for all source vertices $s$ of the demand pairs in $\tmset$, their colorings are fixed: that is, each vertex $s\in S(\tmset)$ is assigned, for each level $1\leq h\leq \rho$, some level-$h$ color $c_h$. Consider now some demand pair $(s,t)\in \tmset$. For each level $1\leq h\leq \rho$, if $s$ is assigned the level-$h$ color $c_h$, then we will assign the same  level-$h$ color $c_h$ to $t$. Note that the same vertex of $G$ may serve as a destination of several demand pairs; in such a case, we view this vertex as having different copies, each of which is colored according to the coloring of the corresponding source vertex.

For now, we will ignore the set $S(\tmset)$ of the source vertices of the demand pairs in $\tmset$, and we will focus on the multi-set $U$ of the destination vertices in $T(\tmset)$. Suppose we compute some coloring $f$ of $\thset$ by $\cset$. Consider some level-$\rho$ square $Q$, and some destination vertex $t\in U\cap Q$. Suppose square $Q$ is assigned the level-$\rho$ color $c_{\rho}$. We say that $t$ \emph{agrees} with the coloring of $Q$ iff $t$ was also assigned the level-$\rho$ color $c_{\rho}$. We can now restate the problem slightly differently. We call the new problem Hierarchical Square Coloring (\HSC). 
In this problem, we are given a hierarchical system $\thset=(\qset_1,\ldots,\qset_{\rho})$ of squares, a hierarchical system $\cset$ of colors, and a multi-set $U$ of vertices that lie in $\bigcup_{Q\in \qset_{\rho}}Q$, together with an assignment of a level-$\rho$ color to each vertex in $U$. Notice that this assignment implicitly assigns, to each vertex $v\in U$, for each level $1\leq h<\rho$, a unique level-$h$ color $c_h$, which is the level-$h$ ancestor of the level-$\rho$ color assigned to $v$.
Our goal is to find a coloring $f$ of $\thset$ by $\cset$, and a subset $U'\subseteq U$ of the vertices, such that:

\begin{itemize}
\item For each vertex $t\in U'$, if $t$ belongs to a level-$\rho$ square $Q$, then the coloring of $t$ agrees with the coloring of $Q$; and
\item For each level $1\leq h\leq \rho$, for each level-$h$ color $c_h$, the total number of vertices in $U'$ whose level-$h$ color is $c_h$ is at most $d_h$.
\end{itemize}

The goal is to maximize $|U'|$.  The main result of this section is a proof of the following theorem.

\begin{theorem}\label{thm: main finding the coloring}
There is an efficient randomized algorithm, that, given as input an instance $(\thset,\hmset)$ of the \HSC problem, with high probability returns an $O(\log^4 n)$-approximate solution to it.
\end{theorem}

We first show that Theorem~\ref{thm: finding the coloring} follows from Theorem~\ref{thm: main finding the coloring}. Let $U'\subseteq U$ be the (multi)-set of the destination vertices computed by the algorithm from Theorem~\ref{thm: main finding the coloring}. As every destination vertex in $U'$ belongs to a unique demand pair in $\tmset$, set $U'$ naturally defines a subset $\tmset'\subseteq \tmset$ of demand pairs. It is immediate to verify that $|\tmset'|\geq \Omega(\opt'/\log^4 n)$. This is since the maximum-cardinality subset $\tmset''\subseteq \tmset$ of the demand pairs, for which a coloring $f'$ of $\thset$ by $\cset$ exists, with $\tmset''$ being perfect for $f'$, naturally defines a feasible solution to \HSC, of value $|\tmset''|$. The set $\tmset'$ of demand pairs has all properties of a perfect set, except that we are not guaranteed that the source and the destination vertices are all distinct. We rectify this as follows. Fix any level-$\rho$ color $c_{\rho}$. Recall that there are at most $d_{\rho}=\eta^3=2^{O(\sqrt{\log n})}$ demand pairs $(s,t)\in \tmset'$, such that the level-$\rho$ color of $s$ is $c_{\rho}$. We discard all such pairs but one from $\tmset'$. It is immediate to verify that $|\tmset'|\geq \opt'/2^{O(\sqrt{\log n})}$ still holds. We claim that all demand pairs in this final set $\tmset'$ have distinct sources and destinations. Indeed, consider any two such distinct pairs $(s,t)$, $(s',t')$. Assume first that $s=s'$. Let $c_{\rho}$ be the level-$\rho$ color of $s$. Then $\tmset'$ contains two demand pairs whose sources have the same level-$\rho$ color, a contradiction. Assume now that $s\neq s'$ but $t=t'$. Then the level-$\rho$ colors of $s$ and $s'$ must be different, and so are the level-$\rho$ colors $t$ and $t'$. However, both $t$ and $t'$ belong to the same level-$\rho$ square $Q_{\rho}$, and their color should agree with the level-$\rho$ color of $Q_{\rho}$, which is impossible.

The remainder of this section is dedicated to proving Theorem~\ref{thm: main finding the coloring}. We formulate an LP-relaxation of the problem, and then provide a randomized LP-rounding approximation algorithm for it.

For convenience, for every pair $1\leq h<h'\leq \rho$ of levels, for every level-$h$ square $Q$, we denote by $\dset_{h'}(Q)$ the set of all level-$h'$ squares that are contained in $Q$, and we think of them as level-$h'$ descendants of $Q$. Recall that for each level-$h$ color $c_h$,  we have also denoted by $\tilde \chi_{h'}(c)$ the set of all its descendant-colors that belong to level $h'$.
\subsection{The Linear Program}\label{subsec: the LP}

For every level-$\rho$ square $Q_{\rho}\in \qset_{\rho}$, and every level-$\rho$ color $c_{\rho}\in \chi_{\rho}$, we let $n(Q_{\rho},c_{\rho})$ denote the number of vertices of $U\cap V(Q_{\rho})$, whose level-$\rho$ color is $c_{\rho}$.

Our linear program has three types of variables. Fix some level $1\leq h\leq \rho$ and a level-$h$ color $c_{h}\in \chi_h$. For every level-$h$ square $Q_h\in \qset_h$, we have a variable $x(Q_h,c_h)$ indicating whether we choose color $c_h$ for square $Q_h$. We also have a global variable $Y_h(c_h)$, counting the total number of all vertices in the final solution $U'$, whose level-$h$ color is $c_h$. Finally, for every level-$\rho$ square $Q_{\rho}\in \qset_{\rho}$ and color $c_{\rho}\in \chi_{\rho}$, we have a variable $y_{\rho}(Q_{\rho},c_{\rho})$. Intuitively, if $Q_{\rho}$ is assigned the color $c_{\rho}$, then $y_{\rho}(Q_{\rho},c_{\rho})$ is the total number of vertices of $U'\cap V(Q_{\rho})$ in our solution $U'$ (note that their level-$\rho$ color must be $c_{\rho}$); otherwise, $y_{\rho}(Q_{\rho},c_{\rho})=0$. 

The objective function of the LP is:

\[\max\quad\quad\sum_{c_{\rho}\in \chi_{\rho}}Y_{\rho}(c_{\rho}).\]

We now define the constraints of the LP.
The first set of constraints requires that for each level-$\rho$ square $Q_{\rho}$, and each level-$\rho$ color $c_{\rho}$, if $Q_{\rho}$ is assigned the color $c_{\rho}$, then $y(Q_{\rho},c_{\rho})$ is bounded by $n(Q_{\rho},c_{\rho})$, and otherwise it is $0$.

\begin{equation}
y(Q_{\rho},c_{\rho})\leq n(Q_{\rho},c_{\rho})\cdot x(Q_{\rho},c_{\rho})\quad\quad \forall  c_{\rho}\in \chi_{\rho}, \forall  Q_{\rho}\in \qset_{\rho} \label{LP: bounding rho colored vertices in each square}
\end{equation}

The next set of constraints states that for each level-$\rho$ color $c_{\rho}$, the total number of vertices of $U'$ whose level-$\rho$ color is $c_{\rho}$
is equal to the summation, over all level-$\rho$ squares $Q_{\rho}$, of the number of vertices of $U'\cap V(Q)$, whose level-$\rho$ color is $c_{\rho}$.

\begin{equation}
Y_{\rho}(c_{\rho})=\sum_{Q_{\rho}\in \qset_{\rho}}y(Q_{\rho},c_{\rho}) \quad\quad \forall c_{\rho}\in \chi_{\rho}\label{LP: accounting for Yrho}
\end{equation}

The next set of constraints states that for every level $h$ and every level-$h$ color $c_h$, the total number of vertices of $U'$ whose level-$h$ color is $c_h$ is equal to the total number of vertices of $U'$ whose level-$\rho$ color is a descendant color of $c_h$. 

\begin{equation}
Y_h(c_h)=\sum_{c_{\rho}\in \tilde{\chi}_{\rho}(c_h)}Y_{\rho}(c_{\rho}) \quad\quad \forall 1\leq h< \rho, \forall c_h\in \chi_h \label{LP: Yh is sum of Yrhos}
\end{equation}

The following set of constraints ensures that for each level $h$, and each level-$h$ color $c_h$, no more than $d_h$ vertices of $U'$ whose level-$h$ color is $c_h$ are in the solution.
\begin{equation}
Y_h(c_h)\leq d_h\quad\quad \forall 1\leq h\leq \rho, \forall c_h\in \chi_h \label{LP: bound of Yh by kh}
\end{equation}

The next set of constraints ensures that every level-$1$ square is assigned some level-$1$ color.
\begin{equation}
\sum_{c_1\in \chi_1}x(Q_1,c_1)=1\quad\quad \forall Q_1\in \qset_1\label{LP: sum of colors 1 for level 1}
\end{equation}

It is immediate to verify that all above constraints are satisfied by any valid integral solution to \HSC, so the above LP is indeed a relaxation of \HSC.

We now add a few additional constraints to coordinate between different levels, that will be crucial for our LP-rounding algorithm. 
Consider some level $1\leq h<\rho$, and some level-$h$ square $Q_h$. Let $Q_{h+1}$ be a child-square of $Q_h$. If $Q_h$ is assigned the level-$h$ color $c_h$, then $Q_{h+1}$ must be assigned some child-color of $c_h$. The following constraint expresses this requirement.

\begin{equation}
\sum_{c_{h+1}\in \tilde{\chi}_{h+1}(c_h)}x(Q_{h+1},c_{h+1})=x(Q_h,c_h)\quad\quad \forall 1\leq h<\rho; \quad \forall Q_h\in \qset_h; \quad \forall c_h\in \chi_h; \quad \forall Q_{h+1}\in \dset_{h+1}(Q_h) \label{LP: color coordination across levels}
\end{equation}

The next constraint requires that for every pair $1\leq h\leq h'\leq \rho$ of levels, for every level-$h$ square $Q_h$ and every level-$h$ color $c_h$, if $c_{h'}$ is a level-$h'$ descendant color of $c_h$, then  the total number of vertices in $U'\cap V(Q_h)$ that are assigned  the level-$h'$  color $c_{h'}$, is at most $d_{h'}$, if square $Q_{h}$ is assigned color $c_h$, and it is $0$ otherwise (notice that we also require this for $h=h'$):

\begin{equation}
\sum_{Q_{\rho}\in \dset_{\rho}(Q_h)}\sum_{c_{\rho}\in \tilde{\chi}_{\rho}(c_{h'})}y(Q_{\rho},c_{\rho})\leq d_{h'}\cdot x(Q_h,c_h)\quad\quad \forall 1\leq h\leq h'\leq \rho; \quad \forall Q_h\in \qset_h; \quad \forall c_h\in \chi_h;  \forall c_{h'}\in \tilde{\chi}_{h'}(c_h)\label{LP: color coordination with kh across levels}
\end{equation}

(Note that indeed for an integral solution, $\sum_{Q_{\rho}\in \dset_{\rho}(Q_h)}\sum_{c_{\rho}\in \tilde{\chi}_{\rho}(\chi_{h'})}y(Q_{\rho},c_{\rho})$ is the total  number of vertices in $U'\cap V(Q_h)$ that are assigned the level-$h'$ color $c_{h'}$.)

Finally, we add non-negativity constraints:

\begin{equation}
x(Q_h,c_h)\geq 0\quad\quad \forall 1\leq h\leq \rho; \quad \forall c_h\in \chi_h \label{LP: non-neg1}
\end{equation}

\begin{equation}
y(Q_{\rho},c_{\rho})\geq 0\quad\quad \forall Q_{\rho}\in \qset_{\rho};\quad \forall c_{\rho}\in \chi_{\rho} \label{LP: non-neg2}
\end{equation}

This completes the description of the LP-relaxation. We now proceed to describe the LP-rounding algorithm.



\subsection{LP-Rounding}\label{subsec: LP-rounding}
For convenience, we denote $|U|=n$.
Our LP-rounding algorithm proceeds in three stages. In the first stage we perform a simple randomized rounding, level-by-level. We will show that the expected number of vertices that we select to our solution $U'$ is equal to the LP-solution value. However, it is possible that we select too many vertices of some color. In the second stage, we define, for every level $h$, and every level-$h$ color $c_h$, a bad event that too many vertices of color $c_h$ are added to $U'$, and show that with high probability none of these bad events happen. If any of the bad events happen, we discard all vertices and return an empty solution. Notice that in this case, we will discard at most $|U|=n$ vertices. Therefore, if we can prove that the probability of any bad event happening is less than $1/(2n)$, then the expected solution cost will remain close to the optimal one. The result of the second stage is a set $U'$, whose expected cardinality is close to the LP-solution cost, and this solution is almost feasible: namely, for every level $h$ and every level-$h$ color $c_h$, the number of vertices in $U'$ whose level-$h$ color is $c_h$ is at most $d_h\cdot \poly\log(n)$. In the third stage, we turn the solution $U'$ into a feasible one, at the cost of losing a polylogarithmic factor on its cardinality.

For convenience, we let $Q_0$ denote the whole grid $\tilde G$, and we define a unique level-$0$ color $c_0$; every level-1 color $c_1\in \chi_1$ is a child-color of $c_0$. We assume that $Q_0$ is assigned the level-$0$ color $c_0$.

\subsection*{Stage 1: Randomized Rounding}
We perform $\rho$ iterations, where in iteration $h$ we settle the colors of all level-$h$ squares, by suitably modifying the LP-solution. 
We will maintain the following invariant:

\begin{properties}{P}
\item For all $1\leq h\leq \rho$, at the beginning of iteration $h$, the current LP-solution ($x', y'$) has the following properties.  For every level-$(h-1)$ square $Q_{h-1}$, if $c_{h-1}$ is the level-$(h-1)$ color assigned to $Q_{h-1}$, then for every level $h-1<h'\leq \rho$, for every level-$h'$ descendant-square $Q_{h'}\in \dset_{h'}(Q_{h-1})$, and every level-$h'$ color $c_{h'}\in \chi_{h'}$:  \label{invariant}

\begin{itemize}
\item If $c_{h'}$ is not a descendant-color of $c_{h-1}$, then $x'(Q_{h'},c_{h'})=0$ (and if $h'=\rho$ then $y'(Q_{h'},c_{h'})=0$);

\item Otherwise, $x'(Q_{h'},c_{h'})=x(Q_{h'},c_{h'})/x(Q_{h-1},c_{h-1})$. Moreover, if $h'=\rho$, then $y'(Q_{h'},c_{h'})=y(Q_{h'},c_{h'})/x(Q_{h-1},c_{h-1})$.
\end{itemize}
\end{properties}

(Here, the $x$ and the $y$-variables are the ones from the original LP-solution).

In order to perform the first iteration, each level-$1$ square $Q_1\in \qset_1$ chooses one of the colors $c_1\in \chi_1$, where each color $c_1$ is chosen with probability $x(Q_1,c_1)$ (recall that by Constraint~(\ref{LP: sum of colors 1 for level 1}), $\sum_{c_1\in \chi_1}x(Q_1,c_1)=1$). Assume now that we have chosen color $c_1$ for square $Q_1$. We update the LP-solution. For every level $2\leq h\leq \rho$, every level-$h$ square $Q_h\in \dset_h(Q_1)$, and every level-$h$ color $c_h$, we do the following: if $c_h$ is not a descendant color of $c_1$, then we set $x'(Q_h,c_h)=0$ (if $h=\rho$, then we also set $y'(Q_h,c_h)=0$). Otherwise, we set $x'(Q_h,c_h)=x(Q_h,c_h)/x(Q_1,c_1)$ (similarly, if $h=\rho$, then we also set $y'(Q_h,c_h)=y(Q_h,c_h)/x(Q_1,c_1)$). We then continue to the next iteration.

The invariant clearly holds at the beginning of the second iteration. Assume now that it holds at the beginning of iteration $h$. The iteration is executed as follows. Consider some level-$h$ square $Q_h$. Suppose its parent square is $Q_{h-1}$, and it was assigned color $c_{h-1}$. Then Constraint~(\ref{LP: color coordination across levels}), together with  Invariant (\ref{invariant}) ensures that:

\[\sum_{c_h\in \tilde{\chi}_h(c_{h-1})}x'(Q_h,c_h)=1.\]

Square $Q_h$ chooses one of the child-colors $c_h$ of $c_{h-1}$, where color $c_h$ is chosen with probability $x'(Q_h,c_h)$.  Assume that we have assigned color $c_h$ to square $Q_h$. We now update the LP-solution. Consider some level $h<h'\leq \rho$, some level-$h'$ square $Q_{h'}\in \dset_{h'}(Q_{h})$ that is a descendant of $Q_{h}$, and some level-$h'$ color $c_{h'}\in \chi_{h'}$. Assume first that if $c_{h'}$ is not a descendant color of $c_{h}$. Then we set $x''(Q_{h'},c_{h'})=0$ (if $h'=\rho$, then we also set $y''(Q_{h'},c_{h'})=0$). Assume now that $c_{h'}$ is a descendant color of $c_{h}$. Let $Q_{h-1}$ be the parent-square of $Q_h$, and let $c_{h-1}$ be the parent-color of $c_h$. Then we set:

\[x''(Q_{h'},c_{h'})=\frac{x'(Q_{h'},c_{h'})}{x'(Q_h,c_h)}=\frac{x(Q_{h'},c_{h'})/x(Q_{h-1},c_{h-1})}{x(Q_h,c_h)/x(Q_{h-1},c_{h-1})}=\frac{x(Q_{h'},c_{h'})}{x(Q_h,c_h)}.\]

 (We have used the fact that  Invariant (\ref{invariant}) held at the beginning of the current iteration).
Similarly, if $h=\rho$, then we set:

\[y''(Q_{h'},c_{h'})=\frac{y'(Q_{h'},c_{h'})}{x'(Q_h,c_h)}=\frac{y(Q_{h'},c_{h'})/x(Q_{h-1},c_{h-1})}{x(Q_h,c_h)/x(Q_{h-1},c_{h-1})}=\frac{y(Q_{h'},c_{h'})}{x(Q_h,c_h)}.\]

We then replace solution $(x',y')$ with the new solution $(x'',y'')$.
This completes the description of iteration $h$. Observe that Invariant (\ref{invariant}) continues to hold at the end of the current iteration.

The final iteration $\rho$ is executed as follows. For every level-$\rho$ square $Q_{\rho}\in \qset_{\rho}$, we select a level-$\rho$ color $c_{\rho}$ for $Q_{\rho}$ exactly as before, but we do not update the LP-solution. If $\frac{y'(Q_{\rho},c_{\rho})}{x'(Q_{\rho},c_{\rho})}\geq 1$, then we add $\ceil{\frac{y'(Q_{\rho},c_{\rho})}{x'(Q_{\rho},c_{\rho})}}$ distinct  vertices of $Q_{\rho}\cap U$, whose level-$\rho$ color is $c_{\rho}$  to our solution  $U'$. This quantity is guaranteed to be bounded by $n(Q_{\rho},c_{\rho})$, 
due to constraint~(\ref{LP: bounding rho colored vertices in each square}) and Invariant (\ref{invariant}). Otherwise, with probability $\frac{y'(Q_{\rho},c_{\rho})}{x'(Q_{\rho},c_{\rho})}$, we add one vertex of $Q_{\rho}\cap U$ whose level-$\rho$ color is $c_{\rho}$ to our solution, and add $0$ such vertices with the remaining probability. (Recall that $\frac{y'(Q_{\rho},c_{\rho})}{x'(Q_{\rho},c_{\rho})}=\frac{y(Q_{\rho},c_{\rho})}{x(Q_{\rho},c_{\rho})}$ from Invariant (\ref{invariant}) --- we use this fact later.) The following theorem concludes the analysis of Stage 1.

\begin{theorem}\label{thm: stage 1 analysis}
The expected number of vertices added to $U'$ at Stage 1 is at least $\sum_{c_{\rho}\in \chi_{\rho}}Y_{\rho}(c_{\rho})$.
\end{theorem}

\begin{proof}
Consider some color $c_{\rho}\in \chi_{\rho}$. Recall that $Y_{\rho}(c_{\rho})=\sum_{Q_{\rho}\in \qset_{\rho}}y(Q_{\rho},c_{\rho})$ due to Constraint~(\ref{LP: accounting for Yrho}). For each level-$\rho$ square $Q_{\rho}$ we define a new random variable $z(Q_{\rho},c_{\rho})$ to be the number of vertices of $Q_{\rho}$, whose level-$\rho$ color is $c_{\rho}$, that are added to the solution $U'$ (this number may only be non-zero if $Q_{\rho}$ is assigned the color $c_{\rho}$). It is now enough to prove that $\expect{z(Q_{\rho},c_{\rho})}\geq y(Q_{\rho},c_{\rho})$. From now on we fix a square $Q_{\rho}$ and a color $c_{\rho}$ and prove the above inequality for them. We use the following claim:

\begin{claim}
The probability that square $Q_{\rho}$ is assigned color $c_{\rho}$ is $x(Q_{\rho},c_{\rho})$.
\end{claim}
\begin{proof}
For all $1\leq h<\rho$, let $Q_h$ be the level-$h$ ancestor-square of $Q_{\rho}$, and let $c_h$ the level-$h$ ancestor-color of $c_{\rho}$. For all $1\leq h\leq \rho$, let $\event_h$ be the event that square $Q_h$ is assigned color $c_h$. Then:

\[\prob{\event_1}=x(Q_1,c_1).\]

We now fix some $1<h\leq \rho$, and analyze the probability of event $\event_h$, assuming that event $\event_{h-1}$ has happened. Let $x',y'$ denote the LP-solution at the beginning of iteration $h$. Then:

\[\prob{\event_h\mid\event_{h-1}}=x'(Q_h,c_h)=\frac{x(Q_h,c_h)}{x(Q_{h-1},c_{h-1})}.\]

Therefore:

\[
\begin{split} 
\prob{\event_{\rho}}&=\prob{\event_{\rho}\mid \event_{\rho-1}}\cdot \prob{\event_{\rho-1}\mid \event_{\rho-2}} \cdots \prob{\event_2\mid \event_1}\cdot \prob{\event_1}\\
&=\frac{x(Q_{\rho},c_{\rho})}{x(Q_{\rho-1},c_{\rho-1})}\cdot \frac{x(Q_{\rho-1},c_{\rho-1})}{x(Q_{\rho-2},c_{\rho-2})}\cdots \frac{x(Q_{2},c_{2})}{x(Q_{1},c_{1})}\cdot x(Q_{1},c_{1})\\
&=x(Q_{\rho},c_{\rho}).\end{split}\]
\end{proof}

We are now ready to complete the proof of Theorem~\ref{thm: stage 1 analysis}.
Assume first that $\frac{y(Q_{\rho},c_{\rho})}{x(Q_{\rho},c_{\rho})}\geq 1$. Then, if $Q_{\rho}$ is assigned color $c_{\rho}$ (which happens with probability $x(Q_{\rho},c_{\rho})$), we select  $\ceil{\frac{y(Q_{\rho},c_{\rho})}{x(Q_{\rho},c_{\rho})}}$ vertices of $Q_{\rho}$, whose level-$\rho$ color is $c_{\rho}$ to the solution. Therefore, $\expect{z(Q_{\rho},c_{\rho})}\geq y(Q_{\rho},c_{\rho})$.

Assume now that $\frac{y(Q_{\rho},c_{\rho})}{x(Q_{\rho},c_{\rho})}< 1$. Then,  if $Q_{\rho}$ is assigned color $c_{\rho}$ (which again happens with probability $x(Q_{\rho},c_{\rho})$), we add one vertex of $Q_{\rho}$, whose level-$\rho$ color is $c_{\rho}$, with probability $\frac{y(Q_{\rho},c_{\rho})}{x(Q_{\rho},c_{\rho})}$. Clearly, in this case, $\expect{z(Q_{\rho},c_{\rho})}= y(Q_{\rho},c_{\rho})$.
\end{proof}

\subsection*{Stage 2: Ensuring that the Solution is Almost Feasible}
In this stage, for each level $1\leq h\leq \rho$ and for each level-$h$ color $c_h\in \chi_h$, we inspect the number of vertices in the solution $U'$ that we constructed at stage 1, whose level-$h$ color is $c_h$. If this number is greater than $64d_h\log^3n$, then we say that color $c_h$ has failed. In that case, we return an empty solution. We think of this as deleting all vertices from our solution, so overall at most $n$ vertices are deleted. 
We will use the following standard Chernoff bound (see e.g. \cite{measure-concentration}).

\begin{theorem}\label{thm: Chernoff}
Let $\set{X_1,\ldots,X_r}$ be a collection of independent random variables taking values in $[0,1]$, and let $X=\sum_iX_i$. Denote $\mu=\expect{X}=\sum_i\expect{X_i}$. Then for all  $0<\eps<1$:

\[\prob{X>(1+\eps)\mu }\leq e^{-\eps^2\mu /3},\]

and 
\[
\prob{X<(1-\eps)\mu}\leq e^{-\eps^2\mu/2}.\]
\end{theorem}

\begin{corollary}\label{cor: Chernoff upper bound}
Let $X_1,\ldots,X_r$ be a collection of independent random variables taking values in $[0,d]$ for some integer $d>0$, and let $X=\sum_iX_i$. Denote $\mu=\expect{X}=\sum_i\expect{X_i}$. Then for all  $0<\eps<1$:

\[\prob{X>(1+\eps)\mu }\leq e^{-\eps^2\mu /(3d)},\]

and 
\[
\prob{X<(1-\eps)\mu}\leq e^{-\eps^2\mu/(2d)}.\]
\end{corollary}

(The Corollary is obtained by replacing each variable $X_i$ with a variable $X'_i=X_i/d$ and then applying Theorem~\ref{thm: Chernoff}).
Following is the main theorem for the analysis of Stage 2.

\begin{theorem}\label{thm: stage 2 bad events}
Fix some level $1\leq h\leq\rho$ and some level-$h$ color $c_h\in \chi_h$. The probability that color $c_h$ fails is at most $1/(2n^3)$.
\end{theorem}

Note that there are at most $n$ colors at each level, and $\rho= O(\sqrt{\log n})$ levels. Therefore, using the union bound, we obtain the following corollary.

\begin{corollary}\label{cor: any color fails}
The probability that any color fails is at most $1/(2n)$.
\end{corollary}

We now turn to prove Theorem~\ref{thm: stage 2 bad events}.

\begin{proofof}{Theorem \ref{thm: stage 2 bad events}}
We assume that the level $h$ and the color $c_h$ are fixed. 
For consistency, we define a fractional solution $(x',y')$ obtained at the end of iteration $\rho$, as follows. Consider a level-$\rho$ square $Q_{\rho}$, and let $c_{\rho}$ be the level-$\rho$ color that was assigned to it at the end of iteration $\rho$. Then we set $x'(Q_{\rho},c_{\rho})=1$, and $y'(Q_{\rho},c_{\rho})=y(Q_{\rho},c_{\rho})/x(Q_{\rho},c_{\rho})$. For every other level-$\rho$ color $c'_{\rho}$, we set $x'(Q_{\rho},c'_{\rho})=0$, and $y'(Q_{\rho},c'_{\rho})=0$.

For simplicity, we denote $\chi'=\tilde{\chi}_{\rho}(c_h)$ - the set of all descendant-colors of $c_h$ at level $\rho$. For every level $1\leq h'\leq \rho$, we define a random variable $N_{h'}$, whose intuitive meaning is the number of color-$c_h$ vertices in the fractional solution, that is obtained at the end of iteration $h'$. 

Formally, let $(x',y')$ be the solution obtained at the end of iteration $h'$. 
Then: $$N_{h'}=\sum_{Q_{\rho} \in \qset_{\rho}}\sum_{c_{\rho}\in \chi'}y'(Q_{\rho},c_{\rho}).$$

For all $1\leq h'\leq \rho$, we define a bad event $\event_{h'}$ to be the event that $N_{h'}>d_h\cdot 16(1+1/\log n)^{h'}\log^3n$. The following lemma is central to the analysis.

\begin{lemma}\label{lem: bounding event rho}
$\prob{\event_{\rho}}\leq \frac 1{4n^3}$.
\end{lemma}

If $\event_{\rho}$ does not happen, then, if we denote by $(x',y')$ the solution obtained at the end of iteration $\rho$: 

\[N_{\rho}=\sum_{Q_{\rho}\in \qset_{\rho}}\sum_{c_{\rho}\in \chi'}y'(Q_{\rho},c_{\rho})<16d_h\left(1+\frac{1}{\log n}\right)^{\rho}\log^3n<32d_h\log^3n,\] 

since $\rho\leq \sqrt{\log n}$.

Recall that for every square $Q_{\rho}\in \qset_{\rho}$, if $y'(Q_{\rho},c_{\rho})\geq 1$, then we add $\ceil{y'(Q_{\rho},c_{\rho})}$ vertices of $Q_{\rho}$, whose level-$\rho$ color is $c_{\rho}$ to the solution. Otherwise, with probability $y'(Q_{\rho},c_{\rho})$, we add only one such vertex to our solution. Overall, from Theorem~\ref{thm: Chernoff}, it is immediate to verify that, if event $\event_{\rho}$ does not happen, then with probability at least $(1-1/n^4)$, the number of vertices in $U'$, whose level-$h$ color is $c_h$ is at most $2N_{\rho} \leq  64d_h\log^3n$. It now remains to prove Lemma~\ref{lem: bounding event rho}.


\begin{proofof}{Lemma \ref{lem: bounding event rho}}
For convenience, we also define event $\event_0$ that $N_0>d_h$, which, from Constraints~(\ref{LP: Yh is sum of Yrhos}) and (\ref{LP: bound of Yh by kh}) happens with probability $0$. Clearly,

\[\begin{split}
\prob{\event_{\rho}}&\leq\prob{\event_{\rho}\mid\neg \event_{\rho-1}}\cdot \prob{\neg \event_{\rho-1}}+\prob{\event_{\rho-1}}\\
&\leq\prob{\event_{\rho}\mid\neg \event_{\rho-1}}+\prob{\event_{\rho-1}\mid\neg \event_{\rho-2}}\cdot \prob{\neg\event_{\rho-2}}+\prob{\event_{\rho-2}} \\
&\vdots\\
&\leq \sum_{h'=1}^{\rho}\prob{\event_{h'}\mid \neg\event_{h'-1}}.
\end{split}
\]

We now analyze $\prob{\event_{h'}\mid \neg \event_{h'-1}}$. The following claim will finish the proof of the lemma.

\begin{claim}
For all $1\leq h'\leq \rho$, $\prob{\event_{h'}\mid \neg \event_{h'-1}}\leq 1/(4n^4)$.
\end{claim}


\begin{proof}
Fix some $1\leq h'\leq \rho$. Let $(x',y')$ be the solution at the beginning of iteration $h'$, and let $(x'',y'')$ be the solution at the end of iteration $h'$. 
We will repeatedly use Constraint~(\ref{LP: color coordination with kh across levels}) that we restate here; we slightly change the indexing to avoid confusion with the variables $h$ and $h'$ that we use here.

The constraint states that for all pairs $1\leq \tilde h\leq \tilde h'\leq \rho$ of levels, for every level-$\tilde h$ square $Q_{\tilde h}\in \qset_{\tilde h}$ and for every pair of a level-$\tilde h$ color $c_{\tilde h}$ and a level-$\tilde h'$ color $c_{\tilde h'}$, that is a descendant-color of $c_{\tilde{h}}$, the following inequality holds: 

\begin{equation}
\sum_{Q_{\rho}\in \dset_{\rho}(Q_{\tilde h})}\sum_{c_{\rho}\in \tilde{\chi}_{\rho}(c_{\tilde h'})}y(Q_{\rho},c_{\rho})\leq d_{\tilde h'}\cdot x(Q_{\tilde h},c_{\tilde h}). \label{LP-constraint restatement}
\end{equation}

Intuitively, if $Q_{\tilde h}$ is assigned the color $c_{\tilde h}$, then at most $d_{\tilde h'}$ vertices of $Q_{\tilde h}$ may be assigned the level-$\tilde h'$ color $c_{\tilde h'}$; otherwise, none of them can. Notice that due to Invariant~(\ref{invariant}), the above constraint remains valid with respect to the solution $(x',y')$, if $\tilde h\geq h'$. 

Consider some level-$h'$ square $Q_{h'}\in \qset_{h'}$. We assume that its parent-square is $Q_{h'-1}$, and that it was assigned some color $c_{h'-1}$. Recall that $(x',y')$ is the LP-solution at the beginning of iteration $h'$, and $(x'',y'')$ is the LP-solution at the end of iteration $h'$.
We define:

\[M'(Q_{h'})=\sum_{Q_{\rho}\in \dset_{\rho}(Q_{h'})}\sum_{c_{\rho}\in \chi'}y'(Q_{\rho},c_{\rho}),\]

and we define $M''(Q_{h'})$ similarly with respect to $y''$. Intuitively, $M'(Q_{h'})$ is the number of vertices of $Q_{h'}$ that are (possibly fractionally) assigned the level-$h$ color $c_h$ at the beginning of iteration $h'$, while $M''(Q_{h'})$ reflects the same quantity at the end of iteration $h'$.
The values of variables $M'(Q_{h'})$ are fixed at the beginning of iteration $h'$, while the values $M''(Q_{h'})$ are random variables.
Then:

\[N_{h'-1}=\sum_{Q_{h'}\in \qset_{h'}}M'(Q_{h'}),\]

and since we assume that event $\event_{h'-1}$ did not happen, $N_{h'-1}\leq d_h\cdot 16(1+1/\log n)^{h'-1}\log^3n$. Similarly,

\[N_{h'}=\sum_{Q_{h'}\in \qset_{h'}}M''(Q_{h'}),\]

and it is enough to prove that the probability that $N_{h'}>16(1+1/\log n)^{h'}\log^3n$ is less than $1/(4n^4)$.

We consider two cases, depending on whether $h>h'$ holds.

\paragraph*{Case 1: $h>h'$.}
Let $c_{h'}$ be the unique ancestor-color of color $c_h$, that belongs to level $h'$.
 Using Constraint~(\ref{LP-constraint restatement}) with $\tilde h= h'$ and $\tilde h'=h$, we get that of each such level-$h'$ square $Q_{h'}$:

\[M'(Q_{h'})\leq d_{h}x'(Q_{h'},c_{h'}).\]

If $Q_{h'}$ is assigned the color $c_{h'}$ (which happens with probability at most $x'(Q_{h'},c_{h'})$), then $M''(Q_{h'})=M'(Q_{h'})/x'(Q_{h'},c_{h'})\leq d_{h}$; otherwise $M''(Q_{h'})=0$. Therefore, $\set{M''(Q_{h'})}_{Q_{h'}\in \qset_{h'}}$ is a collection of independent random variables taking values between $0$ and $d_{h}$, and:

\[\expect{N_{h'}}=\sum_{Q_{h'}\in \qset_{h'}}\expect{M''(Q_{h'})}=\sum_{Q_{h'}\in \qset_{h'}}M'(Q_{h'})=N_{h'-1}\leq d_h\cdot 16(1+1/\log n)^{h'-1}\log^3n.\]


Let $\mu=d_h\cdot 16(1+1/\log n)^{h'-1}\log^3n$.  From Corollary~\ref{cor: Chernoff upper bound}:

\[\prob{N_{h'}\geq (1+1/\log n)\mu}\leq e^{{-\mu}/{(3{d_h}\log^2n)}}\leq e^{-5\log n}\leq 1/(4n^4), \]

as required

 \paragraph*{Case 2: $h\leq h'$.}
Let $Q_{h'}\in \qset_{h'}$ be any level-$h'$ square, and let $c_{h'}\in \tilde{\chi}_{h'}(c_h)$ be any level-$h'$ descendant-color of $c_h$.
 Using Constraint~(\ref{LP-constraint restatement}) with $\tilde h=\tilde h'=h'$, we get that:
 
 \[\sum_{Q_{\rho}\in\dset_{\rho}(Q_{h'})}\sum_{c_{\rho}\in \tilde \chi_{\rho}(c_{h'})}y'(Q_{\rho},c_{\rho})\leq d_{h'}\cdot x'(Q_{h'},c_{h'}).\]
 
 Denote the left-hand-side by this inequality by $M'(Q_{h'},c_{h'})$, so:
 
 \[M'(Q_{h'},c_{h'})=\sum_{Q_{\rho}\in\dset_{\rho}(Q_{h'})}\sum_{c_{\rho}\in \tilde \chi_{\rho}(c_{h'})}y'(Q_{\rho},c_{\rho}),\]
 
 and
 
 \[M'(Q_{h'})=\sum_{c_{h'}\in \tilde{\chi}_{h'}(c_h)} M'(Q_{h'},c_{h'}).\]

For each color $c_{h'}\in \tilde{\chi}_{h'}(c_h)$, if $Q_{h'}$ is assigned the color $c_{h'}$ (which happens with probability $x'(Q_{h'},c_{h'})$), then $M''(Q_{h'})=M'(Q_h',c_{h'})/x'(Q_{h'},c_{h'})\leq d_{h'}\leq d_h$. If none of the colors in $\tilde{\chi}_{h'}(c_h)$ are assigned to $Q_{h'}$, then $M''(Q_{h'})=0$.
Therefore, the expectation of $M''(Q_{h'})$ is:

\[\expect{M''(Q_{h'})}=\sum_{c_{h'}\in \tilde{\chi}_{h'}(c_h)}M'(Q_{h'},c_{h'})=M'(Q_{h'}).\]

Overall, we conclude that variables in set $\set{M''(Q_{h'})}_{Q_{h'}\in \qset_{h'}}$ are independent random variables, taking values in $[0,d_h]$, and that:

\[\expect{N_{h'}}=\sum_{Q_{h'}\in \qset_{h'}}\expect{M''(Q_{h'})}=\sum_{Q_{h'}\in \qset_{h'}}M'(Q_{h'})=N_{h'-1}\leq d_h\cdot 16(1+1/\log n)^{h'-1}\log^3n.\]
Let $\mu=d_h\cdot 16(1+1/\log n)^{h'-1}\log^3n$. From Corollary~\ref{cor: Chernoff upper bound}:

\[\prob{N_{h'}\geq (1+1/\log n)\mu}\leq e^{{-\mu}/{(3{d_h}\log^2n)}}\leq e^{-5\log n}\leq 1/(4n^4). \]
\end{proof} 
\end{proofof} 
\end{proofof} 

\subsection*{Stage 3: Turning the Solution into a Feasible One}

 If any of the colors fail, then we return an empty solution. 
Assume now that no color fails. Let $U'\subseteq U$ be the set of all vertices chosen by the current solution.

\begin{claim}\label{claim: stage 3} There is an efficient algorithm to compute a subset $U''\subseteq U'$ of vertices, with $|U''|\geq \frac{|U'|}{256\log^4n}$, such that $U''$ is a feasible solution to the \HSC problem.
\end{claim}

\begin{proof}
 If, for every level $h$, for every level-$h$ color $c_h$, at most $d_h$ vertices of $U'$ are assigned the level-$h$ color $c_h$, then we return the set $U'$ of vertices, which is a feasible solution to the \HSC problem instance.

Otherwise, denote $r=|U'|$, and consider an ordering $(u_1,u_2,\ldots,u_r)$ of the vertices of $U'$, that has the following property: for every level $1\leq h\leq \rho$, for every level-$h$ color $c_h\in \chi_h$, the vertices of $U'$ whose level-$h$ color is $c_h$ appear consecutively in this ordering. Due to the nested definition of the colors, it is easy to see that such an ordering exists (in particular, we can use the ordering induced by the ordering of the original source vertices on $R^*$). We now let $U''$ contain all vertices $u_i\in U'$, where $i = 1\mod \ceil{128\log^4n}$. It is now immediate to verify that $|U''|\geq \frac{|U'|}{256\log^4n}$, and that it is a feasible solution to the \HSC problem.
\end{proof}
%
%
%

Let $\optLP$ denote the value of the optimal LP solution.
Let $N_1$ be the total number of vertices that belong to the solution at the end of phase $1$, so $\expect{N_1}=\optLP$. Let $N_2$ be the total number of vertices that were deleted during the second phase. Then with probability at least $(1-1/(2n))$, $N_2=0$, and with the remaining probability, $N_2\leq n$. Therefore, $\expect{N_2}\leq 1/2$. Finally, let $N_3$ denote the cardinality of the final set $U''$ of vertices that our algorithm returns. Then:

\[N_3\geq \frac{N_1-N_2}{256\log^4n},\]

and therefore, $\expect{N_3}\geq \frac{\optLP}{512\log^4n}$. 
Since $U''$ is a feasible solution to the \HSC problem, $N_3 \leq \optLP$.

\begin{claim}
    $\prob{N_3 \geq \frac{\optLP}{1024\log^4 n}} \geq \frac1 {1024\log^4 n}$
\end{claim}
\begin{proof}
    Assume otherwise. Then:

    \[ \expect{N_3} < \left ( \frac{1}{1024\log^4 n} \cdot \optLP \right ) + \frac{\optLP}{1024\log^4 n} <\left (\frac{\optLP}{512\log^4 n} \right ),\]
    
    a contradiction.
\end{proof}

We run the above algorithm $c\log^5n$ times independently (for some large constant $c$), and return a solution of largest cardinality across all runs. The probability that the solution value is less than $\optLP/(1024\log^4n)$ is then bounded by $1/\poly(n)$. This concludes the proof of Theorem \ref{thm: main finding the coloring}.



\label{-------------------------------------sec: destinations anywhere-----------------------------------}
\section{Approximating \restrictedNDP}\label{sec:destinations anywhere}
In this section we complete the proof of Theorem~\ref{thm: main}. 
Suppose we are given an instance $(G,\mset)$ of \RNDP.
We assume that we know the value $\opt$ of the optimal solution to this instance, by going over all such possible choices, and running the algorithm on each of them.
Let $\Gamma_1,\Gamma_2,\Gamma_3$ and $\Gamma_4$ denote the four boundary edges of the grid $G$, and let $\Gamma=\bigcup_{q=1}^4\Gamma_q$ be its full boundary.
For every destination vertex $t\in T(\mset)$, we denote by $\tilde{t}$ the vertex of $\Gamma$ minimizing the distance $d(t,\tilde t)$, breaking ties arbitrarily. 
Recall that, given any subset $\mset'\subseteq \mset$ of the demand pairs, we denoted by $S(\mset')$ and $T(\mset')$ the sets of all vertices that serve as the sources and the destinations of the demand pairs in $\mset'$, respectively. We denote by $\tilde{T}(\mset')=\set{\tilde t\mid t\in T(\mset')}$.

Given two disjoint sub-paths $\pi,\pi'\subseteq \Gamma$ (that we call intervals), and a subset $\mset'\subseteq \mset$ of demand pairs, we let $\mset'_{\pi,\pi'}\subseteq \mset'$ contain all demand pairs $(s,t)$ with $s\in \pi$ and $\tilde t\in \pi'$. 

\subsection{Special Instances}
It will be convenient for us to define special sub-instances of instance $(G,\mset)$, that have a specific structure. 
We start by defining interesting pairs of intervals.

\begin{definition}
Let $I,I'\subseteq \Gamma$ be two disjoint intervals, and let $d>0$ be an integer. We say that $(I,I')$ is a \emph{$d$-interesting pair of intervals}, iff:

\begin{itemize}
\item interval $I'$ is contained in a single boundary edge of $G$, and every vertex of $I'$ is within distance at least $16d$ from each of the remaining three boundary edges of $G$; 
\item interval $I$ is contained in a single boundary edge of $G$ (possibly the same as $I'$);

\item $d\leq |V(I')|\leq  \sqrt{n}/2$, and $d(I,I')\geq 16 d$.
\end{itemize}
\end{definition}

\begin{definition}
Suppose we are given an integer $d>0$, and a pair $(I,I')$ of disjoint intervals of $\Gamma$.
Assume further that we are given a subset $\mset'\subseteq \mset$ of the demand pairs, and consider the sub-instance $(G,\mset')$ of the problem, defined over the same graph $G$, with the set $\mset'$ of demand pairs. We say that $(G,\mset')$ is a \emph{valid $(I,I',d)$-instance} iff:
(i) $S(\mset')\subseteq I$; 
(ii) $\tilde{T}(\mset')\subseteq I'$; and
(iii) for each $t\in T(\mset')$, $d\leq d(t,\tilde t)<2d$.
If, additionally, $(I,I')$ is a $d$-interesting pair of intervals, then we say that $(G,\mset')$ is a \emph{perfect $(I,I',d)$-instance}.
\end{definition}

The main idea of the algorithm is to compute a partition of the set $\mset$ of the demand pairs into subsets $\mset_1,\ldots,\mset_z$, where each set $\mset_i$ is a perfect $(I,I',d)$-instance for some pair $(I,I')$ of intervals. We will then solve the problem defined by each such sub-instance separately, by first casting it as a special case in which the destinations lie far from the grid boundary, and then employing Theorem~\ref{thm: main w destinations far from boundary}. We now define the modified instances, to which Theorem~\ref{thm: main w destinations far from boundary}  will eventually be applied.

\subsection{Modified Instances}
Assume that we are given an integer $d>0$, and a $d$-interesting pair $(I,I')$ of intervals, together with a perfect $(I,I',d)$-sub-instance $(G,\mset')$ of $(G,\mset)$. We define a corresponding \emph{modified instance} $(G',\mset'')$. The underlying graph $G'$ will be an appropriately chosen sub-grid of $G$, and each demand pair in the new set $\mset''$ will correspond to a unique demand pair in $\mset'$.

In order to define the grid $G'$, we assume without loss of generality that $I'$ is a sub-path of the bottom boundary edge of the grid. Let $\wset''$ be the subset of the columns of $G$ that intersect $I'$, and let $\rset''$ be the set of $4d$ bottommost rows of $G$. Let $G''\subseteq G$ be the sub-grid of $G$ spanned by the set $\wset''$ of columns and the set $\rset''$ of rows. In order to obtain the final graph $G'$, we add to $\wset''$ $4d$ columns lying immediately to its left and $4d$ consecutive columns lying immediately to its right, denoting the resulting set of columns by $\wset'$. We also add to $\rset''$ a set of $|\wset'|-4d$ rows lying immediately above $\rset''$, obtaining a set $\rset'$ of rows of $G$. Graph $G'$ is the sub-grid of $G$ spanned by the columns of $\wset'$ and the rows of $\rset'$. Notice that $G'$ is a square grid.

The set $\mset''$ of demand pairs is constructed as follows. For each demand pair $(s,t)\in \mset'$, we add a demand pair $(s',t')$ to $\mset''$. Vertex $t'$ is mapped to the same location as vertex $t$ in $G'$. Vertex $s'$ is mapped to one of the vertices on the top row of $G'$, using the following procedure. Let $X$ denote the set of $|S(\mset')|$ leftmost vertices on the top row of $G'$. Let $v$ be the bottom left corner of the grid $G'$. The traversal of $\Gamma(G)$ in the clock-wise direction, starting from $v$, defines an ordering $\pi$ of the vertices in $S(\mset')$. Similarly, a traversal of $\Gamma(G')$ in the clock-wise direction, starting from $v$, defines an ordering $\pi'$ of the vertices in $X$. Consider some vertex $s\in S(\mset')$, and assume that it is the $i$th vertex of $S(\mset')$ according to the ordering $\pi$. We map it to the $i$th vertex of $X$, according to the ordering $\pi'$.
This finishes the definition of the modified instance $(G',\mset'')$ corresponding to the valid $(I,I',d)$ sub-instance $(G,\mset')$ of $(G,\mset)$. We need the following simple claim, whose proof is straightforward and is deferred to the Appendix.

\begin{claim}\label{claim: modified instance preserves solutions}
Let $d>0$ be an integer, let $(I,I')$ be a $d$-interesting pair of intervals, and let $(G,\mset')$ be a perfect $(I,I',d)$ sub-instance of $(G,\mset)$. Let $\opt'$ be the value of the optimal solution to instance $(G,\mset')$. Then the value of the optimal solution to the corresponding modified instance $(G',\mset'')$ is at least $\min\set{\opt',d}$.
\end{claim}

Let $d>0$ be an integer, let $(I,I')$ be a $d$-interesting pair of intervals, and let $(G,\mset')$ be a valid $(I,I',d)$ sub-instance of $(G,\mset)$. Let $\opt'$ be the value of the optimal solution to this instance.
Consider the corresponding modified instance $(G',\mset'')$ of $(G,\mset')$. Then $(G',\mset'')$, together with the parameter $\opt=\min\set{d,\opt'}$ is a valid input to Theorem~\ref{thm: main w destinations far from boundary}. Even though we do not know the value of this parameter, we can try all $d$ possibilities for it, and apply Theorem~\ref{thm: main w destinations far from boundary} to each one of them, selecting the best resulting solution. We are then guaranteed to compute, with high probability, a solution to instance $(G',\mset'')$, whose value is at least $\frac{\min\set{d,\opt'}}{2^{O(\sqrt{\log n}\log\log n)}}$. We let $\aset(G',\mset'')$ denote the resulting solution, and we let $|\aset(G',\mset'')|$ denote its value.

\subsection{Main Partitioning Theorem}

The following theorem is central to our proof.

 \begin{theorem}\label{thm: partition into intervals}
Suppose we are given two disjoint intervals $\pi,\pi'$ of $\Gamma$, each of which is contained in a single boundary edge of $G$, an integer $d>0$, and a valid $(\pi,\pi',d)$-instance $(G,\mset')$. Assume further that $|\mset'|\geq 1024 d$, and all demand pairs in $\mset'$ can be simultaneously routed via node-disjoint paths. Let $z=\ceil{\frac{|\mset'|}{160d}}-1$. Then there is a collection $\Sigma=\set{\sigma_1,\ldots,\sigma_z}$ of disjoint sub-intervals of $\pi$, and a collection $\Sigma'=\set{\sigma_1',\ldots,\sigma'_z}$ of disjoint sub-intervals of $\pi'$, such that:

 \begin{itemize}
 \item the intervals $\sigma_1,\sigma_2,\ldots,\sigma_z,\sigma'_z,\ldots,\sigma_2',\sigma_1'$ appear on $\Gamma$ in this circular order;
 
  \item for every pair of intervals $\sigma\in \Sigma\cup \Sigma'$  and $\sigma'\in \Sigma'$, at least $16d$ vertices of $\Gamma$ separate the two intervals;

  \item every interval $\sigma'\in \Sigma'$ contains at least $16d$ and at most $\sqrt{n}/2$ vertices, and lies within distance at least $16d$ from every boundary edge of $G$, except for that containing $\pi'$; and
  
  \item for all $1\leq i\leq z$, the value of the optimal solution of the \NDP instance $(G,\mset'_{\sigma_i,\sigma'_i})$ is at least $d$.
  \end{itemize}
  \end{theorem}

 Notice that in particular, for each $1\leq i\leq z$, $(\sigma_i,\sigma'_i)$ is a $d$-interesting pair of intervals. Consider the set $\tmset\subseteq \mset$ of all demand pairs, such that for each $(s,t)\in \tmset$, $s\in \pi$, $\tilde t\in \pi'$, and $d\leq d(t,\tilde t)<2d$, and apply Theorem~\ref{thm: partition into intervals} to it. Consider the resulting sets $\Sigma=\set{\sigma_1,\ldots,\sigma_z}$ and $\Sigma'=\set{\sigma'_1,\ldots,\sigma'_z}$ of segments. Then for all $1\leq i\leq z$, $(G,\tmset_{\sigma_i,\sigma'_i})$ is a perfect $(\sigma_i,\sigma'_i,d)$-instance. Let $(G_i,\tmset'_{\sigma_i,\sigma'_i})$ denote the corresponding modified instance for $(G,\tmset_{\sigma_i,\sigma'_i})$.
 Then all resulting sub-graphs $G_1,\ldots,G_z$ of $G$ are disjoint. Moreover, since the length of each interval $\sigma'_i$ is bounded by $\sqrt{n}/2$, if we assume w.l.o.g. that $\pi'$ is contained in the bottom boundary of the grid $G$, then, assuming that $d<\sqrt n/8$, the sub-graphs $G_1,\ldots,G_z$ must be disjoint from the top $\sqrt{n}/2-4d$ rows of $G$. For each $1\leq i\leq z$, with high probability, $|\aset(G_i,\tmset'_{\sigma_i,\sigma'_i})|\geq d/\approxfactor$, and so $\sum_{i=1}^z|\aset(G_i,\tmset'_{\sigma_i,\sigma'_i})|\geq |\tmset|/\approxfactor$ with high probability.
We now turn to prove Theorem~\ref{thm: partition into intervals}.
  
 \begin{proofof}{Theorem \ref{thm: partition into intervals}}
 Let $\pset^*$ be the set of node-disjoint paths routing the demand pairs in $\mset'$.
 For each demand pair $(s,t)\in \mset'$, let $P(s,t)\in \pset^*$ be the path routing it. For convenience, we assume that $\pi'$ is contained in the bottom boundary of $G$.
 
For every interval $\sigma'\subseteq \pi'$, let $\mset'(\sigma')\subseteq \mset'$ be the set of all demand pairs $(s,t)$ with $\tilde t\in \sigma'$. Note that it is possible that for $t,t'\in T(\mset')$ with $t\neq t'$, $\tilde t=\tilde t'$. Given a vertex $v\in \sigma'$, we let $I(v)\subseteq \pi'$ be the shortest interval, whose leftmost vertex is $v$, such that $I(v)$ contains at least $16d$ vertices, and $|\mset'(I(v))|\geq 16d$. If no such interval exists, then $I(v)$ is undefined.

\begin{observation}\label{obs: length of int}
For each vertex $v\in \pi'$, if $I(v)$ is defined, then $|\mset'(I(v))|\leq 20d$.
\end{observation}

\begin{proof}
Let $v'$ be the rightmost vertex of $I(v)$, and let $I'=I(v)\setminus\set{v'}$. Since we have chosen $I(v)$ over the shorter interval $I'$, either $|\mset'(I')|<16d$, or $|V(I')|<16 d$. Assume first that the former is true. Since all demand pairs in $\mset'$ are routable via node-disjoint paths, all vertices in $S(\mset')$ are distinct. As for each terminal $t\in T(\mset')$, $d(t,\tilde t)\leq 2d$, there may be at most $2d$ demand pairs $(s,t)\in \mset'$ with $\tilde t=v'$. But $\mset'(I')$ contains all demand pairs of $\mset'(I(v))$ except for those pairs $(s,t)$ with $\tilde t=v'$. Therefore, $|\mset'(I(v))|\leq |\mset'(I')|+2d\leq 18d$.

Assume now that $|V(I')|<16d$. Let $G'\subseteq G$ be the sub-grid of $G$, whose bottom boundary is $I(v)$, and whose height is $2d$. Then the width of $G'$ is $16d$, and its top, left, and right boundary has total length $20d$. All paths routing the demand pairs in $\mset'(I(v))$ have to cross the boundary of $G'$, and so $|\mset'(I(v))|\leq 20d$.
\end{proof} 
 
  Let $\tilde T=\set{\tilde t\mid t\in T(\mset')}$. We start by defining a sequence $(\mu_1,\ldots, \mu_{2z+2})$ of $2z+2$ disjoint intervals of $\pi'$, as follows. Let $v_1$ be the first vertex of $\pi'$. We then set $\mu_1=I(v_1)$. Assume now that we have defined intervals $\mu_1,\ldots,\mu_i$. Let $v_{i+1}$ be the vertex lying immediately to the right of the right endpoint of $\mu_i$. We then set $\mu_{i+1}=I(v_{i+1})$. Notice that from Observation~\ref{obs: length of int}, for each $i$, $16d\leq |\mset'(\mu_i)|\leq 20d$. Since $z=\ceil{\frac{|\mset'|}{160d}}-1$, all intervals $\mu_1,\ldots,\mu_{2z+2}$ are well-defined.
  For convenience, we denote, for each $1\leq i\leq 2z$, $\mset'(\mu_i)$ by $\mset_i$.
For every vertex $t\in T(\mset')$, let $Q_t$ be the shortest path connecting $t$ to $\tilde t$ in $G$, so $Q_t$ is contained in the column of $G$ which contains $t$. 
We also let  $U(\mu_i)=\bigcup_{(s,t)\in \mset_i}V(Q_t)$. 
   
For all $1\leq i,j\leq 2z$, we say that intervals $\mu_i$ and $\mu_j$ are \emph{neighbors} iff $|i-j|\leq 1$. 
Consider now some interval $\mu_i$, with $1\leq i\leq 2z$ and some demand pair $(s,t)\in \mset_i$. We say that $(s,t)$ is a \emph{bad demand pair} iff path $P(s,t)$ contains a vertex of $U(\mu_j)$ for some interval $\mu_j$ that is not a neighbor of $\mu_i$. Otherwise, we say that $(s,t)$ is a \emph{good demand pair}.

\begin{claim}\label{claim: bound bad pairs}
For each $1\leq i\leq 2z$, at most $8d$ demand pairs of $\mset_i$ are bad.
\end{claim}

We prove the claim below, after we complete the proof of Theorem~\ref{thm: partition into intervals} using it. Note that at most one interval $\mu_j$ may contain more than $\sqrt{n}/2$ vertices; we assume without loss of generality that, if such an interval $\mu_j$ exists, then $j$ is an odd integer.
The set $\Sigma'=\set{\sigma'_1,\ldots,\sigma'_z}$ of intervals of $\pi'$ is defined as follows: for each $1\leq i\leq z$, we let $\sigma'_i=\mu_{2i}$. Notice that every pair of such intervals is separated by at least $16d$ vertices, since each odd-indexed interval $\mu_j$ must contain at least $16d$ vertices by definition. Since we discard the first and the last intervals $\mu_j$, the distance from each resulting interval $\sigma'_i$ to interval $\pi$, and to each of the remaining three boundaries of $G$ is also at least $16d$.

In order to define the set $\Sigma=\set{\sigma_1,\ldots,\sigma_z}$ of intervals, fix some $1\leq i\leq z$. Let $\mset_{i}'\subseteq \mset_{2i}$ be the set of all good demand pairs in $\mset_{2i}$. We let $\sigma_i$ be the smallest sub-interval of $\pi$, containing all vertices of $S(\mset'_{i})$. 

It is immediate to verify that, for each $1\leq i\leq z$, the value of the optimal solution of the \NDP instance $(G,\mset'_{\sigma_i,\sigma'_i})$ is at least $d$, since we can route all demand pairs of $\mset'_i$ in $G$. From our construction, every interval $\sigma'_i\in \Sigma'$ contains at least $16d$ and at most $\sqrt{n}/2$ vertices, and it is separated by at least $16d$ vertices from all other intervals in $\Sigma\cup \Sigma'$,  and from each of the remaining three boundaries of $G$.  It remains to show that the intervals in $\Sigma$ are all disjoint, and that $\sigma_1,\sigma_2,\ldots,\sigma_z,\sigma'_z,\ldots,\sigma_2',\sigma_1'$ appear on $\Gamma$ in this circular order. 

Fix some pair of indices $1\leq i<j\leq z$. If we traverse $\Gamma$ in counter-clock-wise order, starting from the first vertex of $\sigma'_1$, then we will first encounter all vertices of $\sigma_i'$, then all vertices of $\sigma_j'$, and finally all vertices of $\sigma_i\cup \sigma_j$. It is enough to show that we will encounter every vertex of $\sigma_j$ before we encounter any vertex of $\sigma_i$. In particular, it is enough to show that for every pair $(s,t)\in \mset'_j$ and $(s',t')\in \mset'_i$ of demand pairs, we will encounter $s$ before $s'$ in our traversal. 

Assume for contradiction that this is false for some  $(s,t)\in \mset'_j$ and $(s',t')\in \mset'_i$. Then $\tilde{t}',\tilde{t},s',s$ appear on $\Gamma$ in this circular order  (see Figure~\ref{fig: good pairs routing}). Since both demand pairs $(s,t)$ and $(s',t')$ are good, path $P(s,t)$ may not contain a vertex of $Q_{t'}$, and path $P(s',t')$ may not contain a vertex of $Q_t$ (see Figure~\ref{fig: good pairs routing}). But then the union of the path $P(s',t')$ and $Q_{t'}$ separates $s$ from $t$ in $G$, a contradiction. It now remains to prove Claim~\ref{claim: bound bad pairs}.

\begin{figure}[h]
\scalebox{0.5}{\includegraphics{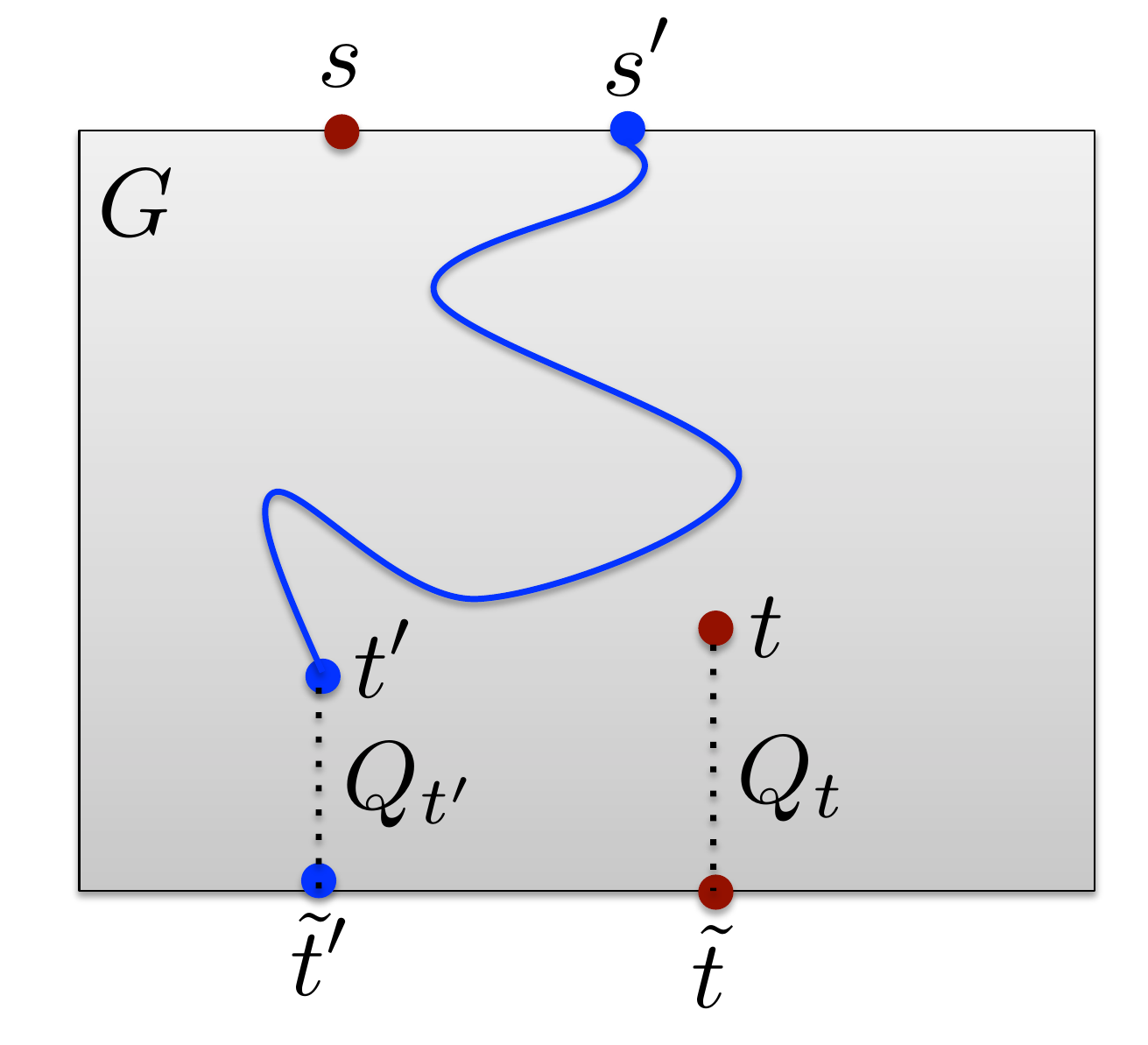}}
\caption{Ordering of the terminals of good demand pairs.\label{fig: good pairs routing}}
\end{figure}

\begin{proofof}{Claim~\ref{claim: bound bad pairs}}
Fix some index $1\leq i\leq 2z$, and let $\nset$ be the set of all bad demand pairs in $\mset_i$. We further partition $\nset$ into two subsets: set $\nset'$ contains all demand pairs $(s,t)$ whose corresponding path $P(s,t)$ contains a vertex in some set $U(\mu_j)$ for $j>i+1$, and set $\nset''$ containing all remaining demand pairs (for each such demand pair $(s,t)$, the corresponding path $P(s,t)$ must contain a vertex in some set $U(\mu_j)$, for some $j<i-1$). It is enough to show that $|\nset'|,|\nset''|\leq 4d$. We show that $|\nset'|\leq 4d$; the bound for $\nset''$ is proved similarly. 
Given a pair $(s,t)\in \nset'$, and a pair $(\hat s,\hat t) \in \mset_{i+1}$, we say that there is a \emph{conflict} between $(s,t)$ and $(\hat s,\hat t)$ iff $P(s,t)$ contains a vertex of $Q_{\hat t}$.

We need the following simple observation.

\begin{observation}\label{obs: containing many U-vertices}
Let $(s,t)$ be any demand pair in $\nset'$. Then there are at least $12d$ demand pairs $(\hat s,\hat t)\in \mset_{i+1}$, such that there is a conflict between $(s,t)$ and $(\hat s,\hat t)$.
\end{observation}

Assume first that the observation is correct, and assume for contradiction that $|\nset'|>4d$. Then there are at least $48d^2$ pairs $((s,t),(\hat s,\hat t))$ with $(s,t)\in \nset'$ and 
$(\hat s,\hat t)\in \mset_{i+1}$, such that there is a conflict between $(s,t)$ and $(\hat s,\hat t)$. However, $|\mset_{i+1}|\leq 20d$, and, since, for each pair $(\hat s,\hat t)\in \mset_{i+1}$, $Q_{\hat t}$ contains at most $2d$ vertices, there may be at most $2d$ pairs $(s,t)\in \nset'$ that conflict with $(\hat s,\hat t)$, a contradiction. We now proceed to prove the observation.


Let $v$ be some vertex on path $P(s,t)$, that belongs to a set $U(\mu_j)$, for some $j>i+1$. Then $v$ lies on some path $Q_{t'}$, for some $t'\in T(\mset_j)$ (see Figure~\ref{fig: obs-proof1}). We construct a (not necessarily simple) curve $\gamma$, by concatenating the path $Q_t$; the sub-path of $\pi'$ between $\tilde t$ and $\tilde{t}'$; the sub-path of $Q_{t'}$ between $\tilde t'$ and $v$; and the sub-path of $P(s,t)$ between $v$ and $t$ (see Figure~\ref{fig: obs-proof2}). Consider now some demand pair $(\hat s,\hat t)\in \mset_{i+1}$. If path $P(s,t)$ does not contain a vertex of $Q_{\hat t}$, then curve $\gamma$ separates $\hat t$ from $\hat s$, and so path $P(\hat s,\hat t)$ has to cross the curve $\gamma$. This is only possible if path $P(\hat s,\hat t)$ contains a vertex of $Q_t\cup Q_{t'}$. But since both $Q_t$ and $Q_{t'}$ contain at most $2d$ vertices each, at most $4d$ paths of $\pset^*$ may intersect the curve $\gamma$. Therefore, there are at least $16d-4d= 12 d$ demand pairs in $\mset_{i+1}$, whose corresponding path in $\pset^*$ does not intersect the curve $\gamma$. From the above discussion, for each such pair $(\hat s,\hat t)$, path $P(s,t)$ must contain a vertex of $Q_{\hat t}$, and so $(s,t)$ conflicts with $(\hat s,\hat t)$.  We conclude that $(s,t)$ conflicts with at least  $12d$ pairs of $\mset_{i+1}$.

\begin{figure}[h]
\scalebox{0.4}{\includegraphics{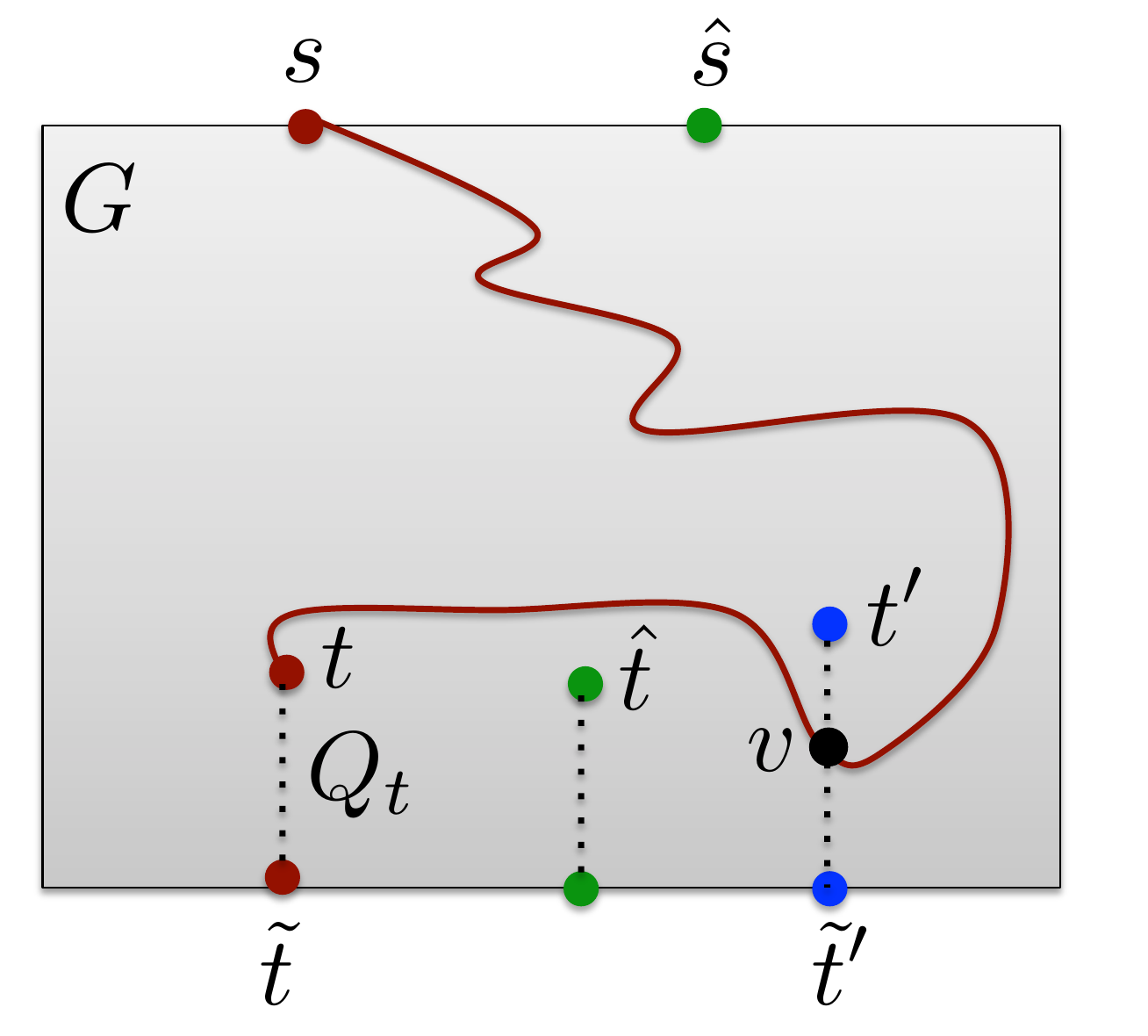}}
\caption{Illustration to the proof of Observation~\ref{obs: containing many U-vertices}. Demand pairs $(s,t)$ and $(\hat s,\hat t)$. Path $P(s,t)$ is shown in red.\label{fig: obs-proof1}}
\end{figure}

\begin{figure}[h]
\scalebox{0.4}{\includegraphics{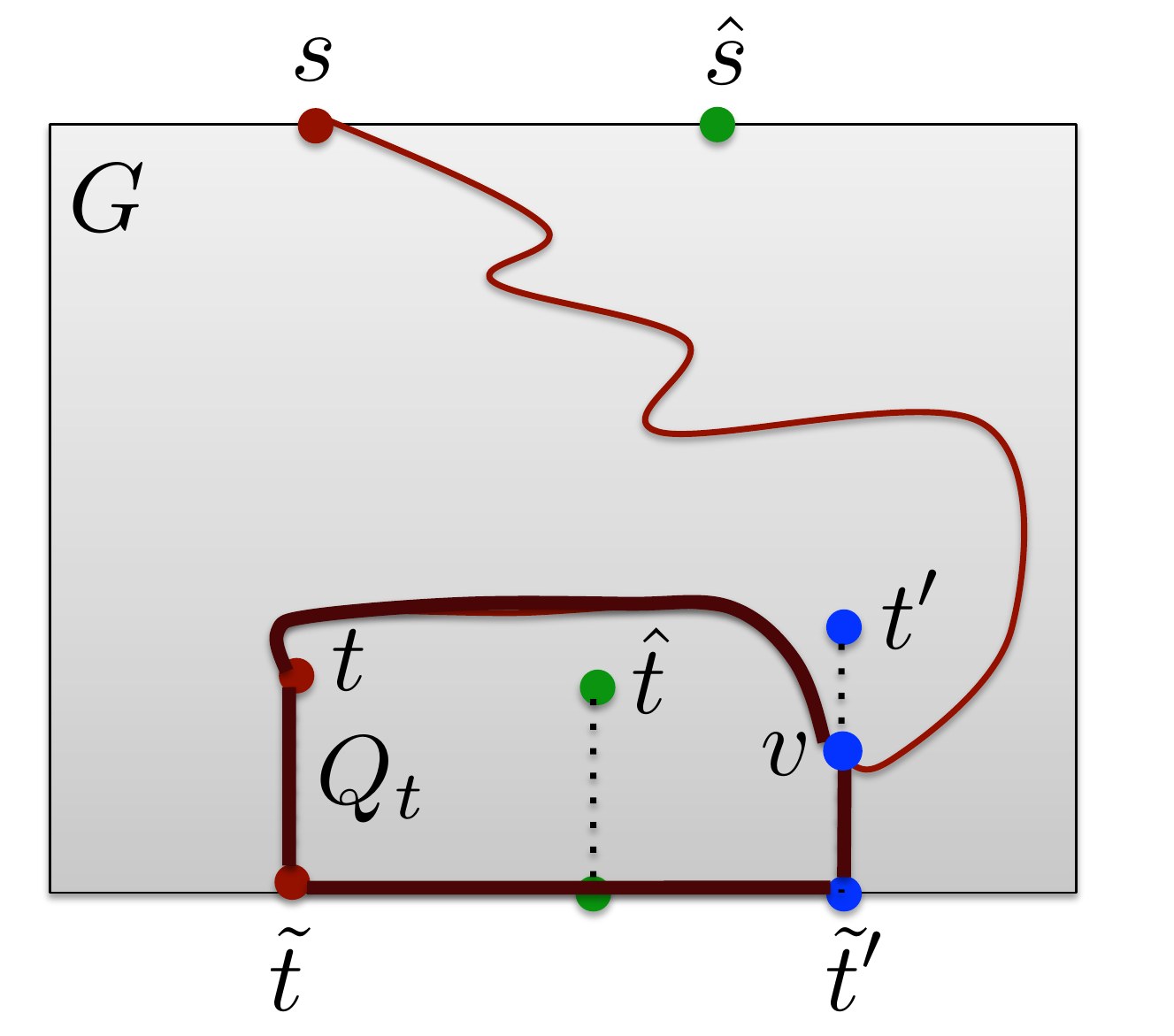}}
\caption{Illustration to the proof of Observation~\ref{obs: containing many U-vertices}. The curve $\gamma$ is shown in brown. 
\label{fig: obs-proof2}}
\end{figure}

\end{proofof}
\end{proofof}

\subsection{The Algorithm} \label{subsec: general case algo}
If all vertices of $T(\mset)$ lie on the boundary $\Gamma$ of $G$, then we can solve the problem efficiently using standard dynamic programming. Therefore, we assume from now on, that for each $t\in T(\mset)$, $d(t,\tilde t)\geq 1$.

Recall that $\Gamma_1,\ldots,\Gamma_4$ denote the four boundary edges of $G$.
We partition the set $\mset$ of demand pairs into $16$ subsets $\mset_{q,q'}$, for $1\leq q,q'\leq 4$: demand pair $(s,t)$ belongs to set $\mset_{q,q'}$ iff $s\in \Gamma_q$ and $\tilde t\in \Gamma_{q'}$. We then solve each of the resulting instances $(G,\mset_{q,q'})$ separately, and return the best of the resulting solutions. Since one of these instances is guaranteed to have a solution of value at least $\opt/16$, it is enough to show a factor-$2^{O(\sqrt{\log n}\log\log n)}$-approximation algorithm for each such instance separately.
 Therefore, from now on we focus on one such instance $(G,\mset_{q,q'})$, and we assume w.l.o.g. that the value of the optimal solution in this instance is at least $\opt/16$.

 We further partition the set $\mset_{q,q'}$ of demand pairs into the following subsets: we let $\mset^0_{q,q'}$ contain all demand pairs $(s,t)$ with $d(t,\tilde t)\geq \opt/\eta$, and for $0<r<\log n$, we let $\mset^r_{q,q'}$ contain all demand pairs $(s,t)$ with $\frac{\opt}{\eta\cdot 2^r} \leq d(t,\tilde t)<\frac{\opt}{\eta\cdot 2^{r-1}}$. Clearly, there is an index $0< r^*< \log n$, for which the value of the optimal solution of instance $(G,\mset^{r^*}_{q,q'})$ is at least $\Omega(\opt/\log n)$. Therefore, it is enough to compute a factor-$2^{O(\sqrt{\log n}\log\log n)}$-approximation for each such instance separately. Notice that we can compute a factor-$2^{O(\sqrt{\log n}\log\log n)}$-approximate solution for instance $\mset^0_{q,q'}$ using Theorem~\ref{thm: main w destinations far from boundary}. Therefore, we now focus on an instance $(G,\mset^{r^*}_{q,q'})$, for some fixed $r^*>0$. 
 
 In order to simplify the notation, we denote by $\mset$ the corresponding set $\mset^{r^*}_{q,q'}$ of demand pairs; we assume that we are given an integer $1\leq d\leq \opt/2\eta$, such that all destination vertices $t\in T(\mset)$ have $d\leq d(t,\tilde{t})< 2d$. We denote by $\opt'$ the value of the optimal solution of the resulting instance $(G,\mset)$, and we assume that $\Theta(\opt/\log n)\leq \opt'\leq \opt$. Moreover, we can assume that $\opt'>2^{13}d$, since otherwise we can directly apply Theorem~\ref{thm: main w destinations far from boundary} to instance $(G,\mset)$ with parameter $\opt=d$ to obtain an $2^{O(\sqrt{\log n}\log\log n)}$-approximation.  We now consider three cases, depending on the location of the boundary edges $\Gamma_q$ and $\Gamma_{q'}$.

 \paragraph*{Case 1: $\Gamma_q$ and $\Gamma_{q'}$ are opposite boundary edges.} Without loss of generality, we assume that $\Gamma_q$ is the top row of the grid $G$, and it contains the vertices of $S(\mset)$, while $\Gamma_{q'}$ is the bottom row of the grid, and it contains the vertices $\set{\tilde t\mid t\in T(\mset)}$.
 
Let $\mset'\subseteq \mset$ be the set of the demand pairs routed by the optimal solution. 
 Let $I$ denote the top boundary of the grid $G$, and let $I'$ denote its bottom boundary.
Theorem~\ref{thm: partition into intervals} guarantees the existence of a set  $\Sigma=\set{\sigma_1,\ldots,\sigma_z}$ of disjoint sub-intervals of $I$, and a set $\Sigma'=\set{\sigma_1',\ldots,\sigma'_z}$ of disjoint sub-intervals of $I'$, for $z=\ceil{\frac{|\mset'|}{160d}}-1$, such that the intervals $\sigma_1,\sigma_2,\ldots,\sigma_z,\sigma'_z,\ldots,\sigma_2',\sigma_1'$ appear on $\Gamma$ in this circular order;
  every pair $\sigma\in \Sigma\cup \Sigma'$, $\sigma'\in \Sigma'$ of  intervals is separated by at least $16d$ vertices of $\Gamma$; the length of each interval $\sigma'\in \Sigma'$ is at least $16d$ and at most $\sqrt{n}/2$, and for each $1\leq i\leq z$, if we consider the perfect $(\sigma_i,\sigma'_i,d)$-instance $(G,\tmset_i)$, where $\tmset_i=\mset_{\sigma_i,\sigma_{i'}}$, and denote by $(G_i,\tmset'_i)$ the corresponding modified instance, then $|\aset(G_i,\tmset'_i)|\geq d/2^{O(\sqrt{\log n}\log\log n)}$. Moreover, for every $1\leq i\leq z$, $(\sigma_i,\sigma_i')$ is a $d$-interesting pair of intervals. Notice, however, that, since the set $\mset'$ of demand pairs routed by the optimal solution is not known to us, we cannot compute the sets $\Sigma$ and $\Sigma'$ of intervals directly. We now show how to overcome this difficulty.
  
For every $d$-interesting pair $(\sigma,\sigma')$ of intervals, with $\sigma\subseteq I$, $\sigma'\subseteq I'$, consider the perfect $(\sigma,\sigma',d)$-instance $(G, \mset_{\sigma,\sigma'})$, and the corresponding modified instance $(G'_{\sigma,\sigma'},\mset'_{\sigma,\sigma'})$. We can then compute an approximate solution $
\aset(G'_{\sigma,\sigma'},\mset'_{\sigma,\sigma'})$ of value $|\aset(G'_{\sigma,\sigma'},\mset'_{\sigma,\sigma'})|$ to this instance, using the algorithm from Theorem~\ref{thm: main w destinations far from boundary}, together with the parameter $\opt=d$. 
We can assume w.l.o.g. that, if $\mset'_{\sigma,\sigma'}\neq\emptyset$, then $|\aset(G'_{\sigma,\sigma'},\mset'_{\sigma,\sigma'})|\geq 1$. Let $c$ be a constant, so that the approximation factor of Algorithm $\aset$ given in Theorem~\ref{thm: main w destinations far from boundary} is $2^{c\sqrt{\log n}\log\log n}$.
Assume now that we have computed the values  $|\aset(G'_{\sigma,\sigma'},\mset'_{\sigma,\sigma'})|$  for all $d$-interesting pairs $\sigma\subseteq I$ and $\sigma'\subseteq I'$ of intervals.
We say that a pair $(\sigma,\sigma')$ of interesting intervals is a \emph{good pair} iff $ |\aset(G'_{\sigma,\sigma'},\mset'_{\sigma,\sigma'})|\geq \ceil{d/2^{c \sqrt{\log n}\log\log n}}$.
It is now enough to compute a collection $\Sigma=\set{\sigma_1,\ldots,\sigma_z}$ of disjoint sub-intervals of $I$, and a collection $\Sigma'=\set{\sigma_1',\ldots,\sigma'_z}$ of disjoint sub-intervals of $I'$, for $z=\ceil{\frac{|\mset'|}{160d}}-1$, such that the intervals $\sigma_1,\sigma_2,\ldots,\sigma_z,\sigma'_z,\ldots,\sigma_2',\sigma_1'$ appear on $\Gamma$ in this circular order; for all $i$, $(\sigma_i,\sigma_i')$ is a good pair of intervals, and $\sigma_i'$ is at a distance at least $16d$ from every interval in $\Sigma\cup \Sigma'$. This can be done by using simple dynamic programming.

Assume now that we have computed the collections $\Sigma$ and $\Sigma'$ of intervals as above. For each $1\leq i\leq z$, let $\pset_i=\aset(G'_{\sigma_i,\sigma'_i},\mset'_{\sigma_i,\sigma'_i})$ be the set of paths routed by the solution of value $|\aset(G'_{\sigma_i,\sigma'_i},\mset'_{\sigma_i,\sigma'_i})|$ that we have computed. Let $\nset'_i\subseteq \mset'(\sigma_i,\sigma'_i)$ be the set of the demand pairs routed by this solution, and let $\nset_i\subseteq \mset(\sigma_i,\sigma'_i)$ be the set of the original demand pairs corresponding to the pairs in $\nset'_i$. For convenience, we denote $G_i=G'_{\sigma_i,\sigma'_i}$. Finally, let $\tilde{\nset}_i$ denote the set of pairs $(s,s')$, where $(s,t)\in \nset_i$, and $(s',t')\in \nset'_i$ is the corresponding demand pair in the modified instance. 

From the above discussion, $\sum_{i=1}^z|\nset_i|\geq \opt'/\approxfactor$.
It is now enough to show that all demand pairs in set $\bigcup_{i=1}^z\nset_i$ can be routed in $G$. In order to do so, it is enough to show that all pairs in $\bigcup_{i=1}^z\tilde{\nset}_i$ can be routed in graph $G$ via paths that are internally disjoint from $\bigcup_{i=1}^zG_i$, as we can exploit the paths in $\bigcup_{i=1}^z\pset_i$ in order to complete the routing. Recall that the graphs $G_i$ are disjoint from the top $\sqrt{n}/4$ rows of the grid, and the source vertices of the demand pairs in  $\bigcup_{i=1}^z\tilde{\nset}_i$ appear in the same left-to-right order as their destination vertices. It is now immediate to complete the routing of the demand pairs in $\bigcup_{i=1}^z\tilde{\nset}_i$.

\paragraph*{Case 2: $\Gamma_q$ and $\Gamma_{q'}$ are neighboring boundary edges.}
Assume w.l.o.g. that $\Gamma_q$ is the left boundary edge of $G$ and $\Gamma_{q'}$ is its bottom boundary edge. This case is dealt with very similarly to Case 1, with minor changes.
As before, we assume that we are given a set $\mset$ of demand pairs and an integer $d$, such that for each pair $(s,t)\in \mset$, $s\in \Gamma_q$, $\tilde t\in \Gamma_{q'}$, and $d\leq d(t,\tilde t)<2d$. We assume that we know the value $\opt'$ of the optimal solution to instance $(G,\mset)$, and we assume w.l.o.g. that $\opt'\geq 2^{13}d$. 
 We let $I$ be the left boundary edge of $G$, and we let $I'$ be the bottom boundary edge of $G$, excluding its $\ceil{\opt'/16}$ leftmost vertices. 
 Let $\mset'\subseteq \mset$ be the subset of all demand pairs $(s,t)$ with $\tilde{t}\in I'$. Let $\opt''$ be the value of the optimal solution to instance $(G,\mset')$.
 
 \begin{observation}
 $\opt''\geq \opt'/2$.
 \end{observation}
 
 \begin{proof}
 Let $\pset^*$ be the optimal solution to instance $(G,\mset)$, and let $\mset^*$ be the set of the demand pairs routed by it. Let $\hmset=\mset^*\setminus \mset'$. Then it is enough to show that $|\hmset|\leq \opt'/2$. Let $R''$ be the bottom row of $G$, and let $I''=R''\setminus I'$. Then for each demand pair $(s,t)\in \hmset$, $\tilde t\in I''$ must hold. Let $Q$ be a sub-grid of $G$, whose bottom boundary is $I''$, and whose height is $\ceil{\opt'/16}>2d$. Then the boundary of $Q$ has length at most $\opt'/4+4$, and every path routing a demand pair in $\hmset$ must cross the boundary of $Q$. Therefore, $|\hmset|\leq \opt'/2$.
 \end{proof}  
 
 The remainder of the algorithm is exactly the same as in Case 2. The only difference is in how the demand pairs in set $\tilde \nset=\bigcup_{i=1}^z\tilde{\nset}_i$ are routed. This is done by utilizing the $\opt'/16$ first columns of $G$ and the $\opt'/2$ top rows of $G$. Since we can assume that $|\tilde \nset|<\opt'/64$,  it is straightforward to find a suitable routing.

\paragraph*{Case 3: $\Gamma_q=\Gamma_{q'}$.}
We assume w.l.o.g. that $\Gamma_q$ is the bottom boundary of $G$. 
We assume that we are given a set $\mset$ of demand pairs and an integer $d$, such that for each pair $(s,t)\in \mset$, $s,\tilde t\in \Gamma_q$, and $d\leq d(t,\tilde t)<2d$. We assume that we know the value $\opt'$ of the optimal solution to instance $(G,\mset)$, and we assume w.l.o.g. that $\opt'\geq 2^{13}d$.

We partition the set $\mset$ of the demand pairs into three subsets: set $\mset^0$ contains all pairs $(s,t)$ with $|\col(s)-\col(t)|\leq 2d$; set $\mset^1$ contains all remaining pairs $(s,t)$ with $\col(s)<\col(t)$; and set $\mset^2$ contains all remaining demand pairs. 
We deal with the demand pairs in $\mset^0$ using the following claim.

\begin{claim}\label{claim: short pairs}
There is an efficient randomized algorithm that computes a factor-$2^{O(\sqrt{\log n}\log\log n)}$ approximation to instance $(G,\mset^0)$ of \NDP.
\end{claim}

\begin{proof}
Let $\rho$ be a random integer between $0$ and $4d$. We define a collection $\qset'=\set{Q_1,\ldots,Q_r}$ of square sub-grids of $G$ as follows. Let $Q_1$ be a sub-grid of size $((4d+\rho)\times (4d+\rho))$, containing the bottom left corner of $G$ as the bottom left corner of $Q_1$. Assume now that we have defined $Q_1,\ldots,Q_i$. Let $N$ be the number of columns of $G$ that do not intersect the grids of $Q_1,\ldots,Q_i$. If $N>8d$, then we let $Q_{i+1}$ be the $(4d\times 4d)$ sub-grid of $G$, whose bottom left corner is the vertex serving as the bottom right corner of $Q_i$. Otherwise, we set $r=i+1$, and we let $Q_r$ be the square grid whose bottom left corner is the  vertex serving as the bottom right corner of $Q_i$, and whose bottom right corner is the bottom right corner of the grid $G$. Notice that each grid $Q_i$ has width and height $4d$, except for $Q_1$ and $Q_r$, whose widths and heights may be between $4d$ and $8d$.

For $1\leq i\leq r$, let $\nset_i\subseteq \mset^0$ be the set of all demand pairs $(s,t)$ with $s,t\in V(Q_i)$. Let $\nset=\bigcup_{i=1}^r\nset_i$. Notice that each demand pair $(s,t)\in \mset^0$ belongs to $\nset$ with constant probability, and therefore, with high probability, $\opt(G,\nset)=\Omega(\opt(G,\mset^0))$.

We partition the set $\qset'$ of the squares into four subsets, $\qset^1,\ldots,\qset^4$, where for $1\leq i\leq 4$, $\qset^i$ contains all squares $Q_j$, where $j=i\mod 4$.
For $1\leq i\leq 4$, let $\nset^i\subseteq \nset$ be the subset of all demand pairs $(s,t)$, such that for some $Q\in \qset^i$, $s,t\in V(Q)$. We solve each problem $(G,\nset^i)$ separately and return the best of the resulting four solutions.

We now fix $1\leq i\leq 4$, and solve the problem $(G,\nset^i)$. Consider some square $Q_j\in \qset^i$. Let $Q^+_j$ be obtained by adding a margin of $4d$ columns to the left and to the right of $Q_j$, and $8d$ rows above $Q_j$ (if $j=1$ or $j=r$, then we do not add columns on one of the sides of $Q_j$, and we add only $4d$ rows above $Q_j$, so that $Q_j^+$ is a square grid). Notice that from our construction, if $Q_j,Q_{j'}\in \qset^i$, then $Q^+_j,Q^+_{j'}$ are disjoint. Using the same reasoning as in the proof of Claim~\ref{claim: modified instance preserves solutions}, $\opt(Q^+_j,\nset_j)\geq \opt(G,\nset_j)$. Notice that $\opt(G,\nset_j)\leq 4d$, as all source vertices of the demand pairs routed by the optimal solution must be distinct, and they lie on the bottom boundary of $Q_j$. It is then easy to see that $(Q^+,\nset_j)$ is a special case of \restrictedNDP where all destination vertices lie at a distance at least $\opt/4$ from the grid boundary, and we can find a $2^{O(\sqrt{\log n}\log\log n)}$-approximation for it using Theorem~\ref{thm: main w destinations far from boundary}.
\end{proof}

If the optimal solution value to instance $(G,\mset^0)$ is at least $\opt'/3$, then we obtain a $2^{O(\sqrt{\log n}\log\log n)}$-approximation to this instance, and a $\approxfactor$-approximation overall, using Claim~\ref{claim: short pairs}. Therefore, we assume from now on that any optimal solution must route at least $2\opt'/3$ demand pairs from $\mset^1\cup \mset^2$.
We compute an approximate solution to each of the problems $(G,\mset^1)$ and $(G,\mset^2)$ separately, and take the better of the two solutions. Clearly, it is enough to design a factor-$2^{O(\sqrt{\log n}\log\log n)}$-approximation to each of these two problems separately. We now focus on one of these problem, say $(G,\mset^1)$, and for simplicity denote $\mset^1$ by $\mset$ and the value of the optimal solution to instance $(G,\mset^1)$ by $\opt'$.

Recall that $\ell=\sqrt{n}$ is the length of the grid $G$. We can naturally associate with each demand pair $(s,t)\in \mset$ an interval $I(s,t)\subseteq [\ell]$: the left endpoint of $I(s,t)$ is the index of $\col(s)$, and its right endpoint is the index of $\col(t)$.

Next, we partition the demand pairs in $\mset$ into $h=\log n$ classes $\mset_1,\mset_2,\ldots, \mset_h$, as follows. Pair $(s,t)$ belongs to class $\mset_i$ iff $2^{i-1}\leq |I(s,t)|<2^i$. Note that for $i\leq \log d$, $\mset_i=\emptyset$, as all such pairs belong to $\mset^0$. Thus, we obtain a collection of $h$ instances $(G,\mset_i)$ of the problem. As before, we will compute a factor-$2^{O(\sqrt{\log n}\log\log n)}$-approximation to each of these instances separately, and return the best of the resulting solutions. We now fix some index $\log d\leq i\leq h$, and design a factor-$2^{O(\sqrt{\log n}\log\log n)}$-approximation algorithm to the corresponding instance $(G,\mset_i)$. 


Let $\rho$ be an integer chosen uniformly at random from $[0,2^{i+3})$. Let $Z$ be the set of all integers $1\leq z\leq \ell$, such that $z=\rho+j\cdot 2^{i+3}$ for some integer $j$.
 Let $\mset'\subseteq \mset$ contain all demand pairs $(s,t)$, such that there is some number $z\in Z$ with $\col(s)<z<\col(t)$, and for all $z'\in Z$, $|\col(s)-z'|,|\col(t)-z'|\geq 2^{i-1}/4$. Let $I(s,t)$ be an interval whose left endpoint is $\col(s)$ and right endpoint is $\col(t)$, and let $I'(s,t)\subseteq I(s,t)$ be its sub-interval that excludes the first and the last $2^{i-1}/4$ vertices of $I(s,t)$. Then it is easy to see that $(s,t)\in \mset'$ iff some integer in $I'(s,t)$ belongs to $Z$, and that this happens with probability at least $1/32$ over the choice of $\rho$. We now focus on solving the instance $(G,\mset')$ of the problem, and we denote by $\opt''$ the value of its optimal solution, so $\expect{\opt''}=\Omega(\opt')$. 

Let $Z'$ contain all integers that lie halfway between consecutive pairs of integers of $Z$, that is: $Z'=\set{z+2^{i+2}\mid z\in Z \mbox{ and } z+2^{i+2}\leq \ell}$. We partition the grid $G$ into sub-grids $H_1,\ldots, H_r$, by deleting all columns $W_z$, where $z\in Z'$. 
Then for every demand pair $(s,t)\in \mset'$, there is a unique sub-grid $H_j$ containing both $s$ and $t$. For each $1\leq j\leq r$, let $\hmset_j\subseteq \mset'$ be the set of all demand pairs $(s,t)$ with both $s$ and $t$ contained in $H_j$. For each $1\leq j\leq r$, we then consider the instance $(G,\hmset_j)$. We can further partition $H_j$ into two sub-grids $H'_j$ and $H''_j$ along the unique column $W_z$ that is contained in $H_j$, and for which $z\in Z$. Notice that all source vertices lie on the bottom boundary of $H'_j$ at distance at least $2^{i-3}$ from the left and the right boundaries of $H'_j$, while all destination vertices lie in $H''_j$, at distance at least $2^{i-3}$ from the left and the right boundaries of $H''_j$. Let $I$ and $I'$ denote the bottom boundaries of $H'_j$ and $H''_j$, respectively.

Assume first that the value $\opt_j$ of the optimal solution to instance $(G,\hmset_j)$ is at least $1024d$. Then we employ an algorithm similar to the one for Case 1: we use dynamic programming, together with the algorithm given by Theorem~\ref{thm: main w destinations far from boundary} to compute  a collection $\Sigma=\set{\sigma_1,\ldots,\sigma_z}$ of disjoint sub-intervals of $I$, and a set $\Sigma'=\set{\sigma_1',\ldots,\sigma'_z}$ of disjoint sub-intervals of $I'$, for $z=\ceil{\frac{|\mset'|}{160d}}-1$, such that the intervals $\sigma_1,\sigma_2,\ldots,\sigma_z,\sigma'_z,\ldots,\sigma_2',\sigma_1'$ appear on $\Gamma$ in this circular order, and for all $q$, $(\sigma_q,\sigma_q')$ is a good pair of intervals. For each $1\leq q\leq z$, let $(G'_{\sigma_q,\sigma'_q},\hmset'_{\sigma_q,\sigma'_q})$ denote the corresponding modified instance. The key is to notice that all such graphs $G'_{\sigma_q,\sigma'_q}$ are mutually disjoint and are contained in $H''_j$. We construct the set $\tilde {\nset}=\bigcup_{q=1}^z\tilde{\nset}_q$ of new demand pairs as before, and show that they can be routed in $H_j$ via paths that are internally disjoint from $\bigcup_{q=1}^zG'_{\sigma_q,\sigma'_q}$. The routing is straightforward and exploits the columns of $H'_j$ and the top $\sqrt{n}/4$ rows of $H''_j$.

Finally, assume that the value of the optimal solution $\opt_j$ to instance  $(G,\hmset_j)$ is less than $1024d$. 
Note that $\opt_j<2^{i+2}$ must also hold, as interval $I$ contains at most $2^{i+2}$ vertices.  

Let $\hat I'$ be obtained from $I'$ by discarding its first $2^{i-3}$ and its last $2^{i-3}$ vertices, and recall that $2^i\geq d$. Then $(I,\hat I')$ is a $d$-interesting pair of intervals, and $(G,\hmset_j)$ is a perfect $(I,\hat I',d)$ sub-instance of $(G,\mset)$.
Let $(\tilde{H}_j,\hmset'_j)$ be the corresponding modified instance. Notice that $\tilde{H}_j\subseteq H''_j$. Then, from Claim~\ref{claim: modified instance preserves solutions}, the value of the optimal solution to instance $(\tilde H_j,\hmset'_j)$ is at least $\min\set{\opt_j,d}=\Omega(\opt_j)$. We can use  Theorem~\ref{thm: main w destinations far from boundary} in order to compute a $2^{O(\sqrt{\log n}\log\log n)}$-approximate solution $\pset_j$ to instance $(H'_j,\hmset'_j)$. From the above discussion $|\pset_j|\geq \frac{\opt_j}{\approxfactor}$. It now remains to route the vertices of $S(\hmset_j)$ to the corresponding vertices of $S(\hmset'_j)$ via node-disjoint paths that are internally disjoint from $\tilde H_j$. This is done in exactly the same way as before.

\label{------------------------------------sec: ndp tp edp--------------------------------------}
\section{Approximating NDP/EDP on Walls with Sources on the Boundary}\label{sec: ndp to edp}
In this section we prove Theorem \ref{thm: main edp}, by extending our results for \rNDPgrid to \EDP and \NDP on wall graphs.

\begin{figure}[h]
\center
\scalebox{0.20}{\includegraphics{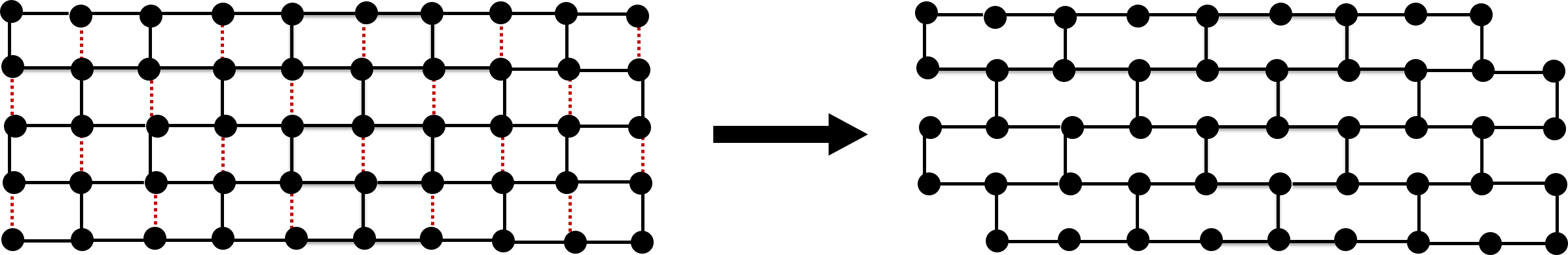}}
\caption{The red dotted edges are deleted to obtain wall $\hat G$ from grid $G$.
\label{fig: dotted wall}}
\end{figure}

Let $G=G^{\ell,h}$ be a grid of length $\ell$ and height $h$.
Assume that $\ell>0$ is an even integer, and that $h>0$.
For every column $W_j$ of the grid, let $e^j_1,\ldots,e^j_{h-1}$ be the edges of $W_j$ indexed in their top-to-bottom order. Let $E^*(G)\subseteq E(G)$ contain all edges $e^j_z$, where $z\neq j \mod 2$, and let $\hat G$ be the graph obtained from $G\setminus E^*(G)$, by deleting all degree-$1$ vertices from it.
Graph $\hat G$ is called a \emph{wall of length $\ell/2$ and height $h$} (see Figure~\ref{fig: dotted wall}).
Consider the subgraph of $\hat G$ induced by all horizontal edges of the grid $G$ that belong to $\hat G$. This graph is a collection of $h$ node-disjoint paths, that we refer to as the \emph{rows} of $\hat G$, and denote them by $R_1,\ldots,R_h$ in this top-to-bottom order; notice that $R_j$ is a sub-path of the $j$th row of $G$ for all $j$. Graph $\hat G$ contains a unique collection $\wset$ of $\ell/2$ node-disjoint paths that connect vertices of $R_1$ to vertices of $R_h$ and are internally disjoint from $R_1$ and $R_h$. We refer to the paths in $\wset$ as the \emph{columns} of $\hat G$, and denote them by $W_1,\ldots,W_{\ell/2}$ in this left-to-right order.
Paths $W_1, W_{\ell/2},R_1$ and $R_h$ are called the left, right, top and bottom boundary edges of $\hat G$, respectively, and their union is the boundary of $\hat G$, that we denote by $\Gamma(\hat G)$.

In both \NDPwall and \EDPwall problems, the input is a wall $\hat G$ of length $\sqrt n/2$ and height $\sqrt n$ (we assume that $\sqrt n/2$ is an integer), and a collection $\mset = \set{(s_1,t_1), \ldots, (s_k, t_k)}$ of $k$ demand pairs. The goal is to route the maximum number of the demand pairs via paths that are node-disjoint (for \NDPwall) or edge-disjoint (for \EDPwall). We will exploit the $(\sqrt n\times \sqrt n)$ grid $G$, and the fact that $\hat G \subseteq G$ can be obtained from $G$ as described above. 

For any two nodes $u,v$ of $\hat G$, we denote the shortest-path distance between them in $\hat G$ by $\hat d(u,v)$. Distances between subsets of vertices, and between a vertex and a subset of vertices, are defined as before, using $\hat d$.

The following simple observation relates the values of the optimal solutions to the \EDPwall and \NDPwall problems, defined over the same wall $\hat G$ and the same set $\mset$ of demand pairs.

\begin{observation}\label{obs: NDP vs EDP in walls}
Let $\hat G$ be a wall of length $\sqrt n/2$ and height $\sqrt n$, and let $\mset = \set{(s_1,t_1), \ldots, (s_k, t_k)}$ be a set of $k$ demand pairs. Let $\optEDP$ be the value of the optimal solution to the \EDPwall problem instance $(\hat G, \mset)$, and let $\optNDP$ be the value of the optimal solution to the \NDPwall problem instance $(\hat G, \mset)$. Then: $\optNDP\leq \optEDP\leq O(\optNDP)$.
\end{observation}
\begin{proof}
Observe first that for any set $\pset$ of paths in $\hat G$, if the paths in $\pset$ are mutually node-disjoint, then they are also mutually edge-disjoint. Therefore, $\optNDP\leq \optEDP$. 

To show the other direction, consider any set $\pset$ of paths, such that the paths in $\pset$ are edge-disjoint. It is enough to show that there is a subset $\pset'\subseteq\pset$ of paths that are node-disjoint, and $|\pset'|\geq \Omega(|\pset|)$. Since the maximum vertex degree in $\hat G$ is $3$, the only way for two paths $P,P'\in \pset$ to share a vertex $x$ is when $x$ is an endpoint of at least one of the two paths.

We construct a directed graph $H$, whose vertex set is $\set{v_P\mid P\in \pset}$, and there is an edge $(v_P,v_{P'})$ iff an endpoint of $P'$ belongs to $P$. It is immediate to verify that the maximum in-degree of any vertex in $H$ is $4$. Therefore, there is a set $U\subseteq V(H)$ of $\Omega(|V_H|)$ vertices, such that no two vertices of $U$ are connected with an edge. We let $\pset'=\set{P\mid v_P\in U}$. Then the paths in $\pset'$ are node-disjoint, and $|\pset'|=\Omega(|\pset|)$.
\end{proof}

Since for any set $\pset$ of node-disjoint paths, the paths in $\pset$ are also mutually edge-disjoint, it is now enough to prove Theorem~\ref{thm: main edp} for the \NDPwall problem, with all source vertices lying on the wall boundary. We do so in the remainder of this section.

Our main approach is to consider the corresponding instance $(G,\mset)$ of the \NDPgrid problem on the underlying grid $G$, and to exploit our algorithm for this problem.
However, we cannot use this approach directly, since the boundary of the wall $\hat G$ is not contained in the boundary of the grid $G$, and so some of the source vertices of  $S(\mset)$ may not lie on $\Gamma(G)$. We overcome this difficulty by mapping each vertex in $S(\mset)$ to its closest vertex lying on $\Gamma(G)$. Formally, we define a new set $\mset'$ of demand pairs, containing, for each pair $(s,t)\in \mset$, a new pair $(s',t')$, with $t'=t$, and $s'$ defined as follows. If $s\in \Gamma(G)$, then $s'=s$; otherwise, either $s$ belongs to the first column of $\hat G$, in which case $s'$ is defined to be the vertex of $\Gamma(G)$ lying immediately to the left of $s$; or $s$ belongs to the last column of $\hat G$, in which case $s'$ is defined to be the vertex of $\Gamma(G)$ lying immediately to its right (see Figure~\ref{fig: dotted wall}). The following simple observation shows that the solution value does not change by much.

\begin{observation} \label{obs: perturbed edp}
Let $\opt$ be the value of the optimal solution to instance $(\hat G,\mset)$ of \NDPwall, and let $\opt'$ be defined in the same way for instance $(\hat G,\mset')$. Then $\opt'=\Omega(\opt)$. Moreover, given any set $\pset'$ of node-disjoint paths routing a subset of the demand pairs in $\mset'$, there is an efficient algorithm to compute a set $\pset$ of $\Omega(|\pset'|)$ node-disjoint paths routing a subsets of the demand pairs in $\mset$.
 \end{observation}
 
 \begin{proof}
 In order to prove the first assertion, let $\pset^*$ be the optimal solution to instance $(\hat G,\mset)$. We show that there is a set $\pset'$ of $\Omega(|\pset^*|)$ paths routing demand pairs in $\mset'$. Let $\pset''$ be obtained from $\pset^*$ as follows. Consider any path $P\in \pset^*$, and assume that it routes some pair $(s,t)\in \mset$. If $(s,t)\in \mset'$, then we add $P$ to $\pset''$; otherwise, we extend $P$ by adding the edge $(s,s')$ to it, and add the resulting path to $\pset''$. Consider the final set $\pset''$ of paths. While $|\pset''|=|\pset^*|$, it is possible that the paths in $\pset''$ are no longer node-disjoint. However, a pair $P_1,P_2\in \pset''$ of paths may share a vertex $x$ iff $x$ is an endpoint of at least one of the two paths. We can now employ the same argument as in the proof of Observation~\ref{obs: NDP vs EDP in walls} to obtain a subset $\pset'\subseteq \pset''$ of node-disjoint paths, with $|\pset'|\geq\Omega(\pset'')$. 
 
 The proof of the second assertion is almost identical.
 \end{proof}

From now on, it is sufficient to design a randomized $\approxfactor$-approximation algorithm to the new instance $(\hat G,\mset')$ of \NDPwall. Therefore, from now on we will assume that all source vertices lie on $\Gamma(\hat G)\cap \Gamma(G)$.


 Our first step is, as before, to consider a special case of the problem, where the destination vertices lie far from the boundary of the wall. We prove the following analogue of Theorem~\ref{thm: main w destinations far from boundary}.

\begin{theorem} \label{thm: far edp}
  There is an efficient randomized algorithm $\hat \aset$, that, given an instance $(\hat G, \mset)$ of the \NDPwall
  problem, with all sources lying on $\Gamma(G)\cap \Gamma(\hat G)$ (where $G$ is the grid corresponding to $\hat G$), and an integer $\opt$, such that the value of the optimal solution to instance $(\hat G,\mset)$ of \NDPwall  is at least $\opt$, and every destination vertex lies at a distance at least $\opt/\eta$ from $\Gamma(\hat G)$ in $\hat G$, returns a solution that routes at least $\opt/\approxfactor$ demand pairs in $\hat G$, with high probability.
\end{theorem}

\begin{proof}
%
 Consider the \rNDPgrid instance defined over the grid $G$, with the set $\mset$ of the demand pairs. Since $\hat G\subseteq G$, the value of the optimal solution to this instance is at least $\opt$. We can apply  Algorithm $\aset$  to instance $(G,\mset)$ of \rNDPgrid, to obtain a set
 $\pset$ of node-disjoint paths routing some subset $\tmset\subseteq \mset$ of demand pairs in graph $G$, such that with high probability, $|\tmset|\geq \opt / \approxfactor$.
We now show that all demand pairs in $\tmset$ can also be routed in the wall $\hat G$.
Indeed, in the analysis in Section \ref{sec: snakes}, we have constructed disjoint level-$\rho$ snakes of width at least $3$, such that each demand pair $(s,t)\in \tmset$ is contained in a distinct snake, that we denote by $\yset(s,t)$.
With a loss of an additive constant in the approximation guarantee, we can ensure that no snake passes through any corner of $G$.

\begin{observation}
  Let $\yset=(\Y_1,\ldots,\Y_z)$ be any snake of width at least $3$ in the grid $G$. Let $H$ be the union of all graphs $\Y_j\cap \hat G$, for $1\leq j\leq z$. Then $H$ is a connected sub-graph of $\hat G$.
\end{observation}

It is now immediate to obtain a routing of the demand pairs in $\tmset$ via node-disjoint paths in $\hat G$. For every snake $\yset(s,t)$, let $H(s,t)$ be the corresponding sub-graph of $\hat G$, given by the above observation. Then all graphs $\set{H(s,t)}_{(s,t)\in \tmset}$ are mutually disjoint, and each such graph is connected. We simply connect $s$ to $t$ by any path contained in $H(s,t)$.\end{proof}

Finally, we complete the proof of Theorem~\ref{thm: main edp}, by removing the assumption that the destination vertices lie far from the wall boundaries. As before, we assume that we are given an instance $(\hat G,\mset)$ of \NDPwall, such that, if $G$ denotes the corresponding grid graph, then all source vertices in $S(\mset)$ lie on $\Gamma(G)\cap \Gamma(\hat G)$. As before, we use the algorithm from Section \ref{sec:destinations anywhere} on the underlying grid $G$ and adapt the resulting solution to the wall $\hat G$.

Recall that for every destination vertex $t \in T(\mset)$, we denoted by $\tilde t$ the vertex of $\Gamma(G)$ minimizing $d(t,\tilde t)$, breaking ties arbitrarily.
For any subset $\mset' \subseteq \mset$ of demand pairs, we denoted by $\tilde T(\mset') = \set{\tilde t \mid t \in T(\mset')}$.
If all vertices of $\tset(\mset)$ lie on the boundary $\Gamma(\hat G)$ of wall $\hat G$, then we can efficiently obtain a constant-factor approximation for the problem using standard dynamic programming techniques.
Thus, we assume from now on that $\hat d(t, \tilde t) \geq 1$ for each $t \in T(\mset)$.

Recall that we showed that it suffices to consider the case where all sources $S(\mset)$ lie on a single boundary edge $\Gamma_q$ of $G$, and all vertices of $\tilde T(\mset)$ lie on a single boundary edge $\Gamma_{q'}$ of $G$, where possibly $\Gamma_q=\Gamma_{q'}$.
We further saw that it suffices to consider only two cases:
(i) $d(t, \tilde t) \geq \opt/\eta$ for all destination nodes $t$; or
(ii) we are given an integer $1 \leq d \leq \opt/2\eta$ such that $d \leq d(t, \tilde t) < 2d$ for all destination nodes $t$, and $\opt > 2^{13}d$.
In case (i), we can compute a factor $\approxfactor$-approximate solution directly using Theorem \ref{thm: far edp}.
In the remainder of this section, we consider case (ii).

Let $\mset'$ be the set of the demand pairs that was chosen to be routed by our algorithm from Section~\ref{sec:destinations anywhere}. Then all vertices of $\tilde T(\mset')$ lie on the same boundary edge of the grid -- we assume w.l.o.g. that it is the bottom boundary. Recall that we have defined a collection $\set{(G_1,\mset_1),\ldots,(G_z,\mset_z)}$ of modified sub-instances, such that all grids $G_1,\ldots,G_z$ are disjoint from each other and from the top $\sqrt n/4$ rows of $G$, and in each of the instances $(G_i,\mset_i)$, the destination vertices lie far enough from the boundary of the grid $G_i$, so that Theorem~\ref{thm: main w destinations far from boundary} could be applied  to each such sub-instance separately. Let $\mset_i'\subseteq \mset'$ be the subset of the demand pairs corresponding to the pairs in $\mset_i$.  For each $1\leq i\leq z$, we have also defined a set $\tilde N_i$ of pairs, connecting the original source vertices in $S(\mset'_i)$ to their corresponding source vertices in $S(\mset_i)$. We have implicitly constructed, for each $1\leq i\leq z$, a snake $\yset_i$, inside which the demand pairs in $\tilde N_i$ are routed. All snakes $\yset_1,\ldots,\yset_z$ are disjoint, and have width at least $3|\nset_i|$ each. We can now translate this routing into a set of node-disjoint paths routing the demand pairs in $\mset'$ in the wall $\hat G$. The routing inside each sub-grid $G_i$ is altered exactly like in the proof of Theorem~\ref{thm: far edp}; the routings of the sets $\tilde{\nset}_i$ of demand pairs exploit the same snakes $\yset_i$ inside the wall $\hat G$.

\section{Approximation Algorithm for the Special Case with Sources Close to the Grid Boundary} \label{sec: sources at dist d}
The goal of this section is to prove Theorem~\ref{thm: sources at dist d}. We assume that we are given an instance $(G,\mset)$ of \NDPgrid and an integer $\delta>0$, such that every source vertex is at a distance at most $\delta$ from the grid boundary. Our goal is to design an efficient randomized factor-$(\delta \cdot \approxfactor)$-approximation algorithm for this special case of the problem. We can assume that $\delta<n^{1/4}$, as there is an $\tilde{O}(n^{1/4})$-approximation algorithm for \NDPgrid~\cite{NDP-grids}. For every source vertex $s\in S(\mset)$, let $\tilde s$ be the vertex lying closest to $s$ on $\Gamma(G)$; for every destination vertex $t\in T(\mset)$, let $\tilde t$ be defined similarly. Recall that for each $s\in S(\mset)$, $d(s,\tilde s)\leq \delta$.
For every subset $\mset'\subseteq \mset$ of the demand pairs, we denote $\tilde S(\mset')=\set{\tilde s\mid s\in S(\mset')}$, and similarly $\tilde T(\mset')=\set{\tilde t\mid t\in T(\mset')}$.
 Using the same arguments as in Section~\ref{sec:destinations anywhere}, at the cost of losing an $O(\log n)$-factor in the approximation ratio, we can assume that all vertices of $\tilde S(\mset)$ are contained in a single boundary edge of the grid $G$, that we denote by $\Gamma$, and that  all vertices of $\tilde T(\mset)$ are contained in a single boundary edge of the grid $G$, that we denote by $\Gamma'$, where possibly $\Gamma=\Gamma'$. Moreover, we can assume that there is some integer $d$, such that for all $t\in T(\mset)$, $d\leq d(t,\tilde t)<2d$.

Let $c$ be a sufficiently large constant. Assume first that $d\leq \delta\cdot 2^{c\sqrt{\log n}\log\log n}$. In this case, we can obtain an $(\delta\cdot\approxfactor)$-approximation using the following claim.

\begin{claim} \label{claim: close pairs}
There is an efficient algorithm to compute a factor $(\delta\cdot\approxfactor)$-approximation for the special case where $d \leq \delta\cdot 2^{c\sqrt{\log n}\log\log n}$.
\end{claim}

\begin{proof}
Let $\hat \mset=\set{(\tilde s,\tilde t)\mid (s,t)\in \mset}$ be a new set of demand pairs. Note that all vertices participating in the demand pairs in $\hat \mset$ lie on the boundary of $G$. Therefore, we can efficiently find an optimal solution to the \NDP problem instance $(G,\hat \mset)$ using standard dynamic programming.  Next, we show that $\opt(G,\hat \mset)=\Omega(\opt(G,\mset)/(\delta\cdot\approxfactor))$. For every vertex $v\in S(\mset)\cup T(\mset)$, let $U_v$ be the shortest path connecting $v$ to $\tilde v$.

Let $\pset^*$ be the optimal solution to instance $(G,\mset)$, and let $\mset^*\subseteq \mset$ be the set of the demand pairs routed by $\pset^*$. For each demand pair $(s,t)\in \mset^*$, let $P(s,t)\in \pset^*$ be the path routing it. 

We construct a directed graph $H$, whose vertex set is $V(H)=\set{v(s,t)\mid (s,t)\in \mset^*}$, and there is a directed edge from $v(s,t)$ to $v(s',t')$ iff the path $P(s',t')$ routing the pair $(s',t')$ in the optimal solution $\pset^*$ contains a vertex of $U_s\cup U_t$. It is easy to verify that the out-degree of every vertex of $H$ is at most $d+\delta$, and so there is an independent set $I$ in $H$ containing at least $\Omega(|V(H)|/(d+\delta))=\Omega(\opt(G,\mset)/(\delta\cdot\approxfactor))$ vertices of $I$. 

For each demand pair $(s,t)\in \mset^*$ with $v(s,t)\in I$, let $P'(s,t)$ be the concatenation of $U_s$, $P(s,t)$ and $U_t$. Then $\set{P'(s,t)\mid v(s,t)\in I}$ is a collection of $\Omega(\opt(G,\mset)/(\delta\cdot\approxfactor))$ node-disjoint paths routing demand pairs in $\hmset$. We conclude that  $\opt(G,\hat \mset)=\Omega(\opt(G,\mset)/(\delta\cdot\approxfactor))$. 

Let $\pset'$ be a solution to instance $(G,\hat \mset)$, obtained by the constant-factor approximation algorithm, so $|\pset'|=\Omega(\opt(G,\mset)/(\delta\cdot\approxfactor))$, and let $\hmset'\subseteq \hmset$ be the set of the demand pairs routed. Since all source and all destination  of $\hmset'$ appear on the boundary of $G$, $|\hmset'|=O(\sqrt n)$. Moreover, the demand pairs in $\hmset'$ must be non-crossing: that is, for every pair $(s,t),(s',t')\in \hmset'$, the circular ordering of the corresponding terminals on the boundary of the grid is either $(s,s',t',t)$ or $(s,t,t',s')$. Since we have assumed that $\delta\leq n^{1/4}$  and $d\leq \delta\cdot 2^{c\cdot\sqrt{\log n}\log\log n}$, it is easy to verify that we can route $|\Omega(\hmset')|$ demand pairs from the original set $\mset$, that correspond to the demand pairs in $\hmset'$.

\end{proof}

From now on we assume that $d> \delta\cdot 2^{c\sqrt{\log n}\log\log n}$. 
Next, we define a new set $\tmset$ of demand pairs, as follows: $\tmset=\set{(\tilde s,t)\mid (s,t)\in \mset}$, so all source vertices of the demand pairs in $\tmset$ lie on the boundary edge $\Gamma$ of $G$. Let $\opt'$ be the value of the optimal solution to problem $(G,\tmset)$. The proof of the following claim is  almost identical to the proof of Claim~\ref{claim: close pairs} and is omitted here.

\begin{claim}
$\opt'\geq \Omega( \opt(G,\mset)/\delta)$.
\end{claim}

From now on we focus on instance $(G,\tmset)$ of \restrictedNDP. Recall that for each source vertex $s\in S(\mset)$, we denoted by $U_s$ the shortest path connecting $s$ to $\tilde s$. We say that a path $P$ routing a demand pair $(\tilde s,t)\in \tmset$ is \emph{canonical} iff $U_s\subseteq P$. We say that set $\pset$ of node-disjoint paths routing demand pairs in $\tmset$ is canonical iff every path in $\pset$ is canonical. We will show that we can apply our $\approxfactor$-approximation algorithm to instance $(G,\tmset)$ to obtain a routing in which all paths are canonical. 
For convenience, we assume w.l.o.g. that $\Gamma$ is the top boundary edge of the grid. Therefore, a canonical path connecting a demand pair $(\tilde s,t)\in \tmset$ must follow the column $\col(\tilde s)$ of $G$ until it reaches the vertex $s$.

Assume first that $d>\opt'/\eta$, where $\eta$ is defined as before. Then we obtain a special case of the problem where the destination vertices are sufficiently far from the boundary of the grid, and can use the algorithm from Theorem~\ref{thm: main w destinations far from boundary} to find a $\approxfactor$-approximate solution to instance $(G,\tmset)$.
%
However, we slightly modify the routing to ensure that the paths in our solution are canonical. Recall that the routing itself is recursive. Let $\qset_1$ be the set of all level-$1$ squares, and let $\qset_1'\subseteq\qset_1$ be the set of all non-empty level-$1$ squares. Recall that each square $Q\in \qset_1$ has size $(d_1\times d_1)$, where $d_1\leq \opt'/2^{c^*\sqrt{\log n}\log\log n}$. For each square $Q\in \qset_1'$, we have created an extended square $Q^+$, by adding a margin of size $d_1/\eta$ around $Q$. At the highest level of the recursion, the paths depart from the source vertices of the demand pairs we have chosen to route, and then visit the squares in $\set{Q^+\mid Q\in \qset_1'}$ one-by-one, in a snake-like fashion, entering and leaving each such square $Q^+$ on a pre-selected set of vertices. Recall that we have assumed that $d>\delta\cdot 2^{c\sqrt{\log n}\log\log n}$, and $d>\opt'/\eta$.  Then it is easy to see that the top $d/2\gg \delta$ rows of $G$ are disjoint from the squares of $\qset_1'$, each of which must contain a destination vertex. As the algorithm routes at most $d_1$ paths, this routing can be accomplished via canonical paths. The recursively defined routing inside the level-$1$ squares remains unchanged.




From now on, we assume that $d<\opt'/\eta$, and we now follow the algorithm from Section~\ref{sec:destinations anywhere}.

As before, we consider  three sub-cases. Recall that $\Gamma'$ is the boundary edge of $G$ containing the vertices of $\tilde T(\tmset)$.
In the first sub-case, $\Gamma'$ is the the bottom boundary of the grid. The algorithm constructs a number of square grids of height less than $\ell/2$, whose bottom boundaries are contained in the bottom boundary of the grid $G$. Each such sub-grid defines a modified instance, that is solved separately. Eventually, we need to connect the source vertices of the routed demand pairs to their counterparts in the modified instances. As the top $\ell/2$ rows of $G$ are disjoint from the modified instances, it is easy to see that this can be done via paths that are canonical.

The second sub-case is when $\Gamma'$ is either the left or the right boundary edge of $G$; assume w.l.o.g. that it is the left boundary edge. In this case, we have discarded the demand pairs whose destinations lie in the top $\opt'/16$ rows of the grid $G$. As before, we then define a  collection of modified instances, each of which is defined over a square sub-grid of $G$ of width less than $\ell/2$. The left boundary of each sub-grid is contained in the left boundary of $G$, and each such sub-grid is disjoint from the top $\opt'/16$ rows of $G$. As before, we need to connect the source vertices of the demand pairs routed in the modified sub-instances to their counterparts. In order to do so, we utilize the top $\opt'/16$ rows and the rightmost $\ell/2$ columns of $G$.
This provides sufficient space to perform the routing via canonical paths, as $\delta\ll d<\opt'/\eta$.

The final sub-case is when $\Gamma'=\Gamma$.  This case is in turn partitioned into two sub-cases. The first sub-case is when a large number of the demand pairs routed by the optimal solution belong to the set $\mset^0$: the set of all demand pairs $(s,t)$ with $|\col(s)-\col(t)|\leq 2d$.  In this case, we defined a collection of disjoint sub-grids of $G$, of size $(\Theta(d)\times \Theta(d))$, reducing the problem to a number of disjoint sub-instances that are dealt with using Theorem~\ref{thm: main w destinations far from boundary}. In each such sub-instance, we can ensure that the routing is canonical as before, since the value of the optimal solution in each sub-instance is $O(d)$, and $\delta \ll d$. In the second sub-case, we partition the grid $G$ into a number of sub-instances, where each sub-instance is defined by a pair of consecutive vertical strips of $G$. One of the two strips contains all source vertices, and the other all destination vertices of the resulting sub-instance. We then define a modified instance inside the strip containing the destinations, and route the source vertices to the corresponding sources of the routed modified demand pairs. This routing utilizes the vertical strip containing the sources, and the bottom half of the vertical strip containing the destinations. As before, it is easy to ensure that this routing is canonical.

\appendix
\label{------------------------------------Appendix-------------------------------}
\label{------------------------------------sec: proofs from overview-------------------------------}
\section{Proofs Omitted from Section~\ref{sec:alg-overview}}\label{sec: proofs from overview}

\subsection{Proof of Claim~\ref{claim: hierarchical system of squares}}

Before we define a hierarchical partition of $\tG$ into squares, we need to define a hierarchical system of intervals.
\begin{definition} Given an integer $1\leq \rho'\leq \rho$, a \emph{$\rho'$-hierarchical system of intervals} is a sequence $\hset=(\iset_1,\iset_2,\ldots,\iset_{\rho'})$ of sets of intervals, such that:

\begin{itemize}
\item for all $1\leq h\leq \rho'$, $\iset_h$ is a $d_h$-canonical family of intervals; and

\item for all $1<h\leq \rho'$, for every interval $I\in \iset_h$, there is an interval $I'\in \iset_{h-1}$, such that $I\subseteq I'$.
\end{itemize}

We let $U(\hset)=\bigcup_{I\in \iset_{\rho'}}I$, and we say that the integers in $U(\hset)$ belong to the system $\hset$.
\end{definition}

We use the following simple observation.

\begin{observation}\label{obs: canonical intervals}
 There is an efficient algorithm that constructs a collection  $\hset_1,\ldots,\hset_{2^{\rho}}$ of $2^{\rho}$ $\rho$-hierarchical systems  of intervals, such that every integer in $[\ell']$ belongs to exactly one such system.
\end{observation}

\begin{proof}
 It is enough to prove that there is an efficient algorithm, that, given an integer $1\leq \rho'\leq \rho$, constructs a collection  $\hset_1,\ldots,\hset_{2^{\rho'}}$ of $2^{\rho'}$ $\rho'$-hierarchical systems  of intervals, such that every integer in $[\ell']$ belongs to exactly one such system.
The proof is by induction on $\rho'$. The base case is when $\rho'=1$. We partition $[\ell']$ into consecutive intervals, where every interval contains exactly $d_1$ integers (recall that $\ell'$ is an integral multiple of $d_1=\eta^{\rho+2}$). Let $(I_1,I_2,\ldots,I_r)$ be the resulting sequence of intervals, where we assume that the intervals appear in the sequence in their natural order. Let $\iset_1$ be the set of all odd-indexed intervals and $\iset_1'$ the set of all even-indexed intervals in the sequence. Clearly, each of $\iset_1$ and $\iset_1'$ is a $d_1$-canonical set of intervals. We define two $1$-hierarchical systems, the first one containing only the set $\iset_1$, and the second one containing only the set $\iset_1'$. Note that every integer in $[\ell']$ belongs to exactly one resulting system.

We assume now that the statement holds for all integers between $1$ and $(\rho'-1)$, for some $\rho'> 1$, and we prove it for $\rho'$. We assume that we are given a collection of $2^{\rho'-1}$ $(\rho'-1)$-hierarchical systems of intervals, such that every integer in $[\ell']$ belongs to exactly one system. Let $\hset$ be one such $(\rho'-1)$-hierarchical system of intervals. We will construct two $\rho'$-systems, $\hset'$ and $\hset''$, such that every integer that belongs to $\hset$ will belong to exactly one of the two systems. This is enough in order to complete the proof of the observation.

Assume that $\hset=(\iset_1,\iset_2,\ldots,\iset_{\rho'-1})$. For simplicity, denote $\iset_{\rho'-1}$ by $\iset$. We now construct two new sets $\iset',\iset''$ of intervals, as follows. Start with $\iset'=\iset''=\emptyset$, and process every interval $I\in \iset$ one-by-one. Consider some interval $I\in \iset$. We partition $I$ into consecutive intervals containing exactly $d_{\rho'}$ integers each. Let $\set{I_1,\ldots,I_r}$ be the resulting partition, where we assume that the intervals are indexed in their natural order. We add to $\iset'$ all resulting odd-indexed intervals, and to $\iset''$ all resulting even-indexed intervals. Once every interval $I\in \iset$ is processed in this manner, we obtain our final sets $\iset',\iset''$ of intervals. It is immediate to verify that each set $\iset',\iset''$ is $d_{\rho'}$-canonical; that $\bigcup_{I'\in \iset'\cup \iset''}I'=\bigcup_{I\in \iset}I$; and that every integer of $\bigcup_{I\in \iset}I$ belongs to exactly one interval of $\iset'\cup \iset''$. We then set $\hset'=(\iset_1,\iset_2,\ldots,\iset_{\rho'-1},\iset')$, and $\hset''=(\iset_1,\iset_2,\ldots,\iset_{\rho'-1},\iset'')$.
\end{proof}

From Observation~\ref{obs: canonical intervals}, we can construct $2^{\rho}$ hierarchical $\rho$-systems $\hset_1,\hset_2,\ldots,\hset_{2^{\rho}}$ of intervals of $[\ell']$. For every pair $1\leq i,j\leq 2^{\rho}$ of integers, we construct a single hierarchical family $\thset_{i,j}$ of squares, such that $V(\thset_{i,j})=\set{v_{x,y}\mid x\in U(\hset_i),y\in U(\hset_j)}$. Since every integer in $[\ell']$ belongs to exactly one set $U(\hset_z)$ for $1\leq z\leq 2^{\rho}$, it is immediate to verify that every vertex of $G'$ belongs to exactly one resulting hierarchical family $\thset_{i,j}$ of squares.

We now define the construction of the system $\thset_{i,j}=(\qset^{i,j}_1,\qset^{i,j}_2,\ldots,\qset^{i,j}_{\rho})$. 
Denote $\hset_i=(I_1,\ldots,I_{\rho})$ and $\hset_j=(I'_1,\ldots,I'_{\rho})$
The construction is simple: for all $1\leq r\leq \rho$, we let $\qset^{i,j}_r=\qset(\iset_r,\iset'_r)$. From the above discussion, since sets $\iset_r,\iset'_r$ are $d_r$-canonical, so is set $\qset^{i,j}_r$. It is also easy to verify that the set of vertices contained in the squares of $\qset^{i,j}_r$ is exactly $\set{v(x,y)\mid x\in \bigcup_{I\in \iset_r}I, y\in \bigcup_{I'\in \iset_r'}I'}$, and that $\thset_{i,j}$ is indeed a hierarchical system of squares of $G$.

\subsection{Proof of Observation~\ref{obs: boosting shadow}}

Assume that $\hmset=\set{(s_1,t_1),\ldots,(s_z,t_z)}$, where the source vertices $s_1,\ldots,s_z$ appear in this left-to-right order on $R^*$. We then let $\hmset'=\set{(s_i,t_i)\mid i\equiv 1\mod 2\ceil{\beta_2/\beta_1}}$. Clearly, $|\hmset'|\geq  \floor{\frac{|\hmset|}{2\ceil{\beta_2/\beta_1}}}\geq \floor{\frac{\beta_1|\hmset|}{4\beta_2}}$.

We claim that every square $Q\in \qset$ has the $\beta_2$-shadow property with respect to $\hmset'$. Indeed, let $Q\in \qset$ be any such square, and assume that its dimensions are $(d\times d)$. Then $J_{\hmset'}(Q)\subseteq J_{\hmset}(Q)$, and $J_{\hmset}(Q)$ contained at most $\beta_2 d$ source vertices of the demand pairs of $\hmset$. From our construction of $\hmset'$, it contains at most $\ceil{\frac{\beta_2 d}{2\ceil{\beta_2/\beta_1}}}\leq \beta_1 d$ demand pairs of $\hmset'$.

\section{Proof of Claim~\ref{claim: partition the forest}}

\begin{proof}
We compute a partition $\yset(\tau)$ for every tree $\tau\in F$ separately. The partition is computed in iterations, where in the $j$th iteration we compute the set $Y_j(\tau)\subseteq V(\tau)$ of vertices, together with the corresponding collection $\pset_j(\tau)$ of paths. For the first iteration, if $\tau$ contains a single vertex $v$, then we add this vertex to $Y_1(\tau)$ and terminate the algorithm. Otherwise, for every leaf $v$ of $\tau$, let $P(v)$ be the longest directed path of $\tau$, starting at $v$, that only contains degree-1 and degree-2 vertices, and does not contain the root of $\tau$. We then add the vertices of $P(v)$ to $Y_1(\tau)$, and the path $P(v)$ to $\pset_1(\tau)$. Once we process all leaf vertices of $\tau$, the first iteration terminates. It is easy to see that all resulting vertices in $Y_1(\tau)$ induce a collection $\pset_1(\tau)$ of disjoint paths in $\tau$, and moreover if $v,v'\in Y_1(\tau)$, and there is a path from $v$ to $v'$ in $\tau$, then $v,v'$ lie on the same path in $\pset_1(\tau)$. We then delete all vertices of $Y_1(\tau)$ from $\tau$.

The subsequent iterations are executed similarly, except that the tree $\tau$ becomes smaller, since we delete all vertices that have been added to the sets $Y_j(\tau)$ from the tree.

It is now enough to show that this process terminates after $\ceil{\log n}$ iterations. In order to do so, we can describe each iteration slightly differently. Before each iteration starts, we gradually contract every edge $e$ of the current tree, such that at least one endpoint of $e$ has degree $2$ in the tree, and $e$ is not incident on the root of $\tau$. We then obtain a tree in which every inner vertex (except possibly the root) has degree at least $3$, and delete all leaves from this tree. The number of vertices remaining in the contracted tree after each such iteration therefore decreases by at least factor $2$. It is easy to see that the number of iteration in this procedure is the same as the number of iterations in our algorithm, and is bounded by $\ceil{\log n}$.
For each $1\leq j\leq \ceil{\log n}$, we then set $Y_j=\bigcup_{\tau\in F}Y_j(\tau)$.
\end{proof}


\section{Proof of Claim~\ref{claim: modified instance preserves solutions}}

Let $\pset^*$ be the optimal solution to instance $(G,\mset)$. If $|\pset^*|>d$, then we discard paths from $\pset^*$ arbitrarily, until $|\pset^*|=d$ holds. 
Recall that $G'$ is a sub-grid of $G$ spanned by some subset $\wset'$ of its columns and some subset $\rset'$ of its rows. Recall also that we have defined a sub-grid $G''\subseteq G'$, obtained from $G'$ by deleting $4d$ of its leftmost columns, $4d$ of its rightmost columns, and $|\wset'|-4d$ of its topmost rows.

Let $\mset^*\subseteq \mset'$ be the set of the demand pairs routed by $\pset^*$. For every path $P\in \pset^*$, we define a collection $\Sigma(P)$ of sub-paths of $P$, as follows. Assume that $P$ routes some demand pair $(s,t)$. Let $x_1,x_2,\ldots,x_r$ be all vertices of $\Gamma(G'')$ that appear on $P$, and assume that they appear on $P$ in this order (where we view $P$ as directed from $s$ to $t$). Denote $x_0=s$ and $x_{r+1}=t$. We then let $\Sigma(P)$ contain, for each $0\leq i\leq r$, the sub-path of $P$ from $x_i$ to $x_{i+1}$. We say that a segment $\sigma\in \Sigma(P)$ is of type $1$ if one of its endpoints is the source $s$; we say that it is of type $2$ if both its endpoints belong to $\Gamma(G'')$, and $\sigma$ is internally disjoint from $G''$; otherwise we say that it is of type $3$. We let $\Sigma_1$ contain all type-1 segments in all sets $\Sigma(P)$ for $P\in \pset^*$, and we define $\Sigma_2$ and $\Sigma_3$ similarly for all type-2 and type-3 segments.

Notice that all type-3 segments are contained in $G''$. Let $\tilde{\mset_2}$ be the set of all pairs $(u,v)$ of vertices, such that some segment $\sigma\in \Sigma_2$ connects $u$ to $v$. We also define a set $\tilde{\mset_1}$ of pairs of vertices, corresponding to the segments of $\Sigma_1$ as follows. Let $\sigma\in \Sigma_1$ be any type-$1$ segment, and assume that its endpoints are $s$ and $v$, with $s\in S(\mset')$ and $v\in \Gamma(G'')$. Let $s'\in \Gamma(G')$ be the new source vertex to which $s$ was mapped. Then we add $(s',v)$ to $\tilde{\mset_1}$. In order to complete the proof of the claim, it is now enough to prove the following observation.

\begin{observation}
There is a set $\pset$ of node-disjoint paths that routes all pairs in $\tmset_1\cup \tmset_2$ in graph $G'$, so that the paths in $\pset$ are internally disjoint from $G''$.
\end{observation}

Indeed, combining the paths in $\pset$ with the segments in $\Sigma_3$ provides a set of node-disjoint paths in graph $G'$ that routes $|\mset^*|$ demand pairs of $\mset''$ -- the demand pairs corresponding to the pairs in $\mset^*$. 

The proof of the above observation is straightforward; we only provide its sketch here. Let $J$ be the sub-path of $\Gamma(G'')$, obtained by deleting all vertices lying on the bottom boundary of $G''$ from it, excluding the two bottom corners of $G''$. For every pair $(u,v)\in \tmset_2$, we can think of the corresponding sub-path of $J$ between $u$ and $v$ as an interval $I(u,v)$. It is immediate to verify that all intervals defined by the pairs in $\tmset_2$ are nested: that is, if $(u,v),(u',v')\in \tmset_2$,  then either intervals $I(u,v)$ and $I(u',v')$ are disjoint, or one of them is contained in the other. 

Let $\iset=\set{I(u,v)\mid (u,v)\in \tmset_2}$ be the corresponding set of intervals.
Notice that no interval in $\iset$ may contain a destination vertex $v$ of any demand pair $(s,v)\in \tmset_1$.
Let $\iset^0\subseteq\iset$ be the set of all intervals containing both the top left and the top right corners of $G''$. 
Let $\iset^1\subseteq \iset\setminus \iset^0$ be the set of all intervals containing the top left corner of $G''$, and similarly, let $\iset^2\subseteq \iset\setminus \iset^0$ be the set of all intervals containing the top right corner of $G''$. 
Observe that for all $0\leq i\leq 2$, for every pair $I,I'\in \iset^i$ of intervals, one of the intervals is contained in the other.
Let $\iset'=\iset\setminus\left(\bigcup_{i=0}^2\iset^i\right )$ be the set of all remaining intervals.

We partition the intervals of $\iset'$ into levels, as follows. We say that $I(u,v)$ is a level-$1$ interval, iff no other interval of $\iset'$ is contained in it. Let $\iset_1$ denote the set of all level-$1$ intervals. Assume now that we have defined the sets $\iset_1,\ldots,\iset_{i}$ of intervals of levels $1,\ldots, i$, respectively. We say that an interval $I\in \iset'$ belongs to level $(i+1)$ iff it does not contain any interval from set $\iset'\setminus\left(\bigcup_{i'=1}^i\iset_{i'}\right)$.

We say that an interval $I(u,v)\in \iset'$ is a left interval, iff $u$ and $v$ belong to the left boundary edge of $G''$. Similarly, we say that $I(u,v)\in\iset'$ is a right interval, iff $u$ and $v$ belong to the right boundary edge of $G''$. Otherwise, we say that it is a top interval. In that case, both $u$ and $v$ must belong to the top boundary of $G''$. It is easy to verify that, if $I$ is a left or a right interval, then it must belong to levels $1,\ldots,2d$, as the height of the grid $G''$ is $4d$. Moreover, if $h$ is the largest level to which a left interval belongs, then $h+|\iset_0|+|\iset_1|\leq 4d$, and if $h'$ is the largest level to which a right interval belongs, then $h+|\iset_0|+|\iset_2|\leq 4d$.

Let $W'_1,W'_2,\ldots,W'_{4d}$ denote the $4d$ columns of $G'$ that lie to the left of $G''$, and assume that they are indexed in the right-to-left order. For each level $1\leq r\leq 4d$, for every level-$r$ interval $I(u,v)$, we route the demand pair $(u,v)$ via the column $W'_r$ in a straightforward manner: we connect $u$ and $v$ to $W'_r$ by horizontal paths, and then complete the routing using the corresponding sub-path of $W'_r$. 

The routing of the right intervals is performed similarly. The top intervals are routed similarly by exploiting the rows of $G'$ that lie above $G''$ -- recall that there are $|\wset''|+4d$ such rows, where $|\wset''|$ is the width of $G''$. We then add paths routing the pairs corresponding to the intervals in $\iset^1,\iset^2$ and $\iset^0$ in a straightforward manner (see Figure~\ref{fig: modified routing}). 

\begin{figure}[h]
\scalebox{0.5}{\includegraphics{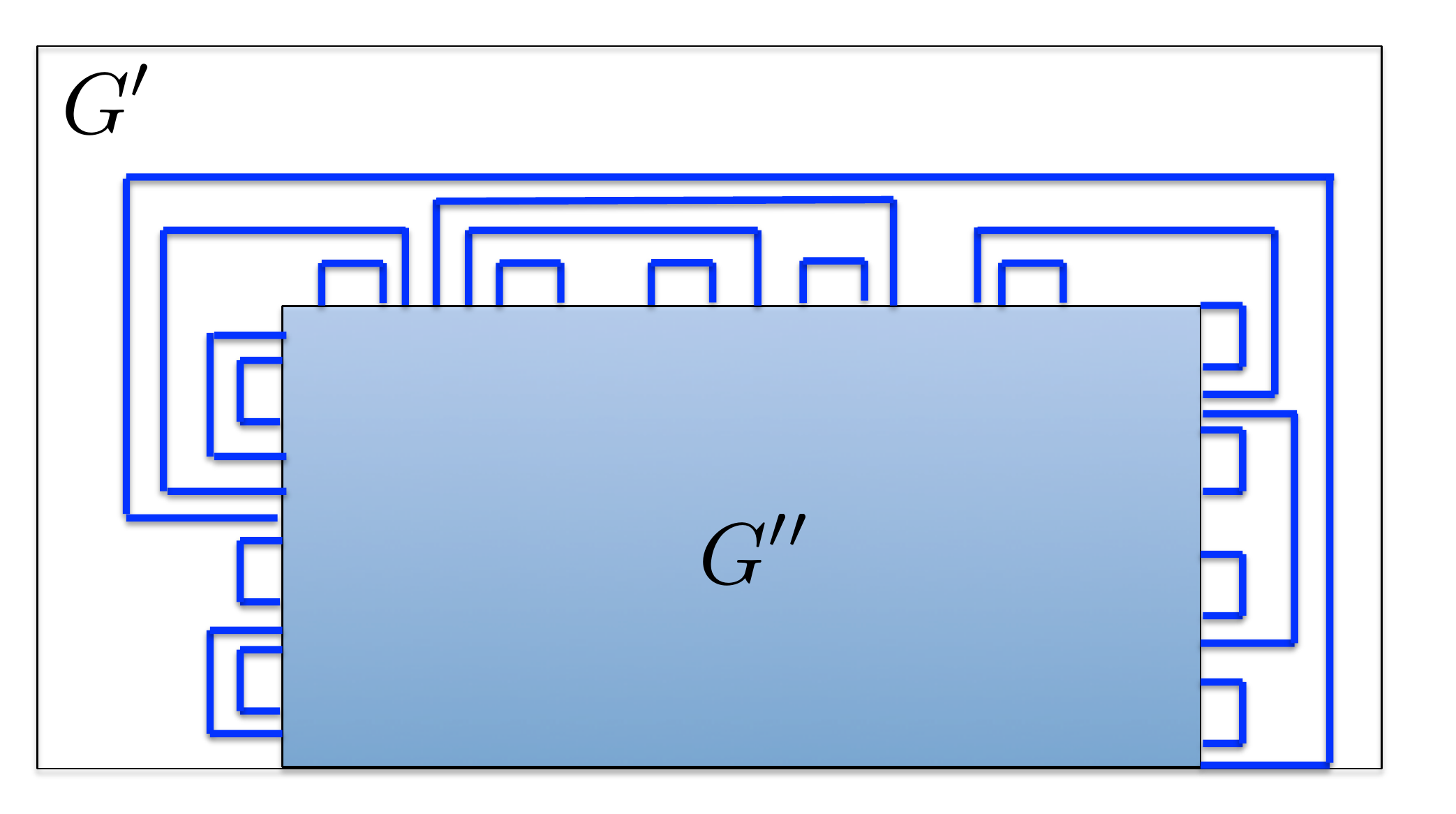}}
\caption{Routing in the modified graph.\label{fig: modified routing}}
\end{figure}

Let $H$ be the graph obtained from $G'$, after we delete all non-boundary vertices of $G''$ from it, and all vertices that participate in the routing of the pairs in $\tmset_2$ that we just defined. It is easy to verify that all demand pairs in $\tmset_1$ can be routed in $H$. In order to do so, we set up a flow network, where we start from the graph $H$, and add two special vertices $s$ and $t$ to it. We connect $s$ to every vertex in $S(\tmset_1)$, and we connect $t$ to every vertex in $T(\tmset_1)$, setting the capacity of every vertex of $H$ to be $1$. It is easy to verify that there is an $s$--$t$ flow of value $|\tmset_1|$ in this network (since every cut separating $s$ from $t$ must contain at least $|\tmset_1|$ vertices). From the integrality of flow, and due to the way in which the mapping between the vertices of $S(\mset')$ and $S(\mset'')$ was defined, we can obtain an integral routing of all demand pairs in $\tmset_1$ via node-disjoint paths in $H$.



\bibliography{NDP-grid-sources-on-top-alg}



\end{document}